\documentclass[11pt, a4paper,reqno]{amsart}
\usepackage{amsfonts}
\usepackage{amssymb}
\usepackage[dvips]{graphics}
\usepackage{euscript}
\usepackage{color}
\usepackage{enumitem}
\usepackage{fullpage}
\usepackage{mathbbol}
\usepackage{enumitem} 
\usepackage{tikz}
\usepackage{dsfont}

\usepackage{fancyvrb}
\usepackage{amsmath,amsthm}
\usepackage{enumitem}
\usepackage{thmtools}
\usepackage[style=apalike,style=numeric,maxbibnames=99,firstinits=true,citestyle=numeric,url=true,doi=true,
backend=bibtex
]
{biblatex}
\renewbibmacro{in:}{}
\AtEveryBibitem{\clearfield{number}}
\DeclareFieldFormat*{volume}{\mkbibbold{#1}}
\DeclareFieldFormat[article]{pages}{#1}

\newlist{thmlist}{enumerate}{1}
\setlist[thmlist]{label=(\roman{thmlisti}),
                  ref=\thethm.(\roman{thmlisti}),
                  noitemsep}

\declaretheorem[
    name=Theorem,
    numberwithin=section]{thm}
\declaretheorem[
    name=Lemma,
    sibling=thm]{lem}

\usepackage[capitalize]{cleveref}

\Crefname{thm}{Theorem}{Theorems}
\Crefname{lem}{Lemma}{Lemmas}
\Crefname{listthm}{Theorem}{Theorems}
\Crefname{listlem}{Lemma}{Lemmas}

\addtotheorempostheadhook[thm]{\crefalias{thmlisti}{listthm}}
\addtotheorempostheadhook[lem]{\crefalias{thmlisti}{listlem}}

 \newtheorem{lemma}[thm]{Lemma}
 \newtheorem{prop}[thm]{Proposition}
 \theoremstyle{definition}
 \newtheorem{defn}[thm]{Definition}
  
 \newtheorem{rem}[thm]{Remark}
 
 \newtheorem{assumption}[thm]{Assumption}
 \numberwithin{equation}{section}

\newcommand\beq{\begin{equation}}
\newcommand\eeq{\end{equation}}
\newcommand{\pa}[1]{\left( #1 \right)}
\newcommand{\trip}{\|\kern-.08em |}

\newcommand{\caA}{{\mathcal A}}
\newcommand{\caB}{{\mathcal B}}
\newcommand{\caC}{{\mathcal C}}
\newcommand{\caD}{{\mathcal D}}
\newcommand{\caE}{{\mathcal E}}
\newcommand{\caF}{{\mathcal F}}
\newcommand{\caG}{{\mathcal G}}
\newcommand{\caH}{{\mathcal H}}

\newcommand{\caL}{{\mathcal L}}

\newcommand{\caO}{{\mathcal O}}

\newcommand{\caQ}{{\mathcal Q}}
\newcommand{\caR}{{\mathcal R}}
\newcommand{\caS}{{\mathcal S}}
\newcommand{\caT}{{\mathcal T}}

\newcommand{\bbE}{{\mathbb E}}

\newcommand{\N}{{\mathbb N}}

\newcommand{\bbP}{{\mathbb P}}

\newcommand{\bbT}{{\mathbb T}}

\newcommand{\Z}{{\mathbb Z}}

\renewcommand{\Re}{\mathrm{Re}}

\newcommand{\tr}{\mathrm{tr}}

\newcommand{\T}{\mathbb{T}\,}

\newcommand{\abs}[1]{\left\vert #1 \right\vert}

\newcommand{\norm}[1]{\left\Vert #1 \right\Vert}

\DeclareMathOperator*{\supp}{supp}

\newcommand{\be}{\begin{equation}}
\newcommand{\ee}{\end{equation}}
\newcommand{\bea}{\begin{eqnarray}}
\newcommand{\eea}{\end{eqnarray}}
\newcommand{\beann}{\begin{eqnarray*}}
\newcommand{\eeann}{\end{eqnarray*}}

\newcommand{\e}{\mathrm{e}}

\newcommand{\diam}{\mathrm{diam}}
\newcommand{\dist}{\mathrm{dist}}
\newcommand{\set}[1]{\left\{ #1 \right\}}

\author{Wojciech De Roeck}
\address{ Instituut Theoretische Fysica, KU Leuven  \\
3001 Leuven  \\ Belgium }
\email{wojciech.deroeck@kuleuven.be}

\author{Alexander Elgart}
\address{Department of Mathematics \\ Virginia Tech \\ Blacksburg, VA 24061-0123 \\ USA}
\email{aelgart@vt.edu}

\author{Martin Fraas}
\address{Department of Mathematics \\Davis, CA 95616 \\ USA}
\email{mfraas@ucdavis.edu}
\thanks{W.D.R. was supported in part by the Fonds Wetenschappelijk Onderzoek under grant G098919N. A.E. and M.F. were  supported in part by the National Science Foundation under grant DMS-1907435.  A.E. was supported in part by the Simons Fellowship in Mathematics Grant  522404.}
\dedicatory{In memoriam: Rachel Vaiman}


\begin{document}

\title{Derivation of Kubo's formula for disordered systems at zero temperature}


\begin{abstract} 
This work justifies the linear response formula for the Hall conductance of a two-dimensional disordered system. The proof rests on controlling the dynamics associated with a random time-dependent Hamiltonian.

The principal challenge is related to the fact that spectral and dynamical localization are intrinsically unstable under perturbation, and the exact spectral flow - the tool used previously to control the dynamics in this context - does not exist. We resolve this problem by proving a local adiabatic theorem: With high probability, the physical evolution of a localized eigenstate $\psi$ associated with a random system remains close to the spectral flow for a restriction of the instantaneous Hamiltonian to a region $R$ where the bulk of $\psi$ is supported. Allowing $R$ to grow at most logarithmically in time ensures that the deviation of the physical evolution from this spectral flow is small.

To substantiate our claim on the failure of the global spectral flow in disordered systems, we prove eigenvector hybridization in a one-dimensional Anderson model at all scales.
\end{abstract}

\maketitle

\section{Introduction}
In this work we examine the response of a disordered quantum system, described by  a random self-adjoint operator $H$, to a weak time-dependent external perturbation $W(t)$, with the interaction strength modulated by the parameter $\beta$. This produces a family of self-adjoint operators
\be\label{eq:H(t)}
H(t)=H+\beta W(t),\quad t\in\mathbb R.
\ee
A typical example of such an $H$ is the Anderson Hamiltonian $H_A$ acting on $\caH=\ell^2(\Z^d)$ with $H_A:=\Delta+ V_\omega$. Here, $\Delta$ is the discrete Laplacian and $V_\omega$ is a multiplication operator, i.e., $\pa{V_\omega\psi}(x)=\omega_x\psi(x)$ for $\psi\in\caH$, where the $\omega_x$ are i.i.d. random variables with some joint probability distribution $\mu$.

This article provides a microscopic derivation of the Kubo formula for Hall conductance, a problem that arises in theoretical condensed matter physics and pertains to the dynamics generated by $H(t)$. It lies in the intersection of two broader problems in mathematical physics: microscopic justification of linear response theory and justification of quantization of Hall conductance.

\subsection{Quantum Hall effect}
\label{sec:1.1}
 In the early 1980s, von Klitzing and his collaborators \cite{KDP} made a remarkable discovery:  At low temperatures, the Hall conductance for the $2D$ electron gas in a strong magnetic field was found to be a staircase-like function of the electron density. The plateaus take values in $\mathbb{Z} \times q^2/h$ with such incredible precision (one part in a billion) that this effect is used in the metrological definitions of the kilogram and the ampere. Further experimentation revealed that the stairs vanish in very clean samples, strongly indicating that the effect requires disorder.  To comprehend the effect, the physical and mathematical theory thus has to address three fundamental questions:
\begin{enumerate}
\item Why is the Hall conductance quantized in the units of $q^2/h$?
\item What is the role of the disorder?
\item What explains the precision of this quantization?
\end{enumerate}

We first discuss Question (ii). Many aspects of Hall conductance can be encapsulated by translation-invariant magnetic Hamiltonians, characterized by  bands of the absolutely continuous spectrum separated by the spectral gaps.   The conductance in such models is quantized when the {\it Fermi energy} $E_F$ falls into the spectral gap, and transitions to a different value as $E_F$ crosses a conducting band. In what follows, we refer to this intensively-studied class of models as the {\it disorder-free case}. However, the critical feature of QHE that  cannot be explained within such a framework is the existence of plateaus, as  the electron density remains constant within the spectral gap.  An appropriate Hamiltonian modeling this aspect of the effect must instead have a spectrum consisting of interlacing intervals of conducting and insulating bands, with quantized conductance for the values of $E_F$ that lie in an insulating band. The role of disorder is precisely to create such a structure. The physics community universally accepts that a suitable $H$, namely a random magnetic Schrödinger operator, is the correct operator to describe this phenomenon. One of the long-standing open problems in mathematical physics is proving that the spectrum of $H$ consists of intervals of alternating absolutely continuous (conducting) and dense pure point (insulating) spectra. The only progress in this direction, namely the proof that the spectrum cannot be entirely pure point, has been made using the topological structure associated with the plateaus in QHE,  \cite{GKS}, which brings us back to the first question.

The mechanism explaining Question (i) above was suggested shortly after the discovery of QHE and is associated with the Kubo formula $\sigma_H$ for the Hall conductance,  which was proven to be a topological invariant. In the disorder-free case, $\sigma_H$ is linked to a Chern number of the ground state bundle whenever $E_F$ lies in the spectral gap. This is now well understood  both in the absence \cite{Thouless,Avron83} and presence \cite{Niu85,Avron85,HastingsMichalakis,Giuliani,BBDF} of interactions between the electrons. For disordered systems, $\sigma_H$ has been linked to a Fredholm index using both non-commutative geometrical \cite{BESB} and analytical \cite{ASS}  methods. The microscopic derivation of the Fredholm index for an Anderson-type Hamiltonian assuming the Kubo formula and that $E_F$ lies in the dense point spectrum was first supplied in \cite{AG}.

The theory associated with Question (iii) aims to justify the Kubo formula for conductance when the Fermi energy is in the insulating band. The Kubo formula is a standard expression for conductances, or more broadly for response coefficients, obtained by a formal first-order perturbation theory in the strength of a driving field $\beta$. To explain the precision, the theory must validate the formal calculation and demonstrate that all higher-order terms in $\beta$  vanish. In the disorder-free case, this was achieved for non-interacting \cite{ASY, ES} and interacting \cite{BDF, MT, teufel, BDFL} models. This work establishes the microscopic proof of this formula for disordered systems.

\subsection{Linear response theory}
\label{sec:1.2}
LRT explores the behavior of macroscopic variables in response to small perturbations. In the field of condensed matter physics, it serves as an essential and versatile tool with numerous variants applicable to a wide range of physical variables and models. To ground the discussion in the application we have in mind, we will discuss the response of the current $\mathbf{J}$ to an  electric field $\mathbf{E}$ with a finite voltage $\mathbf{V}$ applied across  the system in a given direction. {\it Ohm's law} states that for small $\mathbf{V}$  the current is proportional to the voltage,
$$
\mathbf{J} = \sigma \mathbf{V},
$$
where the constant of proportionality is called conductance. The purpose of LRT is to provide a microscopic expression for $\sigma$.

LRT was first developed by Kubo, \cite{K}. The expressions for $\sigma$ corresponding to nonzero and zero temperatures are known as the Green-Kubo and Kubo-St\v{r}eda formulas, accordingly, \cite{Gr, St}. In this work we consider the latter case. LRT  has a wide range of settings, \cite{MARCONI2008111}; we have chosen one guided by simplicity and convenience.

The theory computes the response from a  time-dependent Hamiltonian model of the form \eqref{eq:H(t)}.  In the context of electrical conductance, $W(t) = e^t V(x)$, where $V(x)$ is an electric potential of unit voltage. At $t=-\infty$, the system is initiated in an equilibrium state $\rho$ of the unperturbed Hamiltonian $H$ and then evolves according to the Heisenberg equation
\be\label{eq:adevoldensity}
\dot{\rho}_t = -i [H(t), \rho_t], \quad H(t) = H + \beta e^{\epsilon t} V(x)
\ee
with the adiabatic parameter $\epsilon$. The expected value of the measured current at $t=0$ is $\mathbf{J} = \tr(\rho_0 J)$, where $J$ is the current operator, and the measured conductance is
$$
\sigma_m(\epsilon, \beta) = \beta^{-1} \tr(\rho_0 J).
$$
In a typical experiment that measures conductance, the time scales involved are  such that both  $\epsilon$ and $\beta$ are small parameters. However, $\epsilon$  is significantly smaller than $\beta$ by several orders of magnitude. For a standard experimental setup $\epsilon/\beta < 10^{-9}$ (based on experimental time longer than $1$ milisecond and electric potential greater in magnitude than $10^{-3}V$; \cite{Kampen} estimates that linear approximation in $\beta$ would be justified only for electric fields of order $10^{-16}\frac{V}{m}$). This relationship between timescales ensures that the system will produce a {\it non-trivial} steady current.
On the other hand, the Kubo formula $\sigma_H$ for conductance  is obtained by taking the limit $\beta << \epsilon$,

 \be
\label{eq:lin_resLRT}
\sigma_H = \lim_{\epsilon \to 0}\lim_{\beta \to 0} \sigma_m(\beta, \epsilon) = \lim_{\epsilon \to 0} i \int_{-\infty}^0 e^{\epsilon t} \tr(\rho [e^{iHt} J e^{-iHt},V]) dt,
\ee
and only depends on the spectral data for the unperturbed Hamiltonian $H$. Nevertheless, the formula is spectacularly successful in matching available experimental data. This raises  the question of how the Kubo formula not only works at all in this context but also predicts the experimentally observed conductance with astonishing precision. The {\it problem of linear response} is to either prove that the joint limit
\[
 \lim_{\epsilon << \beta \to 0} \sigma_m(\beta, \epsilon)
\]
 exists and is equal to $\sigma_H$, or to provide an alternative explanation for the validity of expression \eqref{eq:lin_resLRT}.

Although the focus of our attention is on the response of the current to the electric field, i.e., Ohm's law, the same question can be posed for Fourier's law, Fick's law, and other phenomena. The justifications of the Kubo formula for these various physics laws are long-standing open problems in mathematical physics, each posing a unique mathematical challenge, see, e.g.,    \cite[Problem 4B]{Simonfifteen}. Our work provides the first proof of the Kubo formula in a disordered system.

 \subsection{Microscopic derivation of the Kubo formula for Hall conductance}
 \label{sec:1.3}

The Hall conductance $\sigma_H$ is defined in $2$D as the proportionality constant between the applied potential difference and the current flowing in the perpendicular direction.
In what follows, we make a specific choice for the applied electric  potential $V(x)$ and the current operator $J$.  We will assume that the Fermi energy $E_F$ lies  in  the mobility gap for $H$, where the latter concept will be formally defined in Section \ref{subsec:hyb}.

We denote by $(x_1, x_2)$ the coordinates of points in $\mathbb{Z}^2$ and by  $\Lambda_n$ the characteristic function of the subset $ \{x_n \geq 0\}$, $n=1,2$. These functions are examples of so-called {\it switches}, i.e., functions $h$ of one variable that are real valued, monotone, and non-decreasing, with $h(-\infty)=0$ and $h(\infty)=1$.

We consider an electric  potential  $V = \Lambda_2$, which has a unit voltage drop across the $x_2$ direction.  The (Hall) current flowing in the perpendicular direction across the fiducial line $x_1 = 0$  corresponds to the operator $J = i [H, \Lambda_1]$. The equilibrium state $\rho$ is given by the Fermi projection $P_F:=\chi_{< E_F}(H)$.  The Kubo-St\v{r}eda formula (\ref{eq:lin_resLRT}) is then given by 
\be\label{eqdef:sigma}
\sigma_H = \tr(P_F [[P_F,\Lambda_1],[P_F,\Lambda_2]]),
\ee
see e.g. \cite{AG}.
We make two changes to the linear response setup explained above. We replace $e^t$ by a compactly supported switch  $g$, and average the current over a time window of order $\epsilon^{-1}$. More specifically, we consider a Hamiltonian of a form
\[
H(t) = H + \beta g(\epsilon t) \Lambda_2,
\]
where the function $g$ satisfies
\begin{enumerate}
\item $g\in C^{\infty}[-1,1]$;
\item $g(s)=0$ for $s\leq s_0$ for some $s_0>-1$;
\item $g(s)=1$ for $s\geq 0$.
\end{enumerate}
We (re)define the measured conductance as
\be
\label{eq:measured_response}
\sigma_m(\beta,\epsilon) := \beta^{-1} \epsilon \int_0^{1/\epsilon} \tr\left(J (\rho_t - \rho)\right) dt.
\ee
There are no equilibrium currents \cite{BFr}, i.e., $\tr(J \rho) = 0$ when the trace is properly defined. However, in infinite volume  $J \rho_t$ is not a trace class operator and subtracting $J \rho$ is a physically correct way to regularize it. We stress again that our goal is to understand the behavior of $\sigma_m(\beta,\epsilon)$ for $\epsilon << \beta \to 0$.  

Our main result on the problem of linear response establishes the existence of the joint limit  under the constraint $\epsilon = e^{ -\beta^{-p}}$ with the positive exponent $p$. 
 
\begin{thm}\label{thm:QHE}
Suppose that $H$ satisfies Assumptions \ref{assump:FRC}--\ref{assump:FMC} below with $E_F$ lying in the interior of a mobility gap. Then there exist $p>0$ such that  
\[
\mathbb E\abs{ \sigma_H - \sigma_m }\leq e^{-\beta^{-p/2}},
\]
 provided $ \epsilon = e^{-\beta^{-p}}$.
\end{thm}

\begin{rem}\label{rem:Thm1}
\begin{enumerate}
\item[]
\item The use of a compactly supported switch  function $g(t)$ instead of the exponential is a natural choice from a mathematical point of view. That being said, Theorem \ref{thm:QHE} could also be established for $g(t)=e^{t}$.
\item Some form of the current averaging is likely needed for the result to hold. We did not try to minimize the size of the time window over which the average is performed. 
\item The choice of profiles for switches $\Lambda_i$ and $g$ does not affect the result. This is related to the fact that the expression for $\sigma_H$ is universal in the sense that the value of $\sigma_H$ (almost surely) does not change upon modifying the switches or changing $E_F$ within the same interval $J_{loc}$, see \cite{EGS}.
\item One can also study conductivity instead of conductance, where the switch functions $\Lambda_i$ are replaced by the linear relations $X_i(x)=x_i$ and the trace in \eqref{eqdef:sigma}--\eqref{eq:measured_response} is replaced by the trace per unit volume. While working with conductivity simplifies some of the analysis (e.g., one no longer needs to regularize $ \tr(J \rho_t)$ in \eqref{eq:measured_response}), it also offers different technical challenges (e.g., the corresponding Hamiltonian $H(t)$ is no longer bounded and even if $H$ has spectral gaps, they close for $H(t)$). In particular, even the justification of the Kubo formula for conductivity when the limit $\beta\to0$ is taken first requires non-trivial effort for disordered systems, \cite{BGKS}. We refer the reader to \cite{MT, HT} for state-of-the-art articles on the conductivity approach in the disorder-free case. It would be interesting to see whether the techniques developed in our work can also be extended to handle this choice.
\item Using Theorem \ref{thm:QHE}, we can bound the finite temperature corrections to $\sigma_m$ by $\frac1\epsilon e^{- d_\mu/T}$, where  $T$ is the absolute temperature, $\mu$ is the chemical potential, and $d_\mu$ is a distance from $\mu$ to the boundary of the insulating band. Let us mention that the finite temperature correction has been recently  addressed for the gapped systems in the many-body context  \cite{GLMP}.
\end{enumerate}
\end{rem}

The majority of the mathematical work related to the Kubo formula in disordered systems, with or without an application to QHE, falls into two categories: In the first one, the Kubo formula is taken for granted  (or at least the order of limits $\beta << \epsilon$ is assumed) and its various consequences in different settings, such as the mathematical proof of Mott's formula, \cite{KLM}, are studied. The second category aims to justify the Kubo formula itself with the correct order of limits. Since our work lies firmly in the second category, we primarily focus our attention on past works in this direction.  For a recent review of efforts pertaining to both categories, we refer the reader to \cite{HT}.

The Kubo formula has been validated in systems with a {\it spectral gap} ($\dist\pa{\sigma(H),E_F}>0$), under various sets of assumptions on $H$ and the underlying geometry, \cite{ASY,ES,BDF,MT,teufel,BDFL}. In this scenario, the weak field $\beta\to0$ and adiabatic $\epsilon\to0$  limits commute. On the technical level, this can be linked with the stability of the spectral gaps under small perturbations (i.e., $\dist\pa{\sigma(H (t)),E_F}>0$ holds), ensuring that the adiabatic theorem of quantum mechanics could be used. The latter implies that $\rho_t$ is the zero temperature equilibrium state of  $H(t)$ up to uniformly small corrections of order $\epsilon$.  However, in the disordered case, there is no spectral gap to begin with, and the pure point spectrum associated with a mobility gap is unstable under small perturbations, \cite{DMS}. Consequently, in this scenario, the limits are not expected to coincide on physical grounds \cite{Kampen}. In this sense, the result presented above with the joint limit is optimal.

The prior mathematical results in this direction for disordered systems are scarce. As previously mentioned, \cite{BGKS} established the existence of the limit $\beta \to 0$ at fixed $\epsilon$. For $\epsilon \to 0$, the only available result, namely the absence of transport, $\sigma_m=o(1)$, was proven in the case $\beta=\epsilon$ in \cite{NK} under the assumption of {\it complete localization} (i.e., there are no conducting bands).  Under this assumption, the dynamics of the perturbed system can be  controlled for long timescales using the one associated with the unperturbed operator $H$, e.g.,  \cite{SW,BW,NK,DH,Abanin2016}.  Beyond this, despite general interest in the mathematical physics community from the moment that the problem was identified in  \cite{BESB, AG}, it  remained completely open, \cite{HT}.

We have had to develop new concepts in order to handle conducting bands and explore the regime $\epsilon << \beta$. In particular, our proof rests on the construction of the {\it local gap structure} for disordered systems, which is more robust than the standard description of the localization and, in particular, survives the time-dependent perturbations described by \eqref{eq:H(t)}. This is the content of Theorem \ref{thm:hypa} below. We then build an adiabatic theory associated with this structure for the dynamics of $H(t)$, characterized by local rather than global adiabatic behavior. We believe that this new result (Theorem \ref{thm:AT-single'} below), which we will refer to as the {\it local adiabatic theorem}, is of independent interest. The derivation of the Kubo formula then follows via more standard (albeit technically involved) methods.


The rest of the paper is organized as follows:  We formulate our core technical result, the local adiabatic theorem, Theorem \ref{thm:AT-single'}, in Section  \ref{sec:int}. This result relies on the dynamical properties associated with the local gap structure for the time-dependent Hamiltonian $H(s)$, presented in Section \ref{sec:loc_gap_st}. The origin of this structure can be traced back to the time-independent random system $H$ on a torus, which is studied in Section \ref{sec:Loc}. We then study the local adiabatic behavior of disordered systems in Section \ref{sec:patch} and complete the proof of Theorem \ref{thm:AT-single'} in Section \ref{sec:glob}. This theory is used to prove our principal result on the Kubo formula, Theorem \ref{thm:QHE}, in Section  \ref{sec:LRT}. Appendices  \ref{sec:hyb}--\ref{sec:Wannier} contain results of independent interest, namely hybridization delocalization in dimension one and the construction of a Wannier-type basis for disordered systems, respectively. Various auxiliary results are included in Appendix \ref{ap:aux}.

\section{Local adiabatic theorem}\label{sec:int}

In this section, we unveil our core technical result - the local adiabatic theorem, specifically designed to work with disordered systems. Our starting point here is a brief discussion of the localization phenomenon.

\subsection{Localization and delocalization for time-dependent systems}
The presence of disorder in quantum mechanical systems leads to the phenomenon of localization. {\it Spectral} localization manifests in the emergence of energy interval(s) $J_{loc}\subset\mathbb R$ such that, for almost all random configurations $\omega$, $\sigma(H)\cap J_{loc}$ is pure point. Moreover, the eigenvectors of $H$ in $J_{loc}$ are (spatially) exponentially localized in the sense of  \eqref{eq:globstr} below. 

Spectral localization is not stable under perturbation: The rank one perturbation family $H_A(\beta)$ of the form $H_A(\beta)=H_A+\beta\chi_{\set{0}}$ exhibits almost sure singular continuous spectrum for a $G_\delta$-dense set of $\beta$'s, \cite{DMS,Gor}. Although there are no rigorous results beyond rank 1 perturbation, one should not expect much uniformity of the localization properties as a function of $t$ or $\beta$ of the Hamiltonian \eqref{eq:H(t)}, provided that $W$ is sufficiently non-trivial. 

\subsubsection{Dynamical localization}
 Dynamical localization is concerned with the non-spreading of wave packets during time evolution. It is expressed as the (uniform in time) exponential decay of the matrix elements of $e^{-itH}P_{J_{loc}}$, the unitary semigroup generated by $H$ and restricted to the energy interval $J_{loc}$ (here, $P_{J_{loc}}$ denotes the spectral projection of $H$ onto $J_{loc}$). The concept is still well-defined for a time-dependent Hamiltonian $H(t)$, and a natural question  is whether it is still dynamically localized for at least small perturbations $\beta\ll1$. 
 
The properties of the system \eqref{eq:H(t)} have been studied before under various assumptions. In physics literature, one of the earliest  works in this direction  goes back to \cite{Wilkinson1988}, which analyzes the behavior of a  random matrix model. On a mathematical footing, compact (in space) perturbations $W$ have been studied in the time-periodic \cite{SW} and the time-quasi-periodic \cite{BW} settings. The case of spatially extensive periodic systems with few frequencies  was  considered in \cite{DH}. In the $\beta=\epsilon$ adiabatic setting,  it was considered in \cite{NK}.  For time periodic systems, one can also consider the spectral localization of the  associated  Floquet operator, \cite{SW,DH,Abanin2016}. On a  heuristic level \cite[Section 1]{DH}, there should be a transition from a localized regime to a non-localized regime when  $\nu\sim \beta \exp\pa{-c_d\beta^{-p_d}}$, where $\nu$ is the Floquet frequency\footnote{The parameter $\nu$ in \cite{DH} plays the same role as $\epsilon$ in our setting.}  for $W$ and  $c_d,p_d$ are dimension-dependent parameters. For $\nu \gg \beta \exp\pa{-c_d\beta^{-p_d}}$ only a small fraction of Floquet eigenstates delocalizes.  Apart from constraints on $\beta, \epsilon$, in all these works, the analysis heavily depends on the assumption of strong disorder, under which the interval $J_{loc}$ can be replaced by the whole $\mathbb R$. 

The instability of spectral and dynamical localization is due to the phenomena of resonant hybridization that we will describe next.

\subsubsection{Localized systems and resonant hybridization}\label{subsec:hyb}
 We say that an open interval $J_{loc}\subset \sigma(H)$ is a  {\it mobility gap} or a region of exponential localization if the spectrum of $H$ in $J_{loc}$ is of pure point type and there exist constants $0<C,c,m <\infty$, such that for each eigenpair $(E_i,\psi_i), E_i \in J_{loc}$ one can find $x_i\in\Z^d$, called a {\it localization center} for $\psi_i$, satisfying 
\be\label{eq:globstr} \abs{\psi_i(x)}\le C\langle{x}\rangle^{d+1}e^{-c\abs{x-x_i}},\ee
where $\langle{x}\rangle:=\sqrt{|x|^2+1}$. 
The prototypical example of such an $H$ is the Anderson model $H_A$ described earlier. The Anderson Hamiltonian is known to display exponential localization in the vicinity of spectral edges, at large values of disorder (for a sufficiently regular distribution $\mu$) and in dimension $d=1$, for almost all configurations $\omega$.  We will not attempt to cite the extensive literature of history, reviews, results and open problems concerning this model and its variants. We will instead refer the interested reader to a recent monograph \cite{AW} on the subject.  

The  instability of such uniform localization properties with respect  to perturbations can be linked to a mechanism known as \emph{resonant hybridization}, see, e.g., \cite[Chapter 15]{AW}. This  concept can be illustrated by considering a two-level system with a Hamiltonian $H(s)$ of the form
\[H(s)=\begin{pmatrix} g & s\\ s & -g\end{pmatrix} ,\quad s\in(-1,1), \quad g\ll1.\]
When $s=0$, the canonical basis $e_1, e_2$ is an eigenbasis for $H(s)$. These remain approximate eigenvectors for $H(s)$ provided that $\abs{s}\ll g$. However, the picture is different for the case where the relation between the energy gap $2g$ and the tunneling amplitude $\abs{s}$ is reversed: When $g\ll \abs{s}$,  an approximate eigenbasis is given by $\set{e_1\pm e_2}$. I.e., the eigenfunctions are no longer localized in the basis $\set{e_i}$ and instead are given by hybridized functions which are combinations of these vectors.

If we consider the spectral flow of eigenvectors as a function of $s$, then we see that this flow will transition between $e_1$ and $e_2$ in a time of approximate length $g$. As we show in Appendix  \ref{sec:hyb}, this behavior also occurs in the extended disordered system. The hybridization implies that the spectral flow is very nonlocal, as disordered analogues of $e_{1,2}$ can be localized arbitrarily far away from each other.  

More precisely, if we consider a finite volume restriction of $H$, say to a box with side length $\caL$,  we can then label the eigenstates $\psi_{i,s}$ so that for each $i$, $s\mapsto \psi_{i,s}$ is continuous,  \cite{Kato}. However, we do expect the modulus of continuity to diverge badly as $\caL\to\infty$.
 
We are not aware of any prior rigorous results making the two-level heuristics exact for  $\Z^d$ systems for any $d$ (however, see \cite[Chapter 15]{AW} for the results on regular trees). In Appendix \ref{sec:hyb}, we show the emergence of hybridization rigorously for a one-dimensional system. Specifically, we prove Theorem \ref{thm:A-main}, which informally can be expressed as 
\begin{thm}\label{thm:full_s'}
Let  $H$  be the standard Anderson model in $1d$. Then, under some additional regularity assumptions on the random potential and mild assumptions on $W$, the eigenfunction hybridization occurs on all scales with scale-independent probability. The corresponding eigenvalues exhibit avoided level crossings.
\end{thm}

\subsection{Adiabatic theory}
The   Schr\"odinger dynamics  associated with $H(t)$ in \eqref{eq:H(t)} are given by the linear initial value problem (IVP): 
\be\label{eq:clIVP}
i \dot \psi(t)   = H(t)\psi(t), \quad \psi(0) = \psi_o, 
\ee
where  $\psi_o$ is a normalized vector on $\mathcal H$ (the initial wave packet of the system). The solution of the IVP becomes trivial in the case of time-independent operators $H(t) = H_o$ and the initial state $\psi_o$ being an eigenvector for $H_o$. In this case, the evolution $\psi(t)$ coincides with $\psi_o$ up to an acquired phase. 

 A more interesting and physically realistic situation arises when the dependence on time in $H(t)$ is present but is adiabatic. In this case, the evolution $\psi(t)$ is expected to follow the spectral evolution of the Hamiltonian $H(t)$ (the assertion known as the {\it adiabatic theorem of quantum mechanics}). Of course, slow is a relative concept, and we need to quantify the reference time scale for these purposes. In the standard adiabatic theorem, such a parameter is given by the spectral gap in $H(t)$ (note that energy has units time$^{-1}$ in \eqref{eq:clIVP}). To make this statement more quantitative, it is convenient to consider the family  $H(\epsilon t)$, where $\epsilon$ is a small (adiabatic) parameter, and the physical time $t$ runs over the long interval $[0,1/\epsilon]$.  After a change of variables $s=\epsilon t$ where $s$ is a rescaled time, the relevant IVP becomes
\be\label{eq:adevol}
i \epsilon \dot \psi_\epsilon(s)   = H(s)\psi_\epsilon(s), \quad \psi_\epsilon(0) = \psi_o,\quad s\in[0,1] . 
\ee
We denote by $U_\epsilon(s)$ the corresponding propagator, i.e. the unitary operator that solves the IVP 
\be\label{eq:Ueps}
i \epsilon \partial_s U_\epsilon(s)   = H(s)U_\epsilon(s), \quad U_\epsilon(0) = \mathds{1}.
\ee
Let us assume that the spectrum $\sigma(H(s))$ of the operator $H(s)$ contains a set $\mathcal S(s)$  isolated from the rest of the spectrum by a uniform distance $g$ (the spectral gap). Denoting by $P(s)$ the spectral projection of $H(s)$ onto $\mathcal S(s)$, and assuming that $P(0)\psi_o=\psi_o$, the (qualitative) adiabatic theorem states that 
\be\label{eq:qualat}
\lim_{\epsilon\to0}\norm{\psi_\epsilon(s)-P(s)\psi_\epsilon(s)}=0,\ee
provided $H(s)$ is smooth.
A stronger statement holds, namely 
\be\label{eq:qualproj}
\lim_{\epsilon\to0}\norm{U_\epsilon(s)P(0)U^*_\epsilon(s)- P(s)} =0,
\ee
and one can make the error estimate for the norm above explicit in terms of its $\epsilon$ and $g$ dependencies, see e.g., Lemma \ref{lemma:all_orders} below.

As mentioned above, we can  label the eigenstates $\psi_{i,s}$ of a finite system in such a way that the spectral flow $s\mapsto \psi_{i,s}$ is continuous for each $i$. Suppose there are no degeneracies, which is the generic case. Then each eigenvalue is gapped, and the adiabatic theorem says that in the limit $\epsilon \to 0$, the solution of (\ref{eq:adevol}) is the spectral flow. Combined with Theorem~\ref{thm:full_s'}  this implies that dynamical localization fails for $\epsilon \to 0$ as the spectral flow is extremely nonlocal. However, for $\epsilon >0$, the physical evolution cannot be arbitrarily nonlocal.
We believe that the way that this dilemma is resolved is that the physical evolution of an initial eigenvector, for most values of  $s$, stays close to one of the global eigenvectors $\psi_{i,s}$, even though the index $i$ varies wildly with $s$.   A simpler take on this is that the evolution of the initial eigenvector stays for all times $s$ close to an instantaneous eigenvector $\phi_s$ of the restriction of $H(s)$  to a local box around the support of the initial eigenvector $\psi_{i,0}$.   We will refer to this statement as a {\it local adiabatic theorem}, and state it quantitatively as Theorem \ref{thm:AT-single'} below. One can interpret this result as meta-stability of $\phi_s$ with a very long lifetime.

The adiabatic theorem and its derivatives play an fundamental role in the various branches of quantum and statistical mechanics. The first results on adiabatic behavior go back to the dawn of quantum mechanics and are due to Born and Fock in 1928, \cite{BF}. The modern adiabatic theory was initiated by Kato in 1950, \cite{Ka}, and has since been studied intensively in the mathematical physics literature. The adiabatic theorem has been extended to a situation where the family $P(s)$ is smooth, but no gap is present, \cite{B,AE}. This situation usually occurs for a ground state in the threshold of the continuous spectrum. The other possible scenario occurs in rank one perturbed completely localized system, where one can show that the Fermi projection $P_F(t)$ is a continuous function for a set of the full Lebesgue measure, even when $\sigma(H(t))$ is not pure point, \cite{AHS}. In space-adiabatic perturbation theory \cite{Panati2003}, the gap is closed by a locally small but globally large perturbation (for related work in field theory, see \cite{Tenuta2008}). More recently, the adiabatic theorem was established for certain systems with a spectral gap but non-smooth $P(s)$, \cite{BDF,MT}. This situation arises in the context of the thermodynamic limit for many-body systems.  

For the disordered systems that are not entirely localized, it is necessary to consider a scenario where {\it both} conditions fail to hold.

\subsection{Local adiabatic theorem}

To properly formulate this assertion, we must first establish the appropriate framework.

An operator $K$ acting on $\ell^2\pa{\Z^d}$ is $r$-local for some $r\in\mathbb N$ if 
\[K(x,y):=\langle \delta_x,K\delta_y\rangle=0 \mbox{ provided } \abs{x-y}>r,\quad x,y\in\Z^d,\]
where $\abs{x-y}$ stands for the $\ell^\infty$ distance in $\Z^d$.
\begin{assumption}\label{hyp1}
The operators $H(s)$ are uniformly bounded, smooth, $r$-local,  self-adjoint operators acting on $\ell^2\pa{\Z^d}$, of the form \eqref{eq:H(t)} that satisfy $\| H(s)\| \leq C$.  In addition, for all $k \in \mathbb{N}_0$, $W^{(k+1)}(0) = W^{(k+1)}(1)= 0$, and there exists a constant $C_k$ such that
$
\|W^{(k)}(s)\| \leq C_k
$.
\end{assumption}

For any  $\Theta \subset\Z^d$, we denote by $H^\Theta$  the canonical restriction $\chi_\Theta H \chi_\Theta$ of $H$ to $\ell^2(\Theta)$.  
\begin{assumption}[Finite range of disorder correlations]\label{assump:FRC}
For any pair of subsets $\Theta,\Phi$ of $\Z^d$ that satisfy $\dist\pa{\Theta,\Phi}>r$,  the operators $H^\Theta$ and $H^\Phi$ are statistically independent.
\end{assumption}
 For any region $\Theta\subset\Z^d$ and $x,y\in\Theta$, we define
\be
\abs{x-y}_\Theta=\min\pa{\abs{x - y}, \pa{\dist(x, \partial_{1}\Theta) + \dist(y, \partial_{1}\Theta)}} ,
\ee 
with the interior boundary $\partial_{1}\Theta= \{ x \in \Theta, \dist(x,\Theta^c)=1\}$. This distance function regards $\partial_{1}\Theta$ as a single point.  It permits us to work with systems that exhibit localization in the bulk without ruling out absence of  delocalized edge modes.  With this preparation, our assumption of Anderson localization in an interval $J_{loc}$ for $H$ reads
\begin{assumption}[Fractional moment condition  on $J_{loc}$]\label{assump:FMC}
There exist $q\in(0,1)$ and $C_q,c>0$ such that, for any subset $\Theta$ of $\Z^d$, for any $E\in J_{loc}$, and any $\eta\neq0$,  we have
\be\label{eq:FMC}
\sup_{E\in J_{loc}}\mathbb E\pa{ \abs{(H^{\Theta}-E-i{\eta})^{-1}(x,y)}^q }\le C_q e^{-c \abs{x-y}_\Theta}\mbox { for all } x,y\in \Theta,
\ee
where $\bbE\pa{\cdot}$ stands for expectations with respect to $\omega$. 
\end{assumption}
For some of our results we will also need
\begin{assumption}[Finite spectral multiplicity]\label{assump:multiplicities}
There exists $m\in\N$ such that, for any  $\Theta\subset\Z^d$, the multiplicity of eigenvalues of   $H^\Theta$ does not exceed $m$ almost surely.
\end{assumption}
\begin{rem}
For the standard Anderson model with absolutely continuous random distributions   $m=1$, \cite{S1}. This type of result can be extended to a larger class of discrete models, see, e.g., \cite[Theorem 5.8]{AW} and  \cite{DE}. While the simplicity of the spectrum is, in general, not known to hold for models that satisfy Assumptions \ref{assump:FRC}--\ref{assump:FMC}, in practice, a majority of them are generated using finite-rank operators for which Assumption \ref{assump:multiplicities} does hold, \cite{HK}. 
\end{rem}
\begin{rem}
Surprisingly, the basic localization property \eqref{eq:globstr} has only been proven in existing literature under the assumption of spectrum simplicity (i.e., $m=1$ in Assumption \ref{assump:multiplicities} above), cf. \cite[Theorem 7.4]{AW}. In order to avoid this rather restrictive condition, we obtain its analogue for a more general case of finite $m$ in Appendix \ref{sec:Wannier} below. The argument there relies on the construction of the so-called generalized Wannier basis for an eigenprojection of the localized Hamiltonian, consisting of exponentially localized functions.
\end{rem}

The local adiabatic theorem is easier stated in finite volume for a bulk system, we introduce a periodized restriction of $H(s)$ to a discrete torus $\mathbb T= \mathbb T^d_M$, which we  associate with the hypercube $[1,M]^d$ whose opposite faces are identified. This restriction is defined as 
\be\label{eq:periodH}H^{\mathbb T}(x,y)=\frac12\sum_{n\in M \Z^d}H(x,y+n)+H(x+n,y),\quad x,y\in \mathbb T.\ee

Our two main parameters are the adiabaticity parameter $\epsilon$  and the driving strength $\beta$, introduced earlier in \eqref{eq:adevol} and  \eqref{eq:H(t)}, respectively. In our results we will use four exponents,
\be\label{eq:relations1}
\xi = \tfrac{d}{q},\quad \xi'=d +\tfrac{1}{2} + \xi, \quad p_1 >  d+\xi', \quad p_2 > \max \left(\xi', 2 \xi \right), 
\ee
with fixed $p_1, p_2$ satisfying the last two inequalities. Throughout this paper, we will assume that $\beta\ll1$ and $\epsilon\ll1$ satisfy 
\be\label{eq:relations'}
e^{-\beta^{-1/(2p_1)}} < \epsilon < \beta^{p_2 p_1}.
\ee
It will be convenient to work with a (generally flexible) scale parameter $\ell \in\N$ satisfying 
\be\label{eq:scales'}
\ell^{-p_2} \geq \epsilon \geq e^{-c \sqrt{\ell}},\qquad  \beta \leq \ell^{-p_1},
\ee 
whose existence is guaranteed by \eqref{eq:relations'}.

We will use generic, $M,\epsilon,\beta,\ell$-independent constants $C, c$  whose values can change from line to line.  They will, however, in general depend on the other parameters and constants introduced above (such as the range $r$ and the probability distribution $\mu$, as well as on the constants $C_q,C_k$, etc.).
We  allow for the system size $M$ to be  arbitrarily large, and all of our estimates will be uniform in $M$. We will use $C$ to indicate that the constant should be sufficiently large for a bound to hold, and $c$ to indicate that the constant should be sufficiently small.

The following then is the {\it  local adiabatic theorem}. It is based on the emergence of a local gap structure for the spectral data associated with a torus, once partitioned into smaller boxes of linear size $\ell$. To make its presentation more accessible, we will use an extra assumption on the integrated density of states  $\mathcal N_{J_{loc}}$ (see \eqref{eq:IDOS} below) in addition to our standard hypotheses on the model. 
\begin{thm}[Local adiabatic theorem]\label{thm:AT-single'}
Suppose that Assumptions \ref{hyp1}--\ref{assump:multiplicities} hold for $H(0)$ and the integrated density of states $\mathcal N_{J_{loc}}$  is a.s. positive.
Let $\beta,\epsilon,\ell$ satisfy \eqref{eq:relations'}--\eqref{eq:scales'} and $J_{loc}'$ be any closed interval contained in $J_{loc}$. Assuming that $\ell$ is large enough, with probability at least $1 - e^{- c \sqrt{\ell}}$, the following holds for a fraction of at least $1 - e^{- c \sqrt{\ell}}$ of eigenstates $\psi$ of $H^{\T}$ with eigenvalues $E \in J_{loc}'$:  There is a region $R\subset \bbT$ with $\diam(R) \leq c \ell^{3/2} $ such  that
\begin{enumerate}
\item For all $s \in [0.1]$, $H^R(s)$ possesses the spectral patch $S(s) \subset \sigma(H^R(s)) $ which is isolated  from the rest of the spectrum $\sigma(H^R(s))$. We denote the associated spectral projector by $P(s)$.
\item The solution $\psi_\epsilon(s)$  of \eqref{eq:adevol} with $\psi_\epsilon(0)=\psi$ satisfies 
\be\label{eq:AT_sing'}
\max_{s\in[0,1]}\norm{(1- P(s))\psi_\epsilon(s)} \le C \pa{ \epsilon   \ell^{\xi'}  + e^{- c \sqrt{\ell}} }.
\ee
This bound can be improved for $s=1$: For any $N\in \N$,
\be\label{eq:AT_sing1'}
\norm{(1- P(1))\psi_\epsilon(1)} \le C_N \pa{\epsilon^N\pa{ \ell^{N{\xi'}}  +  \ell^{(2N+1)\xi}}   +   e^{- c \sqrt{\ell}}}.
\ee
\end{enumerate}
\end{thm}
This statement will be proved in Section \ref{sec:glob}.  
\begin{rem}
While the assertion is formulated for tori of the arbitrary size $M$, in applications (e.g., in the proof of our main result, Theorem \ref{thm:QHE}), we often have $M\ll e^{- c \sqrt{\ell}}$. In this case, the statement holds for {\it all} eigenstates rather than their fraction, with the same probability.
\end{rem}
\begin{rem}
Let us note that both the upper and lower bounds on $\epsilon$ in \eqref{eq:scales'} have to do with the faithfulness of our approximation of the actual eigenstate for $H^{\T}$ by the local spectral patch for $H^R$. If $R$ is too small, then there is no reason for its eigenvectors (even the bulk ones) to be close to the eigenvectors of $H^{\T}$ (so the spatial faithfulness of our approximation is destroyed). On the other hand, if $R$ is too big, the gaps in the spectrum of $H^{\T}$ become smaller than the size $\beta$ of the perturbation, allowing for transition between eigenstates that are energetically far apart from one another (so the energetic faithfulness of our approximation is destroyed). In particular, one can think of these constraints as a consequence of the uncertainty principle for disordered systems.
\end{rem}
\begin{rem} 
If the spectrum of $H^\caR$ is \emph{level-spaced}, i.e., if the probability of a spacing significantly smaller than $\abs{\caR}^{-1}$ is small (as one can prove, e.g., for the standard Anderson model \cite{KM} and, at the bottom of the spectrum, for more general random models, \cite{DE}), then with large probability the spectral patch $S(s)$ can be chosen to consist of a simple eigenvalue, making $P(s)$ rank-one. 
Moreover, with large probability, for a large fraction of times $s$, the range of $P(s)$ stays close to an eigenprojection of the global Hamiltonian $H^{\T}(s)$. However, we do not expect this property to hold for all times $s$ on the basis of the hybridization result, Theorem \ref{thm:full_s'}, which shows that physical evolution cannot follow the non-local spectral flow.
\end{rem}
\begin{rem}
It follows from the relationship between $\epsilon$ and $\ell$ in \eqref{eq:scales'} that for the times $\sim\epsilon^{-1}$ the dynamical localization length is $O(\ln(\epsilon^{-1}))$. It is consistent with the estimates for rank-one perturbation of completely localized (time-independent) systems, where the growth of localization length is sub-polynomial in $\epsilon$, \cite{Last},  due to the zero Hausdorff dimensionality of the spectrum, \cite{DMS}.
\end{rem}

\section{Local gap structure}\label{sec:loc_gap_st}

Analyzing the spatial structure of spectral gaps is crucial to proving the adiabatic theorem described above. We introduce relevant concepts and state the corresponding results in this section. 

 We start with some supplementary notation.  By $\Lambda_R(y)\subset\Z^d$ we will denote a cube $\Lambda_R=\Lambda_R(y):=\pa{[-R,R]^d+y}\cap\Z^d$ for $y\in\Z^d$, with side length $2 R$. 
For a subset $\Phi\subset\Z^d$, we will denote by $\partial_\ell\Phi$ its $\ell$-extended boundary, i.e., 
\be\label{eq:extbound}
\partial_\ell\Phi=\set{x\in\Phi:\ \dist\pa{x,\Phi^c}\le \ell}.
\ee 
By $\Phi_\ell$ we will denote 
\be\label{eq:ellins}
\Phi_\ell=\Phi\setminus\partial_\ell\Phi.
\ee

For a Hermitian operator $H$, we denote by $P_J(H)$ the spectral projection of $H$ on the set $J \subset \mathbb{R}$.  For a positive real number $a$, $aJ$ denotes the interval obtained from $J$ by scaling the interval with respect to its midpoint by a factor of $a$. 
 For an operator $X$, we denote $\bar{X} := 1 -X$.
For $\mathcal A\subset \mathbb T$, $c \in \mathbb R_+$, and $\ell\in\mathbb N$,  let $\rho^\ell_{\mathcal A}$ be a (scaled) distance function  
\be\label{eq:distf}
\rho^\ell_{\mathcal A}(x)=\frac{\dist\pa{\mathcal A,\set{x}}}{\sqrt\ell}.
\ee 
 We  set
 \be\label{eq:Kcell}
 \norm{K}_{c,\ell} =\norm{\e^{-c\, \rho^\ell_{\mathcal A}}\, {K}\, \e^{c\, \rho^\ell_{\mathcal A}}}
 \ee
 This norm is multiplicative, i.e., 
 \be\label{eq:multipl}
 \norm{AB}_{c,\ell}\le \norm{A}_{c,\ell}\norm{B}_{c,\ell}
 \ee for a pair of operators $A,B$.

We now introduce the concepts of local and ultra-local gap structures. In order to describe our constructions with the least possible number of parameters, we will use the scale variable $\ell \in\N$ introduced in Theorem \ref{thm:AT-single'}. It will be convenient to formulate the concepts on a torus $\T$ whose linear dimension is $\caL=e^{c\sqrt\ell}$, but this condition can be relaxed. 

Let $J \subset J_{loc}$ and let $\set{\pa{E_n,\psi_n}}$ be a collection of eigenpairs for $H^{\T}(0)$  with energies in $J$. We will say that $H^{\mathbb{T}}(0)$  possesses  an {\it ultra-local gap structure} in $J$ if there exists  a disjoint collection $\set{\mathcal T_\gamma}$ of subsets of $\T$ with $\diam\pa{\mathcal T_\gamma}\le C\ell^{3/2}$  such that the following property holds: For each  $\psi_n$, there exists $\gamma$ such that 
\be\label{eq:locst}
\norm{\psi_n- P_{\hat J}\hspace{-.1cm}\pa{H^{{\mathcal T}_\gamma}(0)}\psi_n)}\le  e^{- c \sqrt{\ell}},
\ee
where $\hat J:=\set{x\in\mathbb R:\ \dist\pa{x,J}\le  e^{- c \sqrt{\ell}}}$.  
Let us note that the random Schr\"odinger operators $H(0)$ satisfying Assumptions \ref{assump:FMC} possess the ultra-local property with probability $\ge1-e^{-c\sqrt\ell}$ provided the length of the interval $J$ is of order $\ell^{-\xi}$ (in fact, a stronger statement holds, see Theorem \ref{thm:loc} below). Unfortunately,  localization in the usual sense (or in an ultra-local sense for that matter) breaks down under perturbations due to the hybridization phenomenon. As a result, the first step is to identify a weaker notion than ultra-locality that however remains stable under small perturbations. 
\begin{defn}\label{def:ls}
We will say that $H^{\mathbb{T}}(s)$  possesses  a {\it local gap structure} in $J \subset J_{loc}$ if there exists  a disjoint collection $\set{\mathcal T_\gamma}$ of subsets of $\T$ such that $\diam\pa{\mathcal T_\gamma}\le \ell^{3/2}$ for each $\gamma$ with the following properties:
\begin{thmlist}
\item\label{it:h1c} (Local Gap) There exist intervals $J_\gamma =[E^-_\gamma, E^+_\gamma]$ comparable in length to $J$  such that  
\be\label{eq:locgapc}
  J_\gamma \subset J\mbox{ and }  
\dist\pa{E^\pm_\gamma,\sigma(H^{\mathcal{T}_\gamma}(s))}\ge \Delta;
\ee
\item\label{it:h2c} (Support of spectral projections) Let $\mathcal T := \cup_\gamma{\mathcal T_\gamma}$.  Then 
\be\label{eq:suppprc}
\norm{P_{J}(s)\chi_{\T \setminus \mathcal{T}_{8\ell}}}\leq e^{-c \sqrt\ell},
\ee
and 
\be\label{eq:supppr'c}
\| P_{J_\gamma}(H^{\mathcal{T}_\gamma}(s)) - \chi_{\partial_\ell{\mathcal T}}  P_{J_\gamma}(H^{{\mathcal T}_\gamma}(s)) \chi_{\partial_\ell{\mathcal T}} - \chi_{{\mathcal T}_{8\ell}} P_{J_\gamma}(H^{{\mathcal T}_\gamma}(s)) \chi_{{\mathcal T}_{8 \ell}} \| \leq e^{-c \sqrt\ell}.
\ee
\end{thmlist}
\end{defn}
The unperturbed Hamiltonian possesses a local gap structure for small, but not too small, $\Delta$. As we shall see in the proof of Theorem \ref{thm:hypa}, the local gap structure is stable under perturbation, i.e.,  if the Hamiltonian possesses  a local gap structure  for $s=0$ on $J$, it  possesses it for {\it all} $s$ on a slightly smaller interval $J'$, provided $\beta$ is sufficiently small. The reason for this stability is related to the fact that, under small local perturbations, an eigenstate with energy $E$ is close to the range of a thin spectral projection of the unperturbed operator centered at $E$. Since the latter is supported in the localized patches $\caT_\gamma$, so is the eigenstate. The locality property is fully compatible with the hybridization effect: Even if initially the state is ultra-local (concentrated in a single patch $\caT_{\gamma_o}$), it can hybridize to a number of different patches $\caT_\gamma$ as $s$ increases.

 The scaling of various objects with $\ell$ depends on $q,d$ and our choice of stretch-exponential error $\exp(-c \sqrt{\ell})$. The correct scaling of $\Delta$ and $\beta$ to ensure the existence of local gap structure is given in Theorem~\ref{thm:hypa}.

Once the local gap structure for the family $H(s)$ is established, one can use an (enhanced) version of the standard, gapped adiabatic theorem (Lemma \ref{lemma:all_orders}) to control the behavior of the individual spectral patches $P_{J_\gamma}\hspace{-.1cm}\pa{H^{{\mathcal T}_\gamma}(s)}$, invoking Definition \ref{it:h1c}. This in turn allows us to control the physical evolution of spectral data $Q(s)$ for $H^\T(s)$ near the energy $E$ (see Section \ref{sec:Q} for details). Finally, we show that this translates to the adiabatic theorem for the (distorted) Fermi projection, Theorem \ref{thm:AT-loc}. The principal idea here is that the removal of the spectral data $Q(s)$ on one hand creates a spectral gap for $H$ (making the standard adiabatic theorem applicable) and on the other does not distort the adiabatic behavior of the system too much since $Q(s)$ itself evolves adiabatically, a feature verified in the previous step.

We will use the shorthand $P_J(s) :=P_J(H^{\mathbb T}(s))$ and $P_J := P_J(0)$ in this section. 
  
 We will show in Section~\ref{sec:Loc} that Anderson-type models possess a local gap structure in the sense of Definition \ref{def:ls}. In fact, a stronger statement holds:

\begin{thm}[Local gap structure of $H^{\mathbb{T}}(s)$]\label{thm:hypa}
Suppose that $H$ satisfies Assumptions \ref{assump:FRC}--\ref{assump:FMC} and the family $H(s)$ satisfies Assumption \ref{hyp1}. We consider a torus $\T$ whose linear dimension is $\caL$. Then, there exist constants $c, \set{c_i}_{i=1}^6$ such that for any $a\le c_1$,
\be\label{eq:scales}
\mathcal{L} = e^{a \sqrt{\ell}}, \quad V_\ell = \ell^{d +1/2}, \quad  \delta = c_2 \ell^{-\xi}, \quad \Delta = c_3 V_\ell^{-1} \ell^{ - \xi},
\ee
 $\ell$ large enough, and $\beta \leq \ell^{-p_1}$, $H^{\mathbb{T}}(s)$  {\it possesses a local gap structure} for the energy interval $J = (E - 6 \delta, E + 6 \delta)$: One can find a disjoint collection $\set{\mathcal T_\gamma}$ of subsets of $\Lambda$ such that $|\mathcal T_\gamma| \leq c_4V_\ell$, $\diam\pa{\mathcal T_\gamma}\le c_5 \ell^{3/2}$ for each $\gamma$ and  the following conditions hold true with probability $ > 1 - e^{- c_6 \sqrt{\ell}}$:
\begin{thmlist}
\item\label{it:h1} (Local Gap) There exist intervals $J_\gamma =[E^-_\gamma, E^+_\gamma]$ such that  
\be\label{eq:locgap}
(E - 3 \delta, E+ 3 \delta) \subset  J_\gamma \subset J\mbox{ and }  
\dist\pa{E^\pm_\gamma,\sigma(H^{\mathcal{T}_\gamma}(s))}\ge\Delta;
\ee
\item\label{it:h2} (Support of spectral projections) Let $\mathcal T := \cup_\gamma{\mathcal T_\gamma}$.  Then 
\be\label{eq:supppr}
\norm{P_{J}(s)\chi_{\Lambda \setminus \mathcal{T}_{8\ell}}}\leq e^{-c \sqrt\ell},
\ee
and 
\be\label{eq:supppr'}
\| P_{J_\gamma}(H^{\mathcal{T}_\gamma}(s)) - \chi_{\partial_\ell{\mathcal T}}  P_{J_\gamma}(H^{{\mathcal T}_\gamma}(s)) \chi_{\partial_\ell{\mathcal T}} - \chi_{{\mathcal T}_{8\ell}} P_{J_\gamma}(H^{{\mathcal T}_\gamma}(s)) \chi_{{\mathcal T}_{8 \ell}} \| \leq e^{-c \sqrt\ell}.
\ee
\item\label{it:h3} (Exponential Decay of Correlations) 
Let $\mathcal A_o= \partial_\ell \mathcal{T}_\gamma \cup (\mathcal{T}_\gamma)_{8 \ell}$, then (with $\caA=\caA_o$ in \eqref{eq:distf}--\eqref{eq:Kcell}) we have
\be\label{eq:expdeccor}
\norm{\pa{H^{{\mathcal T}_\gamma}(s))-z}^{-1}}_{c,\ell} \leq \frac{\ell^{3d}}{\Delta} \frac{1}{\langle Im\, z\rangle},
\ee
for $z \in \mathbb{C}$ with 
$\Re(z) = E^{\pm}_\gamma$.
\end{thmlist}

The dependence on $\beta$ here is deterministic, i.e., there exists a subset of configurations of probability $ > 1 - e^{- c_6 \sqrt{\ell}}$ such that the conclusions hold for all $\beta \leq \ell^{-p_1}$.

\end{thm} 
This assertion will be proved in Section \ref{sec:proofhyp}.

An additional statement that we will need in our proof of Theorem \ref{thm:QHE} is

\begin{thm}[Local adiabatic theorem for distorted Fermi projection]\label{thm:AT-loc} 
In the setting of Theorem~\ref{thm:hypa}, assume in addition that \eqref{eq:scales'} holds
and fix $N \in \mathbb{N}$. Then for $\ell$ large enough, there exists a smooth family of orthogonal projections $\mathcal Q(s)$ with the following properties:
\begin{thmlist}
\item\label{it:1Q}
$\norm{[\mathcal Q(s),H^{\mathbb{T}}(s)]}\le C_N \left( \epsilon \Delta^{-1} + e^{- c \sqrt{\ell}} \right)$;
\item\label{it:2Q} $\norm{P_{< E-6\delta}(H^{\mathbb{T}}(s))\bar{\mathcal Q}(s)}+\norm{ \mathcal Q(s)P_{> E+6\delta}(H^{\mathbb T}(s))}\le C_N \left( \epsilon \Delta^{-1} + e^{- c \sqrt{\ell}} \right)$; 
\item\label{it:3Q} If we denote by $\mathcal Q_\epsilon(s)$ the solution of the IVP
$i \epsilon \dot{\mathcal Q}_\epsilon(s)   = [\mathcal Q_\epsilon(s), H^{\mathbb T}(s)]$,  $\mathcal Q_\epsilon(0) = \mathcal Q(0)$,
we have
\be\label{eq:truebnd}\norm{\mathcal Q_\epsilon(s)-\mathcal Q(s)}\le C_N \left( \epsilon^N \left( \frac{1}{\Delta^N} +  \frac{1}{\delta^{2N+1}}  \right) +   e^{- c \sqrt{\ell}} \right).\ee
\end{thmlist}
Furthermore, for $s=0$ and $s=1$, the inequalities in (i) and (ii) hold without the terms proportional to $\epsilon$.
\end{thm}
This assertion will be proved in Section \ref{sec:locad}.

\section{Localization on a torus}\label{sec:Loc}
\subsection{Consequences of Assumptions \ref{hyp1}--\ref{assump:FMC}}
 We first note that Assumptions \ref{hyp1}--\ref{assump:FMC} imply localization on a torus as well (e.g., \cite[Theorem 11.2]{AW}):
\be\label{eq:FMCp}
\sup_{E\in J_{loc}} \mathbb E\pa{\abs{(H^{\mathbb{T}} - E-i\eta)^{-1}(x,y)}^q}\le Ce^{-c {d_\T}(x,y)}\mbox { for all } x,y\in \mathbb{T},
\ee
where $d_\T(x,y)$ represents the usual distance function on a torus.

Another consequence of these hypotheses is 
\begin{lemma}[The Wegner estimate] \label{lemWeg} 
Let $\Theta\subset \mathbb T$. For all  $E\in J_{loc}$,
\begin{align}
\mathbb P \set{\dist\set{E,\sigma (H^{\Theta})} \le \nu}\le   C \nu^q \abs{\Theta}.
\end{align}
\end{lemma}
For a proof, see e.g., \cite[the proof of Proposition 5.1]{ETV}.

Together with Assumption \ref{assump:FRC}, Lemma \ref{lemWeg} yields
\begin{lemma}[Distance between spectra] \label{lemWegBB} 
Let $\Theta,\Phi \subset \mathbb T$ be such that $\dist\pa{\Theta,\Phi}>r$.  Then 
\begin{align}
\mathbb P \set{\dist\pa{\sigma (H^{\Theta})\cap J_{loc},\sigma (H^{\Phi})\cap J_{loc}} \le \nu}\le   C \nu^{2q} \abs{\Theta}\abs{\Phi}.
\end{align}
More generally, if a collection $\set{\Theta_i}_{i=1}^n$ of subsets in $\mathbb T$ satisfies $\dist\pa{\Theta_i,\Theta_j}>r$ for $i\neq j$, $\abs{\Theta_i}\le D$ for all $i$, and $E\in \mathbb R$, then 
\begin{align}
\mathbb P \set{\dist \pa{E,\sigma (H^{\Theta_i})} \le \nu \mbox{ for all } i}\le   \pa{C \nu^q D}^n.
\end{align}
\end{lemma}

We recall that by $P_I(H)$ we denote the spectral projection of $H$ onto a set $I$, and that $P_E(H)$ stands for $P_{(-\infty,E]}(H)$. We will often suppress the $H$ dependence in this notation, denoting by $P_I^\Theta$ a projection $P_I(H^\Theta)$ and analogously for $P_I(H^\T)$. 

A  subtler implication of our assumptions on $H^\Theta$  is the fact that 
the associated {\it eigenfunction correlator} $Q^\Theta(x,y; J_{loc})$ for $x,y\in \Theta$, defined by
\be
Q^\Theta(x,y; J_{loc})=\sum_{\lambda\in\sigma(H^\Theta)\cap J_{loc}}\abs{P^\Theta_{\set{\lambda}}(x,y)}
\ee
satisfies
\be\label{eq:eigcor}
\mathbb E Q^\Theta(x,y; J_{loc})\le  \e^{- c \abs{x-y}_\Theta}
\ee
for some $c >0$ that depends only on $ \mu$ and $q$. For the non correlated randomness, see, e.g. \cite[Theorem 7.7]{AW} (the proof relies on the so-called spectral averaging procedure available in this case). For a more general class of correlated random models, such an assertion was derived in \cite[Theorem 4.2]{ESS}.

The relation \eqref{eq:eigcor} implies that all eigenstates in $P_{J_{loc}}^\Theta$ are localized with  large probability. We make this statement quantitative below.
\begin{defn}\label{def:localizing}
Let  $c,\ell>0$ be fixed. We say that a set $\Theta\subset \mathbb T$  is {\it $(c,\ell)$-localizing} for $H^\T$ in the interval $I \subset J_{loc}$ if for all eigenpairs $\pa{E_n,\psi_n}_{E_n\in I}$ of $H^\Theta$ there exists  a set  $\set{x_n}$ in $\Theta$  such that
\be\label{eq:loceig}
|\psi_n(y)|\le e^{- c |y-x_n|_\Theta}\mbox{ for any } y\in\Theta \mbox{ such that }  |y-x_n|_\Theta\ge \sqrt\ell.
\ee
\end{defn}
We then have the following result:
\begin{thm}\label{thm:loc}
Suppose that Assumptions \ref{assump:FMC}--\ref{assump:multiplicities} hold. Then there exist $c>0$ such that  the probability that  a set $\Theta\subset \mathbb T$  is $(c,\ell)$-localizing for $H^\T$ in the interval $J_{loc}$ is $\ge 1-C|\Theta|^2 e^{-c \sqrt\ell}$.
\end{thm}
The proof of this statement can be found in Appendix \ref{sec:Wannier} (Theorem \ref{lem:cent_mass_deg}).

Sometimes it will be useful to compare a finite volume projection $P_E^\mathbb{T}:=P_E(H^\T)$ with the infinite volume one $P_E$. To be able to do so, we will use the periodic extension $\tilde P_E^\mathbb{T}$ of $P_E^\mathbb{T}$ to $\Z^d$, i.e., 
\[\tilde P_E^\mathbb{T}(x,y)=\begin{cases} P_E^\mathbb{T}(x\hspace{-.3cm}\mod \caL Z^d,y\hspace{-.3cm}\mod \caL Z^d)& x-y\in\T\\ 0 & x-y\notin\T\end{cases}\]
The next assertion implies that deep inside $\mathbb{T}$, $P_E$  and  $\tilde P_E^\mathbb{T}$ are close.
 \begin{prop}\label{prop:speed2}
Suppose that Assumptions \ref{hyp1}--\ref{assump:FMC} hold. Then there  exists $c>0$ such that the probability
\be\label{eq:toruscomp}
\mathbb P\pa{\norm{\pa{P_E-\tilde P^{\mathbb T}_E}\chi_{\Lambda_{\caL/2}(0)}}> \e^{-c\mathcal L}}\le \e^{-c\mathcal L}.
\ee
\end{prop}

For a proof, see \cite[Lemma 4.11]{EPS}. The argument is closely related to the one used in the proof of the following result that establishes the localization property of some bounded functions of $H$ in the mobility gap. 
\begin{lemma}\label{lem:decGh}
Suppose that Assumptions \ref{hyp1}--\ref{assump:FMC} hold. Then for any  $I:=[E_1,E_2]\subset J_{loc}$  and any $\Theta \subset \mathbb T$, there exists $c>0$ such that 
\be\label{eq:decGh'}
 \mathbb E \abs{P_\sharp^\Theta(x,y)\,}\le   \e^{ -c \abs{x-y}_\Theta},\quad \sharp=I,E,
\ee
for all  $x,y\in \Theta$. Moreover, for any $z\in\mathbb C$ with ${Re(z)}\in I/2$, we have
\be\label{eq:decGh}
\mathbb E \abs{\pa{\bar P_{I}^\Theta\pa{H^\Theta-z}^{-1}}(x,y)}\le \frac {1} {E_2-E_1}\frac{\e^{- c \abs{x-y}_\Theta}}{\langle Im\, z\rangle}
\ee
\end{lemma}

\begin{proof}
Let $\sharp=I$. Since $\Theta$ is finite, the spectrum of $H^\Theta$ is a discrete set. By \eqref{eq:FMC}, \[\{E_{1},E_{2}\}\not\subset\sigma\pa{H^\Theta}\] almost surely. Thus the spectral projection $P_I^\Theta$ is equal to
\be\label{eq:conre}
P_I^\Theta=-\pa{2\pi}^{-1}\int_{-\infty}^\infty \sum_{j=1}^2(-1)^j\pa{H^\Theta-iu-E_j}^{-1}du
\ee
almost surely, see \eqref{eq:Rder}. Using $|\pa{H^\Theta-iu-E_j}^{-1}(x,y)| \leq |u|^{-1}$, we get a bound
\[
\left| { P_I^\Theta}(x,y) \right| \leq \max_j  {\pi}^{-1}\int_{-\infty}^\infty \abs{\pa{H^\Theta-iu-E_j}^{-1}(x,y)}^q {|u|^{q-1}} du.
\]
We note that for $|u| \geq 1$, we can  decompose
\[\pa{H^\Theta-iu-E_j}^{-1}=-\pa{iu+E_j}^{-1}+\pa{iu+E_j}^{-1}H^\Theta\pa{H^\Theta-iu-E_j}^{-1}.\]
Thus, using  \eqref{eq:FMCp}, $r$-locality of $H$, and $|H(x,y)| \leq C$,
\begin{align*} \mathbb E \abs{P_I^\Theta(x,y)} &\le   {\pi }^{-1}{\sum_j  \sup_{u\in \mathbb{R}}}\Big( \mathbb E\abs{\pa{H^\Theta-iu-E_j}^{-1}(x,y)}^q\int_{[-1,1]}{\abs{u}^{q-1}}{du}  \\ &\qquad +  C\max_{\substack{z\in\Z^d:\\ \abs{z-x}\le r}}\mathbb E\abs{\pa{H^\Theta-iu-E_j}^{-1}(z,y)}^q\int_{[-1,1]^c}{\abs{u}^{q-2}}{du} \Big)\\ & \le  C\e^{- c \abs{x-y}_\Theta}.
\end{align*}
Since $ \abs{P_I^\Theta(x,y)}\le1$ for all $x,y\in\Theta$, by modifying $c$ if necessary we get \eqref{eq:decGh'} for  $\sharp=I$.  The argument for $\sharp=E$ is nearly identical.

To get the second assertion of the lemma, we use
\[\pa{H^\Theta-z}^{-1}=-\pa{iIm(z)+1}^{-1}+\pa{iIm(z)+1}^{-1}\pa{H^\Theta-Re(z)-1}\pa{H^\Theta-z}^{-1}\]
and 
\[\bar P_{I}^\Theta\pa{H^\Theta-z}^{-1}=
-\pa{2\pi}^{-1}\sum_{j=1}^2\int_{-\infty}^\infty \pa{z-E_j-iu}^{-1}\pa{H^\Theta-iu-E_j}^{-1}du.\]
They yield
\begin{multline*}
\bar P_{I}^\Theta\pa{H^\Theta-z}^{-1}=-\pa{iIm(z)+1}^{-1}\bar P_{I}(H^\Theta)+\\ \pa{2\pi}^{-1}\sum_{j=1}^2\pa{iIm(z)+1}^{-1}\pa{H^\Theta-Re(z)-1}\int_{-\infty}^\infty \pa{z - E_j-iu}^{-1}\pa{H^\Theta-iu-E_j}^{-1}du.
\end{multline*}
Since $\bar P_{I}^\Theta=1-P_{I}^\Theta$, $\abs{iIm(z)+1}={\langle Im\, z\rangle}$, and $\abs{{z-E_j-iu}}^{-1}\le2\pa{E_2-E_1}^{-1}$ for any  ${Re(z)}\in I/2$ and $u\in \mathbb R$, the remaining argument is identical to the one used in the proof of the first bound.
\end{proof}
We will be using the probabilistic version of Lemma \ref{lem:decGh}, which follows from the previous statement by Markov's inequality.
\begin{lemma}\label{lem:locdec}
Suppose that Assumptions \ref{hyp1}--\ref{assump:FMC} hold. 
Let $J:=[E_1,E_2]\subset J_{loc}$. Then, there exists $c>0$ such that for any $\Theta \subset \mathbb T$ with $\abs{\Theta}\le\ell^{3/4}$,  the probability that  for all $x,y$ with $\abs{x-y}_\Theta\ge\sqrt\ell$,
\be\label{eq:locdec1'}
\abs{\pa{ P_{J}^\Theta}(x,y)},
\abs{\pa{\bar P_{J}^\Theta\pa{H^\Theta-z}^{-1}}(x,y)}\le  \e^{- c \abs{x-y}_\Theta}\ee
is $\ge 1-e^{-c \sqrt\ell}$.
\end{lemma}

\subsection{Local gap structure of $H^{\mathbb{T}}$}
Here we will again suppose that Assumptions \ref{hyp1}--\ref{assump:FMC} hold. 

Given scales $\ell < \mathcal L$ with $\mathcal L\hspace{-4pt}\mod\hspace{-2pt}\pa{\frac 3 2 \ell}=\ell$, and $\ell$ even, we cover the torus $\T=\T_{\mathcal L}^d$ with the collection of boxes
\be\label{eq:boxes}
 \set{ {\Lambda}_{\ell}(a)}_{a \in  \Xi_{\ell}},
\ee
where 
\be\label{eq:boxesXi}
 \Xi_{\ell}:= \pa{ \tfrac 3 2\ell  \Z}^{d}/\mathcal L\Z^d.
\ee
Here the boxes $ {\Lambda}_{\ell}(a)$ (defined earlier as a subset of $\Z^d$) are  understood, with a slight abuse of notation, as subsets of $\T$, i.e., $ {\Lambda}_{\ell}(a)=\set{x\in\T:\ d_\T(x,a)\le \ell}$.  We recall that we use a $\max$ distance throughout this paper. We will refer to this collection of boxes as a {\it suitable $\ell$-cover} of $\T$.

The (trivial) properties of suitable covers are encapsulated by the following lemma. 
\begin{lemma}\label{lem:boxes}
Let $r<\ell <  \mathcal L  $. Then, a suitable $\ell$-cover satisfies
\begin{thmlist}
\item $\T=\bigcup_{a \in  \Xi_{\ell}} {\Lambda}_{\ell}(a)$;
\item For all $y \in\T$  there is $a=a(y) \in  \Xi_{\ell}$  such that 
  $ \Lambda_{\ell/4}(y)\subset {\Lambda}_{\ell}(a)$. For such a value of  $a$ we will denote ${\Lambda}_{\ell}^{(y)}:= {\Lambda}_{\ell}(a)$;
  \item $\Lambda_{\ell/4}(a)\cap {\Lambda}_{\ell}(a') =\emptyset$ for all  $a,a'\in \Xi_{\ell},\, a\neq a'$;
\item  $ \pa{\tfrac{\mathcal L} {\ell}}^{d}\le  \abs{\Xi_{\ell}}\le   \pa{\tfrac{2\mathcal L} {\ell}}^{d}.$
\end{thmlist}

 Furthermore, any box  $ {\Lambda}_{\ell}{(a)}$ with $a \in  \Xi_{\ell}$   overlaps with no more than $2d$ other boxes in  the $\ell$-cover, and any non-overlapping boxes are separated by a distance  $>r$.
\end{lemma}
 
 Let  $\mathcal S$ be a subset of a suitable $\ell$-cover such that the boxes $\set{\Lambda_\ell(a)}_{\mathcal S}$ are separated by a distance $r$. Fix $E\in  J_{loc}$, then, by Lemma~\ref{lemWegBB}, for all $\nu>0$  we have
 \begin{align}\label{lemSep} 
\mathbb P\set{\dist \pa{E,\sigma(H^{\Lambda_\ell(a)})} \le \nu \mbox{ for all } \Lambda_\ell(a)\in\mathcal S}\le \pa{C\nu^q \ell^d}^{\abs{\mathcal S}}.
\end{align}

We now inspect the structure of  $ P_I(H^{\mathbb T})$. We will work with the scale $\ell$ and the interval $I\subset J_{loc}$ such that  
\be\label{eq:scalesrel}
\mathcal L\gg \ell\gg1,\quad \abs{I}=c \ell^{-\frac{d}q}.
\ee
for an $\ell$--independent constant $c$. We recall that we are using a convention where $c$ denotes a sufficiently small constant and $C$ a sufficiently large constant. The values of these constants can change equation by equation.

We endow the set $ \Xi_{\ell}$ with the usual graph structure, i.e., we will think of its elements as vertices and introduce edges $\langle a,b\rangle$ between neighboring elements $a,b\in \Xi_{\ell}$ separated by a distance $\frac32\ell$ on the torus $\T$. By $\mathcal R_M$ we will denote a set of all connected subgraphs of $ \Xi_{\ell}$ with cardinality $M$, and by $\mathcal S_M$ we will denote a collection of sets $
\set{\cup_{a\in R} {\Lambda}_{\ell}{(a):\ R\in \mathcal R_M}}$. 

\begin{lemma}\label{lem:S_Mcount}
The cardinality of  $\mathcal R_M$ is bounded by  
\be\label{eq:S_Mfull}
\pa{2d\e}^M  \abs{\Xi_{\ell}}\le \pa{\tfrac{2\mathcal L} {\ell}}^{d}\pa{2d\e}^M.
\ee
\end{lemma}
\begin{proof}[Proof of Lemma \ref{lem:S_Mcount}]
We first note that each set $S$ in $\mathcal S_M$ looks like a 
compressed $d$-dimensional polycube of size $M$, and that we can bound the number of distinct $S_M$s using the same method as for the regular polycubes, see e.g., \cite{BBR}. To make the argument self-contained, we reproduce it here.

A $d$-dimensional polycube of size $n$ is a connected set of $n$ cubical cells on the lattice $\Z^d$, where a pair of polycubes is considered adjoint if they share a (($d-1$)-dimensional) face. Two fixed polycubes are equivalent if one can be transformed into the other by a translation.

Given $S$, we assign the numbers $1,\ldots,M$ to the cubes
of $S$ in lexicographic order. We now search for the (cube) connectivity graph $G$ of $S$, beginning with cube $1$. During the search, any cube $c\in S$ is reached through an
edge $e$ and connected by the edges of $G$ to at most $2d-1$ other cubes. We label each outgoing edge $e'$ with a pair $(i, j)$, where $i$ is the number associated with $c$, and $1 \le j \le 2d-1$ is determined by the orientation of $e'$ with respect to $e$. By the end of the search, 
each of the $M-1$ edges in the resulting spanning tree is given a unique label from
a set of $(2d- 1)M$ possible labels. This is an injection from polycubes of size $M$ to ($M-1$)-element subsets of a set of
size $(2d- 1)M$, and so the number of distinct shapes for $S$ is  
 bounded by 
 \be\label{eq:S_M}
{ {\binom {(2d-1)M}{M-1}}}\le \pa{2d\e}^M .
\ee
 
The total number of sets $S$ can be now bounded by noticing that they are contained in the set of all translates of the distinct shapes of $S$ by elements of $\Xi_{\ell}$, yielding \eqref{eq:S_Mfull}.

\end{proof}

For any given configuration $\omega$, let $\tilde {\mathcal T}$ denote the union of the boxes $\Lambda_\ell(a)$ with $a\in\Xi_{\ell}$ such that the restricted Hamiltonian $H^{\Lambda_\ell(a)}_\omega$  has at least one eigenvalue in the interval $2I$. Let $\mathcal T$ denote the union of boxes $\Lambda_\ell(b)$ with $b\in\Xi_{\ell}$ that has a non-trivial overlap with $\tilde{\mathcal T}$.  We will enumerate by $\set{\mathcal T_\gamma}$ a set of connected (with respect to the graph structure of $\T$) components in $\mathcal T$, i.e.,
\[\mathcal T=\cup_\gamma \mathcal T_\gamma, \quad \mathcal T_\gamma\cap \mathcal T_{\gamma'}=\emptyset,\quad \mathcal T_\gamma\in {\mathcal S_M} \mbox{ for some } M\in\mathbb N.\]

For a given $\mathcal T$, we will denote by $M(\mathcal T)$ the size of the largest connected component,
\[M(\mathcal T)=\max_\gamma \left\{ M:\ \mathcal T_\gamma\in {\mathcal S_M} \right\}.\]
 For an integer $N$, let $\Omega_{N}$ denote a subset  of the full configuration space for which 
 \[M(\mathcal T)< N.\] 
\begin{lemma}\label{clustprob}
Let $\ell>r$ and $I\subset J_{loc}$ with $ |I|^{q}< c \ell^{-d}$. Then for $c$ small enough we have 
\be
\mathbb P( \Omega^c_{N})\le \pa{\tfrac{2\mathcal L} {\ell}}^{d}\e^{-N}.
\ee
\end{lemma}
\begin{proof}
For any $\omega\in  \Omega^c_{N}$, there exists at least one cluster $ \mathcal T_\gamma\in {\mathcal S_M}$ with $M\ge N$.
 Let $\tilde{\mathcal T}_\gamma $ denote the union of boxes that generates $\mathcal T_\gamma$, i.e.,  $\mathcal T_\gamma$ is formed by all boxes that overlap with at least one box in $\tilde{\mathcal T}_\gamma$. We note that $\tilde{\mathcal T}_\gamma $ is in general not uniquely defined, but this will not play a role in our argument. We also remark that any  box $\Lambda_l(a) \subset \tilde{\mathcal T}_\gamma$ overlaps with $3^d$ boxes, so $\abs{\tilde{\mathcal T}_\gamma}\le 3^d \abs{\mathcal T_\gamma}$. Let $U$ be a collection of vectors in $\mathbb{R}^d$ whose components  take binary values. Then $\Xi_{\ell}= \cup_{e \in U} \Xi_{\ell,e}$, where $\Xi_{\ell, e} = \frac{3}{2} e + \pa{3\ell  \Z}^{d}/\mathcal L\Z^d$, and $\Xi_{\ell, e} \cap\Xi_{\ell, e'}=\emptyset$ for $e\neq e'$ and 
\be\label{eq:nonover}
{\Lambda}_{\ell}(a)\cap {\Lambda}_{\ell}(a') =\emptyset \mbox{ for all  } a\in\Xi_{\ell,e},\quad a\in\Xi_{\ell,e'},
\ee 
using the fact that  $\ell$ is even.  Hence, for any $S\subset \Xi_{\ell}$, there exists $e \in U$ such that $\abs{S\cap\Xi_{\ell,e}} \geq 2^{-d} \abs{S}$. In particular, the number of non-overlapping boxes in $\tilde{\mathcal T}_\gamma$ is at least $6^{-d}M$ due to \eqref{eq:nonover}.

We are now in a position to apply \eqref{lemSep} to conclude that the probability that a {\it fixed} configuration $\mathcal T$  has at least one cluster $ \mathcal T_\gamma\in {\mathcal S_M}$ with $M\ge N$ is bounded by $\pa{C \abs{I}^q \ell^d}^{6^{-d}M}$. It follows now from Lemma \ref{lem:S_Mcount} that
\be
\mathbb P( \Omega^c_{N}) \le \sum_{M=N}^\infty \pa{\tfrac{2\mathcal L} {\ell}}^{d}\pa{\pa{2d\e}^{\pa{6^d}}C\abs{I}^{q} \ell^{d}}^{6^{-d}M}.  
\ee
This is less than or equal to $\pa{\tfrac{2\mathcal L} {\ell}}^{d}\e^{-N}$ provided that $c$ in \eqref{eq:scalesrel} is small enough.

\end{proof}

 For an integer $N$, we now consider a subset $\Omega_{loc,N}$ of the full configuration space for which $\mathbb{T}$ and all of the sets in $\set{S_M}_{M=1}^N$ are $\ell/10$-localizing and satisfy \eqref{eq:locdec1'}. 
 \begin{lemma}\label{clustloc}
There exists constants $C,c>0$ such that
 \be
\mathbb P(\Omega_{loc,N}^c)\le CN^2\pa{{2\mathcal L\ell}}^{d}\pa{2d\e}^N e^{-c \sqrt\ell}.
\ee 
\end{lemma}
\begin{proof}
The  total number of $\set{S_M}_{M=1}^N$ is bounded by 
\[\sum_{M=1}^N\pa{\tfrac{2\mathcal L} {\ell}}^{d}\pa{2d\e}^M< 2\pa{\tfrac{2\mathcal L} {\ell}}^{d}\pa{2d\e}^N\] thanks to Lemma \ref{lem:S_Mcount}. Their maximal volume is bounded by $N\ell^d$. Thus, we can bound
 \be
\mathbb P(\Omega_{loc,N}^c)\le  C\pa{\tfrac{2\mathcal L} {\ell}}^{d}\pa{2d\e}^N \pa{N\ell^d}^2 e^{- c \sqrt\ell}=CN^2\pa{{2\mathcal L\ell}}^{d}\pa{2d\e}^N \e^{-c\sqrt\ell}
\ee 
using  Theorem \ref{thm:loc} and Lemma \ref{lem:locdec}.
\end{proof}

 We now optimize $N$ from the previous two lemmas. To this end, we  pick $N=\lfloor c \sqrt{\ell}\rfloor$. Then, using Lemmata~\ref{clustprob}--\ref{clustloc}, for $\ell$ large enough and intervals $I\subset J_{loc}$ satisfying $ |I|< c \ell^{-d/q}$, we have 
   \be
\mathbb P(\pa{\Omega_{N}\cap \Omega_{loc,N}}^c)\le  {\mathcal L}^{d}\e^{-c\sqrt\ell}.
\ee 

For $\omega \in \Omega_{N}\cap \Omega_{loc,N}$, the number of eigenvalues of $H^{\mathcal T_\gamma}$ cannot exceed $\abs{\mathcal T_\gamma}\le N \ell^d \leq C \ell^{d+ 1/2}$. Hence, for each $\gamma$, we can find $J_\gamma: = [E^{-}_\gamma, E^+_\gamma]$ such that 
$$
I/2 \subset J_\gamma \subset I \quad  \mbox{and} \quad \mathrm{dist}(E_\gamma^\pm, \sigma(H^{\mathcal T_\gamma})) \geq c \ell^{-d- 1/2}\abs{I}.
$$
We note that 
\be\label{eq:maxdiam}
\max_\gamma\diam\pa{\mathcal{T}_\gamma}\le  L: = C \ell^{3/2}.
\ee

Let $\Omega_G$ be a subset of the configuration set $\Omega_{N}\cap \Omega_{loc,N}$ such that, for $c$ small enough, $\omega\in \Omega_G$, $z \in \mathbb{C}$ with $\mathrm{Re} (z) = E^{\pm}_\gamma$, and all  $x,y\in \mathcal T_\gamma$, the following bound holds:
\be\label{eq:corprep}
\sup_{\mathcal T_\gamma}\abs{\pa{\pa{H^{\mathcal T_\gamma}-z}^{-1}}(x,y)}\e^{ c\ell^{-1/2}\, \abs{x-y}_{{\mathcal T_\gamma}}}\le C {\ell^{d+\frac{1}{2}}} \abs{I}^{-1}{\langle Im\, z\rangle}^{-1}.
\ee

Applying Lemma \ref{lem:locdec} with $J = E_\gamma^\pm +[-c \ell^{-d- 1/2}\abs{I}, c \ell^{-d- 1/2}\abs{I}] $ and $z \in \mathbb{C}$ with $\mathrm{Re} (z) = E^{\pm}_\gamma$ yields 
\[
\mathbb P\pa{\Omega^c_{G}}\le {\mathcal L}^{d}e^{-c\sqrt\ell}.
\]

\begin{prop}\label{prop:sup}
Let $\omega \in \Omega_G$, and let $I  \subset J_{loc}$ be such that $ |I|< c \ell^{-d/q}$. Suppose that $\ell$ is large enough, then
\begin{thmlist}
\item\label{prop512} (Local Gap) There exist intervals $J_\gamma =[E^-_\gamma, E^+_\gamma]$ such that  
\be\label{eq:locgapH}
I/2 \subset  J_\gamma \subset I \quad \mbox{ and }  \quad 
\dist\pa{E^\pm_\gamma,\sigma(H^{\mathcal{T}_\gamma})}\ge c \ell^{-d- 1/2}\abs{I};
\ee
\item\label{prop512a} (Support of spectral projections) 
\be\label{eq:suppprH}
\norm{P_{I}(H^{\mathbb{T}})\delta_x}\leq e^{-c \sqrt\ell} \mbox{ for any } x \in \mathbb{T} \setminus \mathcal{T}_{\ell}
\ee
(recall \eqref{eq:ellins}), 
and 
\be\label{eq:supppr'H}
\norm{P_{J_\gamma}(H^{{\mathcal T}_\gamma})\delta_x}\leq e^{-c \sqrt\ell} \mbox{ for any } x\notin \partial_{\ell/8}\mathcal{T} \cup \mathcal{T}_{\ell}; 
\ee
\item\label{prop512b} (Exponential Decay of Correlations) 
Let $\mathcal A_o$ be any subset of $\caT_\gamma$, then (with $\caA_o$ in \eqref{eq:distf}--\eqref{eq:Kcell}) we have
\be\label{eq:expdeccor'}
\norm{\pa{H^{{\mathcal T}_\gamma}-z}^{-1}}_{c,\ell} \leq {\ell^{4d+1/2}}{\abs{I}^{-1}} {\langle Im\, z\rangle}^{-1}
\ee
for $z \in \mathbb{C}$ with 
$\Re(z) = E^{\pm}_\gamma$.

\end{thmlist}
\end{prop}

\begin{proof}

Proposition \ref{prop512} has been established earlier, and Proposition \ref{prop512b} is a consequence of  \eqref{eq:corprep}. This leaves us with the task of  proving Proposition \ref{prop512a}. 

Let $\set{\lambda_n,\psi_n}$ be an  eigenpair for $H^{\mathbb{T}}$ in $I$, and let $x_n$ be its localization center. We first check that $x_n\in \tilde{\mathcal T}$. Indeed,  suppose  that $x_n\notin \tilde{\mathcal T}$. Then, by the properties of the suitable cover, there exists a box $\Lambda_\ell(a) \not\subset \tilde {\mathcal{T}}$  such that  $\Lambda_{\ell/4}( x_n)\subset\Lambda_\ell{\pa{a}}$. Moreover, $\omega\in\Omega_G\subset\Omega_{loc,N}$  implies that $\mathbb{T}$ is  $\ell/10$-localizing, so in particular
 \[|\psi_n(y)|\le C\e^{-\mu |y-x_n|_{\Lambda_\ell{\pa{a}}}}\mbox{ for  }   |y-x_n|_{\Lambda_\ell{\pa{a}}}\ge \sqrt{\ell/10}.\]
 We can now use Lemma \ref{outbad} below to conclude 
\be\sigma\pa{H^{ \Lambda_\ell{\pa{a}}}}\cap 2I\neq \emptyset,
\ee 
which means that $ \Lambda_\ell{\pa{a}}\subset \tilde{\mathcal T}$, a
 contradiction. This establishes \eqref{eq:suppprH}, since for any  $x \in \mathbb{T} \setminus \mathcal{T}_{\ell}$ we have $\dist\pa{x,\tilde {\mathcal{T}}}\ge\ell/8$.

Let $\set{\mu_n,\phi_n}$  be an  eigenpair for  $H^{\mathcal T}$  in $I$.
By the argument identical to the one used earlier, its localization center $y_n$ is located either in $\tilde{\mathcal T}$ or in $\partial_{C \sqrt{\ell}} \mathcal T \subset \partial_{{\ell}/8} \mathcal T$. Hence
\be\label{eq:suppprzero'}
\| P_{J_\gamma}(H^{\mathcal{T}_\gamma}) - \chi_{\partial_{\ell/8}{\mathcal T}}  P_{J_\gamma}(H^{{\mathcal T}_\gamma}) \chi_{\partial_{\ell/8}{\mathcal T}} - \chi_{{\mathcal T}_{\ell}} P_{J_\gamma}(H^{{\mathcal T}_\gamma}) \chi_{{\mathcal T}_{ \ell}} \| \leq e^{-c \sqrt\ell},
\ee
which in particular establishes \eqref{eq:supppr'H}.  In fact, the above argument shows more, namely that, recalling the notation in  \cref{it:h3},
\be\label{eq:decprell}
\norm{P_{J_\gamma}(H^{{\mathcal T}_\gamma})\delta_x}\leq e^{-c \sqrt\ell} \mbox{ for any } x\notin \mathcal A_o.
\ee 
The latter bound will be of use to us momentarily.
\end{proof}
This completes the proof that $H^{\mathbb{T}}$ possesses a local gap structure in the sense defined by Theorem \ref{thm:hypa}. Using perturbation theory,  we are now going to show that   $H^{\mathbb{T}}(s)$ possesses a local gap structure as well.

\subsection{Proof of Theorem \ref{thm:hypa}}\label{sec:proofhyp}
It suffices to establish the assertion for $a=c_1$ as probabilities only improve as the system size decreases. 
We note that Proposition \ref{prop:sup} is applicable here with $I=c \ell^{-\xi}$. In particular, for $\omega \in \Omega_G$, we have  $\dist\pa{E^\pm_\gamma,\sigma(H^{\mathcal{T}_\gamma})}\ge \Delta$.
Let 
\[\tilde H^{\mathcal T_\gamma}(s):=H^{\mathcal T_\gamma}(0)+ P_{[E^-_\gamma,E^+_\gamma]}\pa{H^{\mathcal T_\gamma}(0)}+\beta W(s).\] Then, for $\ell$ sufficiently small  
\be\label{eq:emint}
\sigma\pa{\tilde H^{\mathcal T_\gamma}(s)}\cap\pa{\left[-\tfrac\Delta3,\tfrac\Delta3\right]+[E^-_\gamma,E^+_\gamma]}=\emptyset,
\ee
provided that $\beta<\frac{\Delta}{6}$. 

For the next assertion, we recall the definition of a dilation and its norm, introduced in \eqref{eq:distf}--\eqref{eq:Kcell}.
\begin{lemma}\label{lem:decGhp}
There exists  $c>0$ such that for any $z\in\mathbb C$ with   $Re(z)=E^\pm_\gamma$ and for any $\beta<c {\Delta}{\ell^{-3d}}$, we have
\be\label{eq:decGhp}
\norm{\pa{H^{\mathcal T_\gamma}(s)-z}^{-1}}_{c,\ell}+\norm{\pa{\tilde H^{\mathcal T_\gamma}(s)-z}^{-1}}_{c,\ell}\le C {\ell^{3d}} {\Delta^{-1}\langle Im\, z\rangle}^{-1},
\ee
where $\norm{\cdot}_{c,\ell}$ is  defined with  $\caA=\caA_o$.
\end{lemma}
\begin{proof}
 If we denote 
\be\label{eq:R_z}
R^o_{z}=\pa{ H^{\mathcal T_\gamma}(0)-z}^{-1},\  \tilde R^o_{z}=\pa{\tilde H^{\mathcal T_\gamma}(0)-z}^{-1},\ R_{z}=\pa{ H^{\mathcal T_\gamma}(s)-z}^{-1},\ \tilde R_{z}=\pa{\tilde H^{\mathcal T_\gamma}(s)-z}^{-1},\ee
we have 
\be\label{eq:fdil}
\norm{ R^o_{z}}_{c,\ell}\le C {\ell^{3d}} {\Delta^{-1}\langle Im\, z\rangle}^{-1}
\ee 
by  \eqref{eq:corprep}. 

 We  now expand $ R_{z}$ into the  Neumann series
\[ R_{z}=  R_{z}^o\sum_{n=0}^\infty\beta^n\pa{-W R_{z}^o}^n,\]
yielding, via \eqref{eq:multipl},
\begin{multline}\label{eq:Rza}
\norm{ R_{z}}_{c,\ell}\le \norm{ R_{z}^o}_{c,\ell}\sum_{n=0}^\infty\beta^n\norm{WR^o_{z}}_{c,\ell}^n \\ \le C {\ell^{3d}} {\Delta^{-1}\langle Im\, z\rangle}^{-1}\sum_{n=0}^\infty {\pa{\beta C\ell^{3d}}^n} \Delta^{-n}\le  C {\ell^{3d}} {\Delta^{-1}\langle Im\, z\rangle}^{-1},
\end{multline}
provided $\beta\le c{\Delta}{\ell^{-3d}}$.

Using \eqref{eq:decprell}, we deduce that 
\be\label{eq:R_o}
\norm{\e^{c\, \rho^\ell_{\mathcal A}}\,P_{[E^-_\gamma,E^+_\gamma]}\pa{H^{\mathcal T_\gamma}_o}}\le C\ell^d.
\ee
 Since 
\be\label{eq:R_a}
\tilde R^o_{z}= R^o_{z}-P_{[E^-_\gamma,E^+_\gamma]}\pa{H^{\mathcal T_\gamma}_o}R^o_{z}\tilde R^o_{z},
\ee
we obtain, using \eqref{eq:fdil}--\eqref{eq:R_o} and  
\[\norm{R^o_{z}}\le C\Delta^{-1}\langle Im\, z\rangle^{-1}, \quad \norm{\tilde R^o_{z}P_{[E^-_\gamma,E^+_\gamma]}\pa{H^{\mathcal T_\gamma}_o}}\le 2,\]
that
\[\norm{\tilde R^o_{z}}_{c,\ell}\le  C {\ell^{3d}} {\Delta^{-1}\langle Im\, z\rangle}^{-1}.\]

 We  now expand $\tilde R_{z}$ into the  Neumann series
\[\tilde R_{z}= \tilde R_{z}^o\sum_{n=0}^\infty\beta^n\pa{-W\tilde R_{z}^o}^n,\]
and repeat the argument in \eqref{eq:Rza} to complete the proof.
\end{proof}
We are now ready to finish the proof. For this, we will show that  conditions \ref{it:h1}--\ref{it:h3} in Theorem \ref{thm:hypa}  hold on $\Omega_G$, ensuring the desired probability for these events.

We first note that \cref{it:h1} follows from Proposition \ref{prop512} (with $I=c \ell^{-\xi}$) by standard perturbation theory for allowable values of $\beta$. On the other hand,  \cref{it:h3} is a direct consequence of Lemma \ref{lem:decGhp}. 

This leaves us with the task of  proving \cref{it:h2}. We recall that $J_\gamma=[E^-_\gamma, E^+_\gamma]$ and  set $\hat J_\gamma=\left[-\tfrac\Delta8,\tfrac\Delta8\right]+ [E^-_\gamma, E^+_\gamma]$. 
 We will abbreviate $P_\gamma:=P_{J_
\gamma}\pa{H^{\mathcal T_\gamma}(s)}$ and suppress the $s$-dependence for this argument, indicating by the  subscript (or superscript) $o$ the value $s=0$, if needed. We use  the decomposition \eqref{eq:conre} with $E_1 = E^-_\gamma$ and $E_2 = E^+_\gamma$ to write 
\be\label{eq:conP}
 P_\gamma=-\pa{2\pi}^{-1}\int_{-\infty}^\infty \sum_{j=1}^2(-1)^j R_{iu+E_j}  du.
\ee
We note that the integrand can be bounded, using  \cref{it:h1},  by
\be\label{eq:trnormbnew}
\max_{j=1,2}\norm{R_{iu+E_j} }\le   {\Delta^{-1}\langle u\rangle ^{-1}}, \quad u\in\mathbb R.
\ee
Using (recall \eqref{eq:R_z})
\[R_{iu+E_j}= \tilde R_{iu+E_j} - \tilde R_{iu+E_j}P_{J_\gamma}({H^{\mathcal T_\gamma}_o})R_{iu+E_j}
\]
and
\[\int_{-\infty}^\infty \sum_{j=1}^2(-1)^j \tilde R_{iu+E_j}du=0,\]
which holds thanks to \eqref{eq:emint},
we conclude that $ P_\gamma$ is equal to
\begin{equation}
\label{eq:proj_perturbation}
 \pa{2\pi}^{-1}\sum_{j=1}^2(-1)^j \int_{-\infty}^\infty \tilde R_{iu+E_j}P_{J_\gamma}({H^{\mathcal T_\gamma}_o})R_{iu+E_j}\hat P_\gamma du.
\end{equation}
Hence we can bound
\begin{multline}
\label{eq:P_zone_est}
\norm{\e^{\frac c{\sqrt\ell}\, \rho_{\mathcal A}}P_\gamma}\le \int_{-\infty}^\infty\max_j\pa{\norm{\tilde R_{iu+E_j}}_{c,\ell}}  \norm{\e^{\frac c{\sqrt\ell}\, \rho_{\mathcal A}}P_{J_\gamma}({H^{\mathcal T_\gamma}_
o})}\norm{R_{iu+E_j}\hat P_\gamma} \\
\le C\ell^{4d}\Delta^{-2}\int_{-\infty}^\infty {\langle u\rangle^{-2}}du \le C\ell^{4d}\Delta^{-2},
\end{multline}
where we have used Lemma \ref{lem:decGhp}, \eqref{eq:R_o}, and \eqref{eq:trnormbnew} in the second step. 

By perturbation expansion for the resolvent and \eqref{eq:conP}, we  have
$$
P_\gamma = P_{J_\gamma}(H_o^{\mathcal T_\gamma}) -\pa{2\pi}^{-1} \int_{-\infty}^\infty \sum_{j=1}^2 \sum_{n=1}^\infty \beta^n R^o_{iu + E_j} (-W R^o_{iu + E_j})^n.
$$
We first observe that, due to \eqref{eq:suppprzero'}, $
\norm{\chi_{\partial_{\ell}{\mathcal T}}  P_{J_\gamma}(H_o^{\mathcal T_\gamma}) \chi_{{\mathcal T}_{8\ell}}} \leq e^{-c \sqrt{\ell}}
$.

Next, letting  
$\mathcal A_o=  (\mathcal{T}_\gamma)_{8 \ell}$,  we can estimate, using Lemma \ref{lem:decGhp} and \eqref{eq:multipl}, that
\[
\begin{aligned}
\norm{\chi_{\partial_{\ell}{\mathcal T}}  R^o_{iu + E_j} (W R^o_{iu + E_j})^n \chi_{{\mathcal T}_{8\ell}}}&\le C^n\norm{\chi_{\partial_{\ell}{\mathcal T}} \e^{-\frac c{\sqrt\ell}\, \rho_{\mathcal A_o}}}\norm{R^o_{iu + E_j}}_{c,\ell}^{{n+1}}\\ & \le C^n {\ell^{3dn}} {\Delta^{-n}\langle Im\, z\rangle}^{-2}e^{-c \sqrt{\ell}}.
\end{aligned}
\]
Hence
\[
\begin{aligned}
\norm{\chi_{\partial_{\ell}{\mathcal T}} \sum_{n=1}^\infty \beta^n R^o_{iu + E_j} (-W R^o_{iu + E_j})^n\chi_{{\mathcal T}_{8\ell}}}&\le e^{-c \sqrt{\ell}}\langle Im\, z\rangle^{-2}\sum_{n=1}^\infty \beta^n C^n {\ell^{3dn}} \Delta^{-n}\\&\le e^{-c \sqrt{\ell}}\langle Im\, z\rangle^{-2}.
\end{aligned}
\]
Integrating over the $u$ variable, we see that
$
\norm{\chi_{\partial_{\ell}{\mathcal T}}  P_I^\gamma \chi_{{\mathcal T}_{8\ell}}} \leq e^{-c \sqrt{\ell}}
$
holds. 
Combining this bound with \eqref{eq:P_zone_est}, we get  \eqref{eq:supppr'}.

The proof of \eqref{eq:supppr} is essentially identical to the one above, and so is left out.

\section{Adiabatic theory for localized spectral patches}\label{sec:patch}
Throughout this section we continue to work on a torus, in the setting of Theorem~\ref{thm:hypa}. To simplify the notation, we will shorthand $H(s):=H^{\mathbb T}(s)$ in this section.

We note that for  $\beta,\epsilon,\ell$ satisfying \eqref{eq:relations'}--\eqref{eq:scales'} and the exponents  in \eqref{eq:relations1} and  \eqref{eq:scales},  the conditions $\epsilon,\beta\ll1$ imply that $\epsilon^{-1} e^{-c \sqrt{\ell}} \leq e^{-c \sqrt{\ell}}$ and $\epsilon/\Delta \ll 1$.
We will use this repeatedly. We will also assume that $1\ge\Delta\ge\beta>0$ (in fact, \eqref{eq:relations'}--\eqref{eq:scales'} imply $\Delta\gg\beta$ for large $\ell$, but this will only matter later on).
\subsection{Kato's operator}
 In this subsection, we will consider the general adiabatic framework, keeping the notation consistent with that in \eqref{eq:H(t)}. Let $1\ge\Delta\ge\beta>0$ and let $H(s)$ be a smooth family of self-adjoint operators on $[0,1]$ such that
\begin{assumption}\label{as:Kato}
  \begin{enumerate}[label=(\it{\alph*})]
\item\label{it:n1}
 $\norm{H(s)}\le  C$ and $\norm{H^{(k)}(s)}\le \beta C_k$ for $k\in\mathbb N$, where $H^{\pa{k}}(s)$ stands for the $k$-th derivative of $H(s)$ with respect to the $s$ variable;
\item\label{it:n2} There exist $E_{1,2}\in\mathbb R$  and $\Delta>0$ such that $\min_{s\in[0,1]}\dist\pa{\sigma(H(s)),\set{E_1,E_2}}\ge 2\Delta$;
\item\label{it:n3} $H^{(k)}(s)=0$ for $s=\set{0,1}$ and $k\in\mathbb N$.
\end{enumerate}
\end{assumption}
Throughout this section, we will   denote by $P(s)$ the spectral projection of $H(s)$ onto the interval $[E_1,E_2]$ and will use the shorthand $R_z(s)$ for $\pa{H(s)-z}^{-1}$. For an operator $A$ (which can be $s$-dependent) we define the operator $X_{A}(s)$ by
\be\label{eq:X_Adef}
X_{A}(s)= \frac 1{2\pi}\sum_{j=1}^2(-1)^j\int_{-\infty}^\infty R_{ix+E_j}(s)\,A\,R_{ix+E_j}(s)\,dx.
\ee
This operator was introduced by Kato in his work on the adiabatic theorem, and henceforth we will refer to it as Kato's operator.

We note that, for $H(s)$ satisfying Assumption \ref{as:Kato},
\be\label{eq:resnorm}
\max_{j=1,2}\norm{R_{ix+E_j}(s)}\le \pa{x^2+\Delta^2}^{-1/2}
\ee
and consequently 
\be\label{eq:X_A}
\norm{X_{A}(s)}\le \frac {\norm{A}}{\pi}\int_{-\infty}^\infty \pa{x^2+\Delta^2}^{-1}dx\le  \Delta^{-1}{\norm{A}}.
\ee
Using the Leibniz rule and  \eqref{eq:Rdera}, it is straightforward to see that, more generally,
\be\label{eq:X_Ader}
\norm{X_{A}^{\pa{k}}(s)}\le  C_{k}\norm{A}_{k} , \quad k\in\mathbb Z_+,
\ee
where  $\norm{\cdot}_{k}$ denotes the Sobolev-type norm
\be
\norm{A}_{k}=\sum_{j=0}^k\norm{{A}^{\pa{j}}(s)}.
\ee
The importance of Kato's operator is related to the fact that it solves the commutator equation
\be\label{eq:commKato}
[H(s),X_{A}(s)]=[P(s),A],
\ee
which plays a role in a construction of adiabatic theory for gapped Hamiltonians, particularly in the Nenciu's expansion presented below.

To handle the adiabatic behavior of localized spectral patches, we will also need to understand the locality properties of Kato's operator. 
\begin{lemma}\label{lem:dercomm}
Let $A(s)$ be a smooth family of operators on $[0,1]$.  Suppose that in addition to Assumption \ref{as:Kato}, there exists  some set $\mathcal{A}$ and $M,c>0$ such that
\be\label{eq:locresN}
\norm{R_{ix+E_j}(s)}_{c,\ell} \leq  M\langle x\rangle^{-1},\quad j=1,2.
\ee

Then, 
\be\label{eq:locpab}
\norm{\e^{{ {c}}\rho^\ell_{\mathcal A}} {X_{A(s)}^{\pa{1}}(s)}}\le C \pa{\beta M^2\abs{\ln\Delta}+\beta M\Delta^{-1}}\norm{\e^{{c}\rho^\ell_{\mathcal A}} {A(s)}}+C M\abs{\ln\Delta}\norm{\e^{ {c} \rho^\ell_{\mathcal A}}  A^{\pa{1}}(s)}.
\ee 
\end{lemma}
\begin{proof}
We will suppress the $s$-dependence in the proof below.
Using \eqref{eq:Rdera}  and \eqref{eq:multipl}, we can bound
\begin{eqnarray*}
 \norm{\e^{{ {c}}\rho^\ell_{\mathcal A}} X_{A}^{\pa{1}}}&\le&  \sum_{j=1}^2 \left(\frac {C\beta}{\pi}\int_{-\infty}^\infty \norm{R_{ix+E_j}}^2_{ c,\ell} \norm{\e^{ {c} \rho^\ell_{\mathcal A}}A}\norm{R_{ix+E_j}}dx\right.\\ && \hspace{1.8cm} + 
 \frac {1}{\pi}\int_{-\infty}^\infty \norm{R_{ix+E_j}}_{ c,\ell} \norm{\e^{ {c} \rho^\ell_{\mathcal A}} A^{\pa{1}}}\norm{R_{ix+E_j}}dx\\&& \hspace{1.8cm}+\left.
 \frac {C\beta}{\pi}\int_{-\infty}^\infty \norm{R_{ix+E_j}}_{ c,\ell} \norm{\e^{ {c} \rho^\ell_{\mathcal A}}A}\norm{R_{ix+E_j}}^2dx\right).
\end{eqnarray*}
Using \eqref{eq:locresN} and Assumption  \ref{as:Kato}.\ref{it:n2}, we get \eqref{eq:locpab}.
\end{proof}

\subsection{Nenciu's expansion}\label{sec:Nenciu}
An elegant approach for the analysis of the adiabatic behavior of gapped systems was discovered by Nenciu \cite{Nenciu}. We will use it as a starting point for our construction. 
\begin{lemma}[Nenciu's expansion]\label{lem:Nenciu}
Let  $H(s)$ be a smooth family of self-adjoint operators on $[0,1]$  that satisfies Assumption \ref{as:Kato}. Let $B_n(s)$  be a smooth family defined recursively as follows:  $B_0(s)=P(s)$ and, for $n\in\N$,
\be\label{eq:Nexp}
B_n(s)= \pa{\bar P(s) X_{\dot B_{n-1}(s)}(s)P(s)+h.c.}+S_n(s)-2P(s)S_n(s)P(s),
\ee
where 
\be
S_n(s)=\sum_{j=1}^{n-1}B_j(s)B_{n-j}(s).
\ee
We then have 
\begin{thmlist}
\item\label{it:1l1} 
\be\label{eq:maincomm'}
\dot B_n(s)=-i\left[H(s), B_{n+1}(s) \right]
\ee
for all $n\in \Z_+$;
\item\label{it:1l2}  $B_n(s)=0$ for $s=\set{0,1}$ and $n\in\mathbb N$;
\item\label{it:1l3}  We have
\be\label{eq:B_nbnd}
\sup_s\norm{B_n^{\pa{k}}(s)}\le C_{n,k} \Delta^{-n} , \quad k,n\in\mathbb Z_+.
\ee
\end{thmlist}
\end{lemma}

\begin{proof}
Property \ref{it:1l1} is due to Nenciu, \cite{Nenciu}. Property \ref{it:1l2}  follows directly from the recursive definition of $B_n$s. We establish \ref{it:1l3}   by induction: 

{\it Induction base:} For $n=0$ and an arbitrary $k$, the bound $\norm{B_0^{\pa{k}}(s)}\le C_{k}$ in \ref{it:1l3} can be seen from \eqref{eq:Rder}, \eqref{eq:Rdera}, Assumption \ref{it:n1},  and the Leibniz rule.

{\it Induction step:} Suppose now that the statement holds for all $n< n_o$ and all $k\in\Z_+$. Differentiating  \eqref{eq:Nexp} $k$ times with $n=n_o$ using the Leibniz rule and then using \eqref{eq:resnorm} and \eqref{eq:X_Ader}, we get that it also holds for $n=n_o$ and all $k\in\Z_+$. 
\end{proof}

For localized spectral patches, we slightly modify the statement.

\begin{lemma}\label{lem:HQDel}
Suppose that in addition to the assumptions of Lemma \ref{lem:Nenciu}, there exists  some set $\mathcal{A}$ and $M,c>0$ such that \eqref{eq:locresN} holds.
Let us also assume that   
\be\label{eq:locHN}
\max_{s\in[0,1]}\norm{\e^{c\, \rho^\ell_{\mathcal A}}P(s)}\le C,\quad \max_{s\in[0,1]}\norm{H^{\pa{k}}(s)}_{c,\ell}\le C_k\beta \mbox{ for } k\in\mathbb N.\ee
Let 
\[ \nu=\min\pa{M^{-1}\abs{\ln\Delta}^{-1},\Delta} ,\] 
and assume that $\beta\le\nu$. 
 Then the operators $B_n$ defined in Lemma \ref{lem:Nenciu}  satisfy 
\be
 \norm{\e^{c\rho^\ell_{\mathcal A}} { B^{\pa{k}}_n(s)}}\le C_{n,k} \nu^{-n},\quad k,n\in\mathbb Z_+.
\ee
\end{lemma}

\begin{proof}
We will suppress the $s$-dependence in the proof and use induction in $n$ and $k$. 

{\it Induction base:} For $n=0$ and arbitrary $k$, by the Leibniz rule we have
\be\label{eq:manyder} P^{\pa{n}}=\pa{P^{n+1}}^{\pa{n}}=\sum _{k_{1}+k_{2}+\cdots +k_{n+1}=n}{n \choose k_{1},k_{2},\ldots ,k_{n+1}}\prod _{1\leq j\leq n+1}P^{(k_{j})},\ee
where   the sum extends over all m-tuples $(k_1,\ldots,k_{n+1})$ of non-negative integers satisfying $ \sum _{j=1}^{n+1}k_{j}=n$ (so that for at least one value of $j$ we have $k_j=0$). 

Using the integral representation \eqref{eq:Rder}, the formula
\eqref{eq:Rdera}, the Leibniz rule, \eqref{eq:locresN}, \eqref{eq:multipl}, and Assumption \eqref{eq:locHN}, we can bound
\[\norm{P^{(k)}}_{c,\ell}\le C_kM^k,\quad k\in\mathbb N.\]

We can now use \eqref{eq:multipl} and \eqref{eq:manyder}  to deduce that
\begin{multline}
\norm{\e^{c\, \rho^\ell_{\mathcal A}}P^{(n)}} \\ \le \sum _{k_{1}+k_{2}+\cdots +k_{n+1}=n}{n \choose k_{1},k_{2},\ldots ,k_{n+1}}\prod _{1\leq j<j_o}\norm{P^{(k_j)}}_{c,\ell}\,\norm{\e^{c\, \rho^\ell_{\mathcal A}}P}\,\prod _{j_o\leq j\le n+1}\norm{P^{(k_j)}}\\ \le \sum _{k_{1}+k_{2}+\cdots +k_{n+1}=n}{n \choose k_{1},k_{2},\ldots ,k_{n+1}}\prod _{1\leq j\le n+1}  C_{k_j} M^{k_j} = C_n  M^{n},
\end{multline}
where $j_o$ is the first value of the index $j$ for which $k_j=0$.

{\it Induction step:}  Suppose now that the assertion holds for all $n<n_o$ and all $k$. Differentiating  \eqref{eq:Nexp} $k$ times with $n=n_o$ using the Leibniz rule and then using Lemma  \ref{lem:dercomm} (the assumption there  is satisfied by \cref{eq:expdeccor}), we get the induction step.
\end{proof}
\subsection{Gapped adiabatic theorem}
An immediate consequence of Lemma~\ref{lem:Nenciu} is  
\begin{lemma}[Gapped adiabatic theorem to all orders]
\label{lemma:all_orders}
In the setting of Lemma~\ref{lem:Nenciu}, let $P_N(s) := \sum_{n=0}^N \epsilon^n B_n(s)$. 
Then for all $N \in \mathbb{N}$,
\[
\norm{U_\epsilon(s) P(0) U_\epsilon(s)^* - P_N(s)} \leq C_N \epsilon^N \Delta^{-N },
\]
where $U_\epsilon$ was defined in \eqref{eq:Ueps}.

In particular, for $\epsilon<\Delta$, we have
\[
\norm{U_\epsilon(s) P(0) U_\epsilon(s)^* - P(s) } \leq C \epsilon \Delta^{-1}
\]
and 
\[
\norm{U_\epsilon(1) P(0) U_\epsilon(1)^* - P(1) } \leq C_N \epsilon^N \Delta^{-N }.
\]
\end{lemma}
\begin{proof}
By Lemma~\ref{lem:Nenciu},
\[
\epsilon \dot{P}_N(s) = -i [H(s), P_N(s)] + \epsilon^{N+1} \dot{B}_N(s).
\]
Using the fundamental theorem of calculus, we obtain
\[
U_\epsilon(s)^* P_N(s) U_\epsilon(s) - P_N(0) = \epsilon^{-1} \int_0^s \epsilon^{N+1} \frac d{ds}\pa{U_\epsilon(s)^* {B}_N(s) U_\epsilon(s)}.
\]
Using the unitarity of $U_\epsilon$, Assumption \ref{as:Kato}, and Lemma \ref{it:1l3}, we obtain 
\[
\| U_\epsilon(s)^* P_N(s) U_\epsilon(s) - P_N(0) \| \leq C_N \epsilon^{N} \Delta^{-N}.
\]
The assertion follows from $P_N(0) = P(0)$, $\norm{P_N(s)-P(s)} \le C\epsilon \Delta^{-1}$, and $P_N(1) = P(1)$. 
\end{proof}
\subsection{Adiabatic theorem for a localized spectral patch}
The goal of this subsection is to prove the following assertion, which is of independent interest.
\begin{thm}[Local adiabatic theorem on a torus]\label{thm:AT-single}
Suppose that  the family $H(s)$ satisfies Assumption \ref{hyp1} and $H(0)$ satisfies Assumptions \ref{assump:FRC}--\ref{assump:FMC}.  Let $\caG_\omega$ be the event that $H^{\mathbb{T}}(0)$  possesses  a local gap structure for the energy interval $J = (E - 6 \delta, E + 6 \delta)$ in the sense of Definition \ref{def:ls}. Then $\mathbb P\pa{\caG_\omega} > 1 - e^{- c \sqrt{\ell}}$. Moreover, for each $\omega\in \caG_\omega$,   the physical evolution $\psi_\epsilon(s)$ of each eigenvector $\psi=\psi_n$ with $E_n\in J$ given by  \eqref{eq:adevol}, satisfies
\be\label{eq:AT_sing}
\max_{s\in[0,1]}\norm{\bar P_{J_\gamma}\hspace{-.1cm}\pa{H^{{\mathcal T}_\gamma}(s)}\psi_\epsilon(s)}\le C \pa{ \epsilon \Delta^{-1} + e^{- c \sqrt{\ell}} }
\ee
for some $\gamma$. 
For any $N\in \N$, we can further improve \eqref{eq:AT_sing} for $s=1$:
\be\label{eq:AT_sing1}
\norm{\bar P_{J_\gamma}\hspace{-.1cm}\pa{H^{{\mathcal T}_\gamma}(1)}\psi_\epsilon(1)}\le C_N \pa{\epsilon^N\pa{ \Delta^{-N} +  \delta^{-2N-1}}   +   e^{- c \sqrt{\ell}}}.
\ee
\end{thm} 
\begin{proof}[Proof of  Theorem \ref{thm:AT-single}]
The first part of Theorem \ref{thm:hypa} has already been established. We now show the second part. 
We first note that  $\caG$ is a subset of $\Omega_{loc,N}$, the portion of the  configuration space for which $\mathbb{T}$ and all sets in $\set{\caT_\gamma}$ are $\ell/10$-localizing, see  Lemma \ref{clustloc} below. Thus,  \cref{it:h2} implies the existence of a patch $\mathcal T_\gamma$ such that  $
\norm{\bar \chi_{ \pa{\mathcal T_\gamma}_{8\ell}}\psi}\leq e^{-c \sqrt\ell}$. It then follows from Lemma \ref{outbad} below, specifically \eqref{eq:distmu}, that $E\in J_\gamma$ (see also  \eqref{eq:locgap}). Let $\hat{\mathcal{T}}_\gamma =  \pa{\mathcal T_\gamma}_{4\ell}$ and set 
\be\label{eq:Q_gamma}
Q_\gamma(s) = \chi_{\hat{\mathcal T}_\gamma}  P_{J_\gamma}(H^{{\mathcal T}_\gamma}(s)) \chi_{\hat{\mathcal T}_\gamma}.
\ee

 By   Lemma \ref{outbad}, specifically \eqref{eq:rangmu}, we know that \eqref{eq:AT_sing1} holds for $s=0$ (with $\epsilon=0$ on the right hand side). Let $\rho:=Q_\gamma(0)$ be the (truncated) initial spectral patch. Then, since 
 \[\bar \rho= \chi_{\hat{\mathcal T}}  \bar P_{J_\gamma}(H^{{\mathcal T}_\gamma}(0)) \chi_{\hat{\mathcal T}}+\bar \chi_{\hat{\mathcal T}},\] we deduce that  $\norm{\bar \rho \psi}\le    e^{- c \sqrt{\ell}}$.
  Hence, by the unitarity of the quantum evolution,
\be
\norm{\bar \rho_\epsilon(s) \psi_\epsilon(s)}\le    e^{- c \sqrt{\ell}}
\ee
for all $s$, where $\rho_\epsilon$ denotes the (full) Heisenberg evolution of the (truncated) initial spectral patch $\rho:=Q_\gamma(0)$, i.e., 
\be\label{eq:IVPrho}
i\epsilon \dot \rho_\epsilon(s) =[H(s), \rho_\epsilon(s)],\quad \rho_\epsilon(0)=\rho.
\ee

Therefore the result  follows from
\begin{lem}\label{lem:Qcomp}
\begin{thmlist}
\item\label{it:sppa} We can estimate \be\label{eq:AT_patch}
\max_{s\in[0,1]}\norm{\rho_\epsilon(s)-Q_\gamma(s)}\le C \pa{ \epsilon \Delta^{-1} + e^{- c \sqrt{\ell}} }.
\ee
Moreover, for any $N\in \N$, we have 
\be\label{eq:AT_patch1}
\max_{s=\set{0,1}}\norm{\rho_\epsilon(s)-Q_\gamma(s)}\le C_N \pa{\epsilon^N\pa{ \Delta^{-N} +  \delta^{-2N-1}}   +   e^{- c \sqrt{\ell}}}.
\ee

\item\label{it:sppa1} 
In addition, 
\be\label{eq:Qcomp}
\max_{s\in[0,1]}\norm{ \bar P_{J_\gamma}(H^{{\mathcal T}_\gamma}(s))- \bar P_{J_\gamma}(H^{{\mathcal T}_\gamma}(s))\,\bar Q_\gamma(s)}\le e^{- c \sqrt{\ell}}.
\ee
\end{thmlist}
\end{lem}
\end{proof}
\begin{rem}\label{rem:psirep}
We note that in the proof of Theorem \ref{thm:AT-single}, the initial spectral data $\psi_n$  can be replaced by any vector $\psi\in Ran(P_[E-\delta,E+\delta)$ that satisfies   $\norm{\bar \chi_{ \pa{\mathcal T_\gamma}_{8\ell}}\psi}\leq e^{-c \sqrt\ell}$ for some patch $\mathcal T_\gamma$. 
\end{rem}
\begin{proof}[Proof of Lemma \ref{lem:Qcomp}]
We suppress the $s$ dependence in the proof below. 
The property \eqref{eq:Qcomp} can be seen by decomposing \[\bar P_{J_\gamma}(H^{{\mathcal T}_\gamma})= \bar P_{J_\gamma}(H^{{\mathcal T}_\gamma})\,\bar Q_\gamma+\bar P_{J_\gamma}(H^{{\mathcal T}_\gamma})\, Q_\gamma\]
and noticing that
\[\begin{aligned}\bar P_{J_\gamma}(H^{{\mathcal T}_\gamma})\, Q_\gamma&=\bar P_{J_\gamma}(H^{{\mathcal T}_\gamma})\, \chi_{\hat{\mathcal T}_\gamma}  P_{J_\gamma}(H^{{\mathcal T}_\gamma}) \chi_{\hat{\mathcal T}_\gamma}\\ &=\bar P_{J_\gamma}(H^{{\mathcal T}_\gamma})\,  P_{J_\gamma}(H^{{\mathcal T}_\gamma}) \chi_{\hat{\mathcal T}_\gamma}+O\hspace{-1mm}\pa{e^{- c \sqrt{\ell}}}=O\hspace{-1mm}\pa{e^{- c \sqrt{\ell}}},\end{aligned}\]
by \eqref{eq:supppr'}.

\cref{it:sppa}: By our assumption, $H^{\mathcal{T}_\gamma}$ is a gapped Hamiltonian with gap $\Delta$. Following the argument in Section~\ref{sec:Nenciu}, we denote by  $B_n^\gamma$ the n-th order in Nenciu's expansion and use Lemma~\ref{lem:Nenciu} with $B_0^\gamma = P_{J_\gamma}(H^{\mathcal{T}_\gamma})$. We set
\be\label{eq:Qgam}
Q_{\gamma,N}:= \sum_{n=0}^N \epsilon^n\chi_{\hat{\mathcal T}}B^\gamma_n\chi_{\hat{\mathcal T}}.
\ee
and proceed to show that 
\be\label{eq:Qepsco}
\max_{s}\norm{\rho_\epsilon-Q_{\gamma,N}}\le C_N \pa{\epsilon^N\pa{ \Delta^{-N} +  \delta^{-2N-1}}   +   e^{- c \sqrt{\ell}}}.
\ee
The result  then follows immediately from \eqref{eq:Qepsco} by the definition of $Q_{\gamma,N}$ and Lemma \ref{it:1l2}--\ref{it:1l3} (we recall that $B^\gamma_0= P_{J_\gamma}(H^{{\mathcal T}_\gamma})$).

To get \eqref{eq:Qepsco}, we  observe that by \eqref{eq:maincomm'},
\begin{align*}
\epsilon \dot Q_{\gamma,N}&= -i\sum_\gamma\sum_{n=0}^N \epsilon^{n+1}\chi_{\hat{\mathcal T}}\left[H^{ {\mathcal T}_\gamma}, B^\gamma_{n+1}\right]\chi_{\hat{\mathcal T}} \\
					 &=-i[H,Q_{\gamma,N}]-i\epsilon^{N+1}\chi_{\hat{\mathcal T}}\dot B^\gamma_{N} \chi_{\hat{\mathcal T}}  \\& \quad + \pa{i\sum_\gamma\sum_{n=0}^N \epsilon^{n+1}\left[H^{ {\mathcal T}_\gamma}, \chi_{\hat{\mathcal T}}\right]B^\gamma_{n+1}\chi_{\hat{\mathcal T}}+h.c.},
\end{align*}
where we have used $ H^{{\mathcal T}}(s) \chi_{\hat{\mathcal T}}=H(s)\chi_{\hat{\mathcal T}}$. 
We bound the second term on the second line by 
$ C_N  \epsilon^{N+1} \Delta^{-N} $
 using \eqref{eq:B_nbnd}.  For the term on the third line, we note that
\[\norm{\left[H^{ {\mathcal T}_\gamma}(s), \chi_{\hat{\mathcal T}}\right]B^\gamma_{n+1}(s)}\le \nu^{-n-1} \e^{-c \sqrt{ \ell}}\]
using Lemma~\ref{lem:HQDel}. Putting these bounds together, we get 
\be\label{eq:maincomm''}
\norm{\epsilon \dot Q_{\gamma,N} + i[H, Q_{\gamma,N}]}\le C_N   \epsilon^{N+1}\Delta^{-N}  + C e^{- c \sqrt{\ell}}.
\ee
Finally,  we observe that 
\[\partial_s\pa{U_\epsilon(t,s) Q_{\gamma,N}(s)U_\epsilon(s,t)}=\epsilon^{-1} U_\epsilon(t,s)\pa{\epsilon\dot Q_{\gamma,N}(s)+i[H(s), Q_{\gamma,N}(s)]}U_\epsilon(s,t).\]
where $U_\epsilon(t,s)$ was defined in \eqref{eq:Ueps}.

Integrating over $s$ and using \eqref{eq:maincomm''}, we deduce that
\be\label{eq:disQ}
\norm{U_\epsilon(t,r) Q_{\gamma,N}(r)U_\epsilon(r,t)- Q_{\gamma,N}(t)}\le  \epsilon^{-1} \pa{C_N   \epsilon^{N+1}\Delta^{-N}  + C e^{- c \sqrt{\ell}}},
\ee
We now note that $Q_{\gamma,N}(0)=\rho$, so $U_\epsilon(t,0) Q_{\gamma,N}(0)U_\epsilon(0,t)=\rho_\epsilon(t)$ by uniqueness of the solution for the IVP \eqref{eq:IVPrho}.  Combining this with \eqref{eq:disQ} yields \eqref{eq:Qepsco}.
\end{proof}
\subsection{Adiabatic theorem for a thin spectral set near $E$}\label{sec:Q}
In preparation for the proof of Theorem \ref{thm:AT-loc}, we will first investigate the adiabatic behavior of spectral data corresponding to a  thin set of non-trivial thickness that contains energy $E$. It will play the role of a natural barrier suppressing transitions  between the spectral data below and above $E$, which will make Theorem \ref{thm:AT-loc} applicable. The idea here is to combine the localized spectral patches near $E$ analyzed in the previous subsection into such a set. Specifically, we define
\be\label{eq:Qhatdel}
Q(s):=\sum_\gamma  Q_\gamma(s),
\ee
where the spectral patch $Q_\gamma$ was defined in \eqref{eq:Q_gamma}. Our first assertion encapsulates the basic properties of this operator.

\begin{lemma}\label{lem:smothapp} For $\ell$ large enough, the operator $Q(s)$ satisfies the following properties: 
\begin{thmlist}
\item\label{smoth1}  If $H(s)$ is $k$ times differentiable, so is $Q(s)$:
\[\max_{s\in[0,1]}\norm{\frac{d^jQ(s)}{d^js}}\le C_j\beta , \quad j=1,\ldots,k;
\]
\item\label{smoth2} Near commutativity with $H(s)$: 
\be
\norm{[H(s),Q(s)]}\le C e^{- c \sqrt{\ell}};
\ee
\item\label{smoth3}  Almost projection:
\be\label{eq:clospro}
\norm{\bar{Q}(s)Q(s)}\le C e^{- c \sqrt{\ell}};
\ee

\item\label{smoth4} Spectrally thin but with non-trivial thickness: 
Let $J_+ = (E -6 \delta, E + 6\delta)$, and $J_- = (E - \delta, E + \delta)$. Then
\be\label{eq:confesttil'}
\norm{ \bar{P}_{ J_+}(s)Q(s)}\le C e^{- c \sqrt{\ell}},
\quad  \norm{\bar{Q}(s)P_{ J_-}(s)} \le C e^{- c \sqrt{\ell}}.
\ee
\end{thmlist}
\end{lemma}
\begin{proof}
Lemma \ref{smoth1}: Note that, for $\ell$ large enough, $\beta \ll  \Delta$. The assertion follows from the integral representation \eqref{eq:Rder} for  $P_{J_\gamma}(H^{{\mathcal T}_\gamma}(s))$  with $E_{1,2}=E_\pm^\gamma$,  the formula
\eqref{eq:Rdera},  \eqref{eq:resnorm}, 
and  the Leibniz rule.

Lemma \ref{smoth2}: We compute 
\begin{align*}
\left[H(s),Q_\gamma(s) \right] &=\left[H^{{\mathcal T}_\gamma}(s),Q_\gamma(s)\right]\\ 
					     &=\left[H^{{\mathcal T}_\gamma}(s),\chi_{\hat{\mathcal T}}  \right] P_{J_\gamma}(H^{ {\mathcal T}_\gamma}(s)) \chi_{\hat{\mathcal T}}+\chi_{\hat{\mathcal T}}  P_{J_\gamma}(H^{ {\mathcal T}_\gamma}(s))\left[H^{ {\mathcal T}_\gamma}(s), \chi_{\hat{\mathcal T}}\right],
\end{align*}
and estimate both terms by $C \e^{- c \sqrt\ell}$ using Assumption \ref{hyp1} and  \cref{it:h2}.

Lemma \ref{smoth3}: We note that, for disjoint sets $\Omega_\gamma$,
\be\label{eq:normsum}
\| \sum_\gamma \chi_{\Omega_\gamma} A_\gamma \chi_{\Omega_\gamma} \| \leq \max_{\gamma} \| \chi_{\Omega_\gamma} A_\gamma \chi_{\Omega_\gamma}\|.
\ee
Since $\mathcal{T}_\gamma$ are disjoint, we have
\[
\norm{\bar{Q}(s)Q(s)} = \norm {\sum_\gamma{{\chi_{\hat{\mathcal T}}  P_{J_\gamma} (H^{ {\mathcal T}_\gamma}(s)) \bar{\chi}_{\hat{\mathcal T}}P_{J_\gamma}(H^{{\mathcal T}_\gamma}(s))}\chi_{\hat{\mathcal T}}}}.
\]
The right hand side is bounded by $C e^{-c \sqrt{\ell}}$ using \cref{it:h2}.

Lemma \ref{smoth4}:   We apply Lemma \ref{lem:PbarP} with  $H_1=H(s)$, $H_2=H^{{\mathcal T}}(s)$, and $R= \chi_{\hat{\mathcal T}}  $  to bound
\[\norm{\bar P_{ J_+}(s) \chi_{\hat{\mathcal T}}  P_{J}(H^{{\mathcal T}}(s)) }\le Ce^{-c \sqrt{\ell}},\]
where we have used \eqref{eq:supppr'} and the fact that $H(s)$ has range $r$. Since
\[Q(s)\le \chi_{\hat{\mathcal T}}  P_{J}(H^{{\mathcal T}}(s)) \chi_{\hat{\mathcal T}}\]
by \eqref{eq:locgap}, we deduce that
\[\norm{ \bar{P}_{ J_+}(s)Q(s)}\le \norm{ \bar{P}_{ J_+}(s)\chi_{\hat{\mathcal T}}  P_{J}(H^{{\mathcal T}}(s)) }\le C e^{- c \sqrt{\ell}}.\]

On the other hand, letting $J' = (E - 3\delta, E + 3\delta)$ and using Lemma \ref{lem:PbarP} with  $H_1=H^{{\mathcal T}}(s)$ and $H_2=H (s)$, we get
\[\norm{\bar P_{J'}(H^{{\mathcal T}}(s)) \chi_{\hat{\mathcal T}}P_{ J_-}(s)} \le C e^{- c \sqrt{\ell}}\]
Since
\[\bar Q(s)\le \chi_{\Lambda\setminus\hat{\mathcal T}}+\chi_{\hat{\mathcal T}}  \bar P_{J'}(H^{{\mathcal T}}(s)) \chi_{\hat{\mathcal T}}\]
by \eqref{eq:locgap}, we deduce that
\[\norm{\bar{Q}(s)P_{ J_-}(s)}\le  \norm{\chi_{\Lambda\setminus\hat{\mathcal T}}P_{ J_-}(s) }+\norm{ \bar{P}_{ J_+}(s)\chi_{\hat{\mathcal T}}  P_{ J_-}(s)}\le C e^{- c \sqrt{\ell}},\] 
using \eqref{eq:supppr} to bound the first term on the right hand side.

\end{proof}

One disadvantage of working with $Q$ is the fact that it is not a projection. We rectify this problem in the next assertion.
\begin{lemma}\label{lemma:AT-loc}

Let $N\in\mathbb N$. Suppose that $\ell$ is sufficiently large. Then there exists a smooth family of projections $Q_s$ with the following properties:
\begin{thmlist}
\item\label{it:1b}
 \be\label{eq:comQ}\max_{s\in[0,1]}\norm{[Q_s,H(s)]}\le C\pa{ \epsilon + e^{- c \sqrt{\ell}}}
 \ee
 and 
\be\label{eq:comQ1}\max_{s\in\set{0,1}}\norm{[Q_s,H(s)]}\le  C_N   {\epsilon^{N+1}}\Delta^{-N}  + C e^{- c \sqrt{\ell}};
\ee
\item\label{it:2b} Let $J_+ = (E -6 \delta, E + 6\delta)$ and $J_- = (E - \delta, E + \delta)$. Then
\be\label{eq:confesttil}
\max_{s \in [0,1]} \pa{\norm{\bar{P}_{ J_+}(s)Q_s}, \norm{\bar{Q}_sP_{ J_-}(s)}}\le  C\pa{  {\epsilon}\Delta^{-1} +   e^{- c \sqrt{\ell}}}
\ee
and 
\be\label{eq:confesttil''}
\max_{s \in \set{0,1}} \pa{\norm{\bar{P}_{ J_+}(s)Q_s}, \norm{\bar{Q}_sP_{ J_-}(s)}}\le  C e^{- c \sqrt{\ell}}
\ee
\item\label{it:3b}  $Q^{\pa{k}}_0=Q_1^{\pa{k}} = 0$ for all $k\in\Z_+$ and 
\[\max_{s\in[0,1]}\norm{Q^{\pa{k}}_s}\le C_k\beta,\quad k\in\N;\]

\item \label{it:4b}
\be\label{eq:maincomm}
\norm{\epsilon\dot Q_s + i[H(s), Q_s]}\le C_N   {\epsilon^{N+1}}\Delta^{-N}  + C e^{- c \sqrt{\ell}} ;
\ee
\item\label{it:6b}
If we denote by $Q_\epsilon(s)$ the solution of the IVP
$i \epsilon \dot Q_\epsilon(s)   = [H(s), Q_\epsilon(s)]$,  $Q_\epsilon(0) = Q_0$,
then we have
\be
\label{eq:adbnd}\max_{s\in[0,1]}\norm{Q_\epsilon(s)-Q_s}\le C_N   {\epsilon^{N}}\Delta^{-N}  +  C e^{- c \sqrt{\ell}}.
\ee
\end{thmlist}
\end{lemma}

\begin{proof}
We set
\be
Q_N(s):= \sum_\gamma Q_{\gamma,N}(s),
\ee
where $Q_{\gamma,N}$ was defined in \eqref{eq:Qgam}, 
and first show that the assertions of the lemma hold if we replace $Q_s$ with $Q_N(s)$ there. Note that the latter operator is not a projection.

It follows from Lemma \ref{lem:Nenciu} and the hypothesis $\epsilon \leq \Delta$ that 
\be
\norm{{Q}_N(s) -{Q}_0(s)}=\norm{{Q}_N(s)- Q(s)}\le C_N\,{\epsilon}\Delta^{-1}.
\ee
 Hence, combining this bound with Lemma \ref{lem:smothapp}, we conclude that ${Q}_N(s)$ satisfies the properties \ref{it:2b}--\ref{it:3b}.

We next observe that the property \ref{it:4b} holds for $Q_N(s)$ by \eqref{eq:maincomm''},  Assumption \ref{hyp1}, and \eqref{eq:normsum}.

The property \ref{it:6b} is established by replicating the argument employed in the proof of \cref{it:sppa}.

Finally, the property \ref{it:1b} holds for $Q_N(s)$ by the properties  \ref{it:3b}--\ref{it:4b} we already established.

\vspace{.3cm}

We now note that ${Q}_N(0)=Q(0)$.
Hence, defining $Q_\epsilon(t):=U_\epsilon(t,0)Q(0)U_\epsilon(0,t)$, we get $\norm{Q_\epsilon(t)\bar Q_\epsilon(t)}=\norm{Q(0)\bar Q(0)}\le C e^{- c \sqrt{\ell}}$ by \eqref{eq:clospro}. Thus, by the triangle inequality, we get 
\begin{align*}
\norm{{Q}_N(t)\bar{{Q}}_N(t)} &\le \norm{{Q}_N(t)\bar{ {Q}}_N(t)-Q_\epsilon(t)\bar Q_\epsilon(t)}+C e^{- c \sqrt{\ell}}\\  &\le \pa{\norm{\bar{{Q}}_N(t)}+\norm{Q_\epsilon(t)}}\norm{{Q}_N(t)-Q_\epsilon(t)} +C e^{- c \sqrt{\ell}} \\ &\le C_N   {\epsilon^{N}}\Delta^{-N}  +  C e^{- c \sqrt{\ell}},
\end{align*}
where in the last step we have used the properties  \ref{it:3b} and \ref{it:6b} for $Q_N$. 

It follows that 
\[\max_s\dist\pa{\sigma\pa{Q_N(s)},\set{0,1}} \le C_N   {\epsilon^{N}}\Delta^{-N}  + C e^{- c \sqrt{\ell}}.\] If $\epsilon/\Delta$ is small enough and $\ell$ large enough, the right hand side is smaller than $1/4$.  We set $Q_s$ to be the spectral  projection for ${Q}_N(s)$ onto the interval $[\tfrac12,\tfrac32]$. Then by functional calculus for self-adjoint operators and the triangle inequality, Lemma \ref{it:1b}, \ref{it:2b}, and \ref{it:6b} hold for this operator. To establish Lemma \ref{it:3b}, we use  the following integral representation for $Q_s$:
\be
\label{gamma_aux} Q_s=\pa{2\pi i}^{-1}\oint_\Gamma \pa{Q_N(s)-z}^{-1}dz,\quad \Gamma=\set{z\in\mathbb C:\ \abs{z-1}=1/2}.
\ee
Since 
\[\partial_s\pa{Q_N(s)-z}^{-1}=-\pa{Q_N(s)-z}^{-1}\partial_s Q_N(s)\pa{Q_N(s)-z}^{-1} ,\]
and $\norm{\pa{Q_N(s)-z}^{-1}}$ is uniformly bounded for $z \in \Gamma$, the property 
\ref{it:3b} follows by the Leibniz rule and the bounds on $Q_N^{(k)}(s)$. 

Lemma \ref{it:4b}:
\begin{multline*}\dot Q_s=-\pa{2\pi i}^{-1}\oint_\Gamma \pa{Q_N(s)-z}^{-1}\dot Q_N(s)\pa{Q_N(s)-z}^{-1}dz\\ = -i\pa{2\pi i}^{-1}\oint_\Gamma \pa{Q_N(s)-z}^{-1}[H(s), Q_N(s)]\pa{Q_N(s)-z}^{-1}dz\\ -\pa{2\pi i}^{-1}\oint_\Gamma \pa{Q_N(s)-z}^{-1}\pa{\dot{Q}_N(s)-i[H(s), Q_N(s)]}\pa{Q_N(s)-z}^{-1}dz,\end{multline*}
and the statement follows from the properties \ref{it:4b} and \ref{it:1b} already proved for $Q_N(s)$.

For $s\in \set{0,1}$, we have $Q_N(s) = Q(s)$, so \eqref{eq:comQ1} and \eqref{eq:confesttil''} follow from \cref{lem:smothapp}.

\end{proof}


\subsection{Adiabatic behavior of the distorted Fermi projection}\label{sec:locad}
The idea behind the proof of Theorem~\ref{thm:AT-loc} is that, since the projection $Q_s$ evolves adiabatically, it effectively induces a gap on its spectral support and decouples the energies separated by this induced gap. 

Let $\bar{H}(s) = \bar{Q}_s H(s) \bar{Q}_s$.  By Lemma~\ref{lemma:AT-loc}, $\bar{Q}_s$ is close to a spectral projection of $H(s)$ and so the spectrum of $\bar{H}(s)$ is approximately a subset of the original spectrum and the point $0$. To avoid discussing the position of $0$ with respect to $E$, we assume without loss of generality that $E <0$. We will need a pair of  preparatory results.
\begin{lemma}
\label{lem:gapped_approx}
Let $I=(E-\delta/2,E+\delta/2)$. Suppose that $\ell$ is large enough. Then we have
$\sigma(\bar{H}(s))\cap I=\emptyset$ for $s \in [0,1]$.  In addition, we have
\be\label{eq:smbarH}
\max_{s \in [0,1]} \norm{\bar{H}(s)^{\pa{k}}}\le C_k  \quad \mbox{ for } \quad k=1,\ldots,N.
\ee
\end{lemma}
\begin{proof} For $\ell$ large enough, $0 \notin I$. Hence, it is enough to show the claim when $\bar{H}(s)$ is understood as an operator on the range of $\bar{Q}_s$.
Let $w\in I$; we will show that $\pa{\bar{H}(s)-w}^2>0$, from which the assertion follows. To this end, we suppress the $s$-dependence and note that
\begin{align*}
\pa{\bar{H}-w}^2=\bar{Q}\pa{H-w}\bar{Q}\pa{H-w}\bar{Q}&=\bar{Q}\pa{H-w}^2\bar{Q}-\bar{Q}HQH\bar{Q}\\ &\ge\bar{Q}\bar P_{J-}\pa{H-w}^2\bar{Q}+\bar{Q}[H,{Q}][H,{Q}]\bar{Q},
\end{align*}
while we can bound
\[\bar{Q}\bar P_{J_-} \pa{H-w}^2\bar{Q}\ge \frac{\delta^2}4\bar{Q}\bar P_{J_-}\bar{Q}=\frac{\delta^2}4\bar{Q}-\frac{\delta^2}4\bar{Q} P_{J_-}\bar{Q}\ge \frac{\delta^2}4\bar{Q}-\frac{\delta^2}4\pa{C_N \epsilon +  C \exp\pa{- c \sqrt{\ell}}}^2\bar{Q},\]
using Lemma~\ref{lemma:AT-loc} \ref{it:2b}, and 
\[\bar{Q}[H,{Q}][H,{Q}]\bar{Q}\le \norm{[H,\bar{Q}]}^2\bar Q\le \pa{C_N \epsilon +  C \exp\pa{- c \sqrt{\ell}}}^2\bar Q\]
using Lemma~\ref{lemma:AT-loc} \ref{it:1b}.
Hence 
\[\pa{\bar{H}-w}^2\ge \pa{{\delta^2}/4-2\pa{C_N \epsilon +  C \exp\pa{- c \sqrt{\ell}}}^2}\bar Q>0\]
on $Ran\pa{\bar Q}$.

The bound \eqref{eq:smbarH} follows  from Lemma \ref{it:3b},  Assumption \ref{hyp1}, and the Leibniz rule.
\end{proof}
\begin{lemma}
\label{lem:smooth_dyn}
Let  $T(s,s')$ be the unitary semigroup generated by $ i [\dot{Q}_s, Q_s] $, i.e.,  $T(s,s')$ is the solution of the IVP
\be\label{eq:paralltr}
i \partial_s T(s,s') = i [\dot{Q}_s, Q_s] T(s,s'), \quad T(s',s')=1.
\ee
Then $T(s,s')$ satisfies
 \be\label{eq:interweaving}
 T(s,s')Q_{s'}=Q_sT(s,s').
 \ee
Suppose in addition that $\epsilon/\Delta$ is small enough and $\ell$ is sufficiently large. Then
 \be\label{eq:smoot}
\max_s \norm{T^{\pa{k}}(s,0)}\le C_k\beta \quad \mbox{ for } \quad  k=1,\ldots,N.
 \ee
\end{lemma}
\begin{proof}
The interweaving relation \eqref{eq:interweaving} follows from observing that 
\[\frac{d}{ds}\pa{T(s',s)Q_{s}T(s,s')}=T(s',s)\left[Q_{s},[\dot{Q}_s, Q_s]\right]T(s,s')+T(s',s)\dot Q_{s}T(s,s')=0,\]
and $T(s',s')Q_{s'}T(s',s')=Q_{s'}$.

The bound \eqref{eq:smoot} follows  from Lemma \ref{it:3b}, the unitarity of $T$, and the Leibniz rule.
\end{proof}

We now consider the evolution $U_\epsilon(s,s')$ generated by the equation
$$
i \epsilon \partial_s U_\epsilon(s,s') = H(s) U_\epsilon(s,s'), \quad U_\epsilon(s',s')=1.
$$

Let ${Q}^+_s$ (${Q}^-_s$) be the spectral projection of $\bar{H}_s$ associated with the interval $(E,\infty)$  ($(-\infty, E)$ respectively).

\begin{lemma}
\label{lem:dec}
Suppose that $\ell$ is large enough. Then we have
\be
\max_s\norm{ Q_1^+ U_\epsilon(s,0) Q_0^-} \leq  C\pa{  {\epsilon}\Delta^{-1} +   e^{- c \sqrt{\ell}}}
\ee
and
\be
\norm{ Q_1^+ U_\epsilon(1,0) Q_0^-} \leq  C_N \left( {\epsilon^{N}}\Delta^{-N}  +  {\epsilon^N}{\delta^{-2N-1}} \right) +   C e^{- c \sqrt{\ell}}.
\ee
\end{lemma}
\begin{proof}
We first note that Lemma~\ref{lemma:AT-loc}  implies that
\be\label{eq:offdb}
\norm{ Q_s U_\epsilon(s,s') \bar Q_{s'}} \leq C_N   {\epsilon^{N}}\Delta^{-N} +  C e^{- c \sqrt{\ell}}.
\ee
Indeed, using the semigroup property for $U_\epsilon$, 
\[
Q_s U_\epsilon(s,s') \bar Q_{s'}  =Q_s (Q_s - Q_\epsilon(s))U_\epsilon(s,s') -Q_s U_\epsilon(s,s') (Q_{s'} - Q_\epsilon(s')) , 
\]
and both terms on the right hand side can now be bounded using Lemma \ref{it:6b}.

Let $V_\epsilon(s) = \bar Q_s U_\epsilon(s,0) \bar Q_0$. Then a straightforward computation yields
\begin{align*}
i \epsilon \partial_s V_\epsilon(s) &= -i \epsilon \dot{Q}_s U_\epsilon(s,0) \bar Q_0 + \bar Q_s H(s) U_\epsilon(s,0)\bar Q_0 \\
						 &= i \epsilon [\dot{Q}_s, Q_s] V_\epsilon(s) + \bar{H}(s) V_\epsilon(s) + R_\epsilon(s),
\end{align*}
where 
$$
R_\epsilon(s) = -i \epsilon \dot{Q}_s Q_s U_\epsilon(s,0)\bar Q_0 + \bar Q_s H(s) Q_s U_\epsilon(s,0)\bar Q_0.
$$
We note that
\be\label{eq:normR}
\norm{R_\epsilon(s)} \leq   \pa{\epsilon\norm{ \dot{Q}_s}+\norm{ [H(s), Q_s]}}\norm{Q_s U_\epsilon(s,0)\bar Q_0}\le  C_N  \epsilon {\epsilon^{N}}\Delta^{-N}  +  C e^{- c \sqrt{\ell}}
\ee
by Lemma~\ref{lemma:AT-loc} and \eqref{eq:offdb}.

Let  $W_\epsilon(s) = T(0,s) V_\epsilon(s)$, where $T$ was defined in \eqref{eq:paralltr}. Then,
$$
i \epsilon \partial_s W_\epsilon(s) = T(0,s) \bar{H}(s) T(s,0) W_\epsilon(s) + T(0,s) R_\epsilon(s).
$$
By Lemma~\ref{lem:gapped_approx}, the operator $\bar{H}(s)$ has a gap $\delta$ in its spectrum that  separates the associated spectral projections $Q_s^\pm$. 
This implies that $T(0,s) \bar{H}(s) T(s,0)$ has the same gap with the associated projections given by $\caQ^\pm_s:=T(0,s) Q_s^\pm T(s,0)$. We  can bound 
\[\norm{\pa{T(0,s) \bar{H}(s) T(s,0)}^{\pa{k}}}\le C_k\beta  \quad \mbox{ for } \quad k=1,\ldots,N,\]
using \eqref{eq:smbarH}, \eqref{eq:smoot}, and the Leibniz rule. 

Let $\tilde W_\epsilon(s)$ denote the evolution generated by $T(0,s) \bar{H}_s T(s,0)$:
\be
i \epsilon \partial_s \tilde W_\epsilon(s) = T(0,s) \bar{H}(s) T(s,0) \tilde W_\epsilon(s), \quad \tilde W_\epsilon(0)=1.
\ee
Then,  it follows from our previous analysis and the Leibniz rule that  $T(0,s) \bar{H}(s) T(s,0)$ satisfies Assumption \ref{as:Kato} and  the gapped adiabatic theorem to all orders, Lemma~\ref{lemma:all_orders},  is applicable. Hence
\be
\max_s\norm{ \caQ^+_1 \tilde W_\epsilon(s) \caQ^-_0}\le  C \epsilon \delta^{-1},\quad \norm{ \caQ^+_1 \tilde W_\epsilon(1) \caQ^-_0}\le C_N  {\epsilon^N}{\delta^{-N}}.
\ee

We now observe that
\[
W_\epsilon(s)= \tilde W_\epsilon(s)+ i\epsilon^{-1} W_\epsilon(s)\int_0^sW^*_\epsilon(s')T(0,s') R_\epsilon(s')\tilde W_\epsilon(s')ds',
\]
so 
\be
\norm{W_\epsilon(s)- \tilde W_\epsilon(s)}\le \epsilon^{-1}
\max_{s'\le s}\norm{R_\epsilon(s')} \leq  C_N  {\epsilon^{N}}\Delta^{-N} +  C e^{- c \sqrt{\ell}},
\ee
using \eqref{eq:normR}.
We conclude that
\begin{align*}
\norm{Q_1^{+} V_\epsilon(s) Q_0^{-} } &=\norm{Q_1^{+} T(s,0) W_\epsilon(s) Q_0^{-} } =\norm{\caQ^+_1  W_\epsilon(s) \caQ^-_0}\\ &\leq \begin{cases}C_N {\epsilon^{N}}\Delta^{-N}   +   C\pa{ \epsilon \delta^{-1}+e^{- c \sqrt{\ell}}} & \mbox{uniformly in } s;\\ C_N \pa{ {\epsilon^{N}}\Delta^{-N}  +  {\epsilon^N}{\delta^{-N}} } +   C e^{- c \sqrt{\ell}} & \mbox{if } s=1.\end{cases}
\end{align*}
As $V_\epsilon(s) = \bar Q_s U_\epsilon(s,0) \bar Q_0$, and $ \bar Q_0 Q_0^{-}=Q_0^{-}$, it follows that 
\[
\begin{aligned}
\norm{Q_1^{+} U_\epsilon(s,0) Q_0^-}&\leq \norm{Q_1^{+} V_\epsilon(s) Q_0^{-} }+\norm{Q_1 U_\epsilon(s,0)  \bar{Q}_0 }  \\ &\le  
\begin{cases}C_N {\epsilon^{N}}\Delta^{-N}   +   C\pa{ \epsilon \delta^{-1}+e^{- c \sqrt{\ell}}} & \mbox{uniformly in } s;\\ C_N \pa{ {\epsilon^{N}}\Delta^{-N}  +  {\epsilon^N}{\delta^{-N}} } +   C e^{- c \sqrt{\ell}} & \mbox{if } s=1,\end{cases}
\end{aligned}
\]
where in the last step we have used \eqref{eq:offdb}.
\end{proof}

Let $P^{-}(s)$ be the spectral projection of $H(s)$ on the interval $(-\infty, E-6\delta)$ and $P^{+}(s)$ be the spectral projection on the interval $(E+6\delta, \infty)$.

We are now ready to complete the proof.
\begin{proof}[Proof of Theorem~\ref{thm:AT-loc}.]
We pick $\mathcal Q(s)=Q^-_s$. 

\cref{it:1Q}:  Using the integral representation \eqref{eq:Rder},
\[Q^-_s=\pa{2\pi i}^{-1}\oint_{\Gamma}\pa{\bar{H}(s)-z}^{-1}dz,
\]
we get
\[[\mathcal Q(s),H(s)]=\pa{2\pi i}^{-1}\oint_{\Gamma}\pa{\bar{H}(s)-z}^{-1}[H(s),\bar{H}(s])\pa{\bar{H}(s)-z}^{-1}dz,\]
and we can bound 
\[\norm{[\mathcal Q(s),H(s)]}\le C \delta^{-1} \norm{[H(s),\bar{H}(s)]}.\]
But
\[[H(s),\bar{H}(s)]=[H(s),\bar{Q}_sH(s)\bar{Q}_s]=[H(s),\bar{Q}_s]H(s)\bar{Q}_s+h.c.,\]
which yields
\[\norm{[H(s),\bar{H}(s)]}\leq C_N \epsilon + C e^{- c \sqrt{\ell}} \]
by Lemma~\ref{lemma:AT-loc}. Hence
\[\norm{[\mathcal Q(s),H(s)]}\leq C_N {\epsilon}{\delta^{-1}} + C e^{- c \sqrt{\ell}},\]
and \ref{it:1Q} follows.

\cref{it:2Q}: Using \eqref{eq:confesttil} and $Q^-_s\bar{Q}_s=Q^-_s$, we deduce that
\[\norm{\pa{{H}(s)-\bar{H}(s)}P_{< E-6\delta}(H(s))}+\norm{\pa{{H}(s)-\bar{H}(s)}{\mathcal Q}(s)}\leq C_N {\epsilon}\Delta^{-1} + C e^{- c \sqrt{\ell}}.\]
Hence, we can use Lemma \ref{lem:PbarP} with $H_1=\bar{H}(s)$, $H_2={H}(s)$, and $R=P_{< E-6\delta}(H(s))$ to first get
\[\norm{\bar{\mathcal Q}(s)P_{< E-6\delta}(H(s))}\leq C_N {\epsilon}\Delta^{-1} + C e^{- c \sqrt{\ell}},\]
and then use the same lemma with $H_1={H}(s)$, $H_2=\bar{H}(s)$, and $R=\mathcal Q(s)$  to get
\[\norm{ P_{> E+6\delta}(H(s))\mathcal Q(s)}\leq C_N {\epsilon}\Delta^{-1} + C e^{- c \sqrt{\ell}}.\]

\cref{it:3Q}:  This part follows directly from Lemma \ref{lem:dec} and the $\pm$ symmetry in the argument there, as
\[\norm{\mathcal Q_\epsilon(s)-\mathcal Q(s)}=\norm{U_\epsilon(s,0)Q_0^- U_\epsilon(0,s) -Q_1^-}\le\norm{ Q_1^+ U_\epsilon(1,0) Q_0^-}+\norm{ Q_1^- U_\epsilon(1,0) Q_0^+}.\]

\end{proof}

\section{Uniformly localized eigenfunctions for $H(s)$ and the proof of Theorem \ref{thm:AT-single'}}\label{sec:glob}
Disclaimer: In the process of completing this paper, we learned about a recent paper  \cite{KlS}, which has a significant thematic overlap with the results presented here.  
\subsection{Non-uniform bound on localization}
Let $H_\omega$ be an infinite volume operator satisfying Assumptions \ref{hyp1}--\ref{assump:multiplicities}.  We will need a stronger concept of a localizing Hamiltonian than the one introduced earlier in Definition \ref{def:localizing}.
\begin{defn}\label{def:nutheta}
For $\omega\in\Omega$ and a pair $\pa{c,\theta}$ of positive valued parameters, we will say that $H_\omega$ is {\it non-uniformly $\pa{c,\theta}$-localizing} if there exists an eigenbasis $\{\psi_i\}$ for $H_\omega$   such that 
\be\label{eq:AW1}
\abs{\psi_i(y)}^2\le \frac1\theta\langle x_i\rangle^{d+1}\e^{-c\abs{y-x_i}}\mbox{ for some } x_
i\in\Z^d.
\ee
Here, the quantifier "non-uniformly"  refers to the presence of the factor $\langle x_i\rangle^{d+1}$.
\end{defn}

\begin{thm}[Non-uniform eigenfunction localization]\label{thm:AWa}
Let $H_\omega$ be an infinite volume operator satisfying Assumptions \ref{hyp1}--\ref{assump:multiplicities} with $m=1$. 
Then 
\be\label{eq:thetaloca}
\mathbb P\pa{\set{\omega\in\Omega:\ \mbox{$H_\omega$ is non-uniformly $\pa{c,\theta}$-localizing}}}\ge 1-C\theta
\ee
for some $C>0$.
\end{thm}
\begin{proof}
The assertion above follows from {\cite[Theorem 7.4]{AW}} by Markov's inequality. 
\end{proof}
\subsection{From non-uniform to uniform estimates}\label{sec:6}
Our first goal in this section is to remove the "non-uniform" part from the above statement, at the price of a small fraction of eigenstates for which the statement will fail to hold. 

We first note that  that the integrated density of states (IDOS) $\mathcal N_{J_{loc}}$ of $H_o$, associated with the interval $J_{loc}$, given by
\be\label{eq:IDOS}
\mathcal N_{J_{loc}}=\lim_{R\to\infty}\frac{\tr \chi_{\Lambda_R(0)}P_{J_{loc}}(H_\omega)}{R^d},
\ee
is well-defined and almost surely non-random, see e.g., \cite[Theorem 3.15 and Corollary 3.16]{AW}. Moreover, if $\mathcal N_{J_{loc}}>0$,  the convergence to the mean in \eqref{eq:IDOS} is exponentially fast, so in particular  
\be\label{eq:LD}
\mathbb P\pa{\frac{\tr \chi_{\Lambda_R(0)}P_{J_{loc}}(H_o)}{R^d}<\frac{\mathcal N_{J_{loc}}}2}\le \e^{-mR}
\ee
for some $m>0$. This is a typical large deviations result, see e.g., \cite{CL}.  

We now adjust the concept of localized eigenvectors to make it uniform. We will assume here that $\mathcal N_{J_{loc}}>0$.

\begin{defn}\label{def:nuthet}
For $\omega\in\Omega$ and a pair $\pa{c,\theta}$ of positive parameters, we will say that a normalized $\psi\in\ell^2(\Z^d)$ of $H_\omega$ is {\it $\pa{c,\theta}$-localized} if there exists $x\in\Z^d$ (called a localization center) such that 
\be\label{eq:nutheta}
\abs{\psi(x)}^2\ge {\abs{\ln\theta}^{-d-1}} \mbox{ and } \abs{\psi(y)}\le \frac{\abs{\ln \theta}^{\frac{d+1}{2}}}{\theta}\e^{-c\abs{y-x}},\quad y\in\Z^d.
\ee

We will say that the orthogonal projection $P\in\caL(\ell^2(\Z^d))$ is {\it $\pa{c,\theta}$-Wannier decomposable} if there exists an orthonormal basis $\{\psi_i\}$ for $Ran(P)$ such that each $\psi_i$ is $\pa{c,\theta}$-localized.   
\end{defn}
Armed with this definition, we proceed in getting the uniform estimates, first for  finite (albeit arbitrary large) systems, and then for  infinite volume ones.

Let $H^{\mathbb{T}}_L$ denote the periodic restriction of $H_\omega$ to the torus $\mathbb{T}_L$ of a linear size $L$. The following assertion follows from the judicious use of Markov's inequality and the  deterministic Lemma \ref{lem:cent_mass_deg} below.
\begin{thm}\label{thm:Gevent}
Suppose that Assumptions \ref{hyp1}--\ref{assump:multiplicities} hold and that in addition  $\mathcal N_{J_{loc}}>0$. For a given configuration $\omega\in \Omega$, let $\bbP_E$ denote the normalized counting measure of eigenvalues of $H_L^{\mathbb{T}}$ in the interval $J_{loc}$ (counting multiplicities). Let $\caG$ be the set
\[\caG:=\set{E_n\in\sigma(H_L^{\mathbb{T}})\cap J_{loc}:\  P_{\set{E_n}}\mbox{ is }\pa{\tfrac{c}m,\theta^2}\mbox{-Wannier decomposable}}.\]
Then there exist $c,C>0$ such that for sufficiently small $\theta$ and any $L$ we have a bound
\be\label{eq:thetalocm}
\mathbb P \pa{\bbP_E\pa{\caG}\ge 1-\sqrt\theta}\ge 1-C\sqrt\theta.
\ee
\end{thm}
\begin{proof}
For a pair $(E_n,P_{\set{E_n}})$, let 
\be\label{eq:w_n}w_n=w(\omega,P_{\set{E_n}})= \sum_{x,y}\abs{P_{\set{E_n}}(x,y)}\e^{c\abs{x-y}}.\ee
We then have, by the bound \eqref{eq:eigcor} on the eigenvector correlator  and $\mathcal N_{J_{loc}}>0$, 
\[
\bbE_\omega \bbE_E[w_n] \leq C.
\]
  Letting $a,b>0$, we have by Markov's inequality that
\[
\bbP_\omega \pa{\bbE_E[w_n] \leq \theta^{-a}} \geq 1- C\theta^a
\]
We now pick an $\omega$ such that  $\bbE_E[w_n] \leq {\theta^{-a}}$. 
Another application of Markov's inequality then gives  
\be\label{eq:doubleM}
\bbP_E (w_n \leq \theta^{-b})  \ge  1- \theta^{b-a}.
\ee
The assertion now follows from \eqref {eq:doubleM} with $a=\frac12$, $b=1$, and Lemma \ref{lem:cent_mass_deg}.
\end{proof}
We are now ready to complete
\begin{proof}[Proof of Theorem \ref{thm:AT-single'}]
Here we will use $\theta=e^{-c \sqrt{\ell}}$.

Let $\caL=C\epsilon^{-1}$ and consider
\be\label{eq:boxesXi'}
 \Xi_{\caL}:=  \pa{\tfrac 3 2\caL  \Z}^{d},
\ee
cf. \eqref{eq:boxesXi}, and an $\caL$-cover of $\Z^d$ of the form
\[\Z^d=\bigcup_{a \in  \Xi_{\caL}} {\Lambda}_{\caL}(a).\]
We note that for any $x\in\Z^d$ we can find $a \in  \Xi_{\caL}$ such that $\dist\pa{{\Lambda}_{\caL}^c(a),x}\ge \caL/4$.

We also cover $J_{loc}'$ with the overlapping intervals $\set{J_i}$ so that 
\begin{enumerate}
\item  The length of each interval $J_i$ is equal to $c\ell^{-\xi}$;
\item For each $E\in J_{loc}'$ that satisfies $\dist\pa{E,\pa{J_{loc}'}^c}\ge \ell^{-\xi}$ we can find $J_i$ such that $\dist\pa{E,\pa{J_i}^c}\ge c\ell^{-\xi}/3$;
\item $\cup_i J_i\subset J_{loc} $.
\end{enumerate}
One can always construct such a covering using $C\ell^{\xi}$ intervals $J_i$ for $\ell$ sufficiently large.

We will say that a property $\caA$ is satisfied for at least a fraction $1-\sqrt\theta$ of boxes ${\Lambda}_{\caL}(a)$ (which we will be calling  good boxes) if 
\be
\lim_{R\to\infty}\frac{\# {\Lambda}_{\caL}(a)\subset \Lambda_R:\ \caA \mbox{ is satisfied for } {\Lambda}_{\caL}(a)}{{\# {\Lambda}_{\caL}(a)\subset \Lambda_R}}\ge1-\sqrt\theta.
\ee

For a given  box ${\Lambda}_{\caL}(a)$ in the cover we construct the corresponding torus $\T_a$ and pick any  $J_i$ from the cover of $J_{loc}'$. It follows that   the conclusions of Theorem \ref{thm:AT-single}  are satisfied with probability $\ge 1-e^{-c\sqrt\ell}$. Moreover, as the number of $J_i$s in the cover is $C\ell^{\xi}$, we deduce that with the same probability  the conclusions of Theorem \ref{thm:AT-single} hold for {\it all} $J_i$s in the cover. We next note that, given $N$ tori $\set{\T_a}$, we can choose at least $6^{-d}N$ of them to be separated by a distance greater than $r$, see the proof of Lemma \ref{clustprob}. Hence, using Assumption \ref{assump:FRC} and ergodicity, we obtain  that the fraction $1-e^{-c\sqrt\ell}$  of  tori $\set{\T_a}_{a \in  \Xi_{\caL}} $ satisfy the conclusions of Theorem \ref{thm:AT-single} for each interval $J_i$ in the cover of $J_{loc}$.

Let $\Omega_1\subset\Omega$ be a collection of $\omega$ such that  $\bbP_E\pa{\caG}\ge 1-\sqrt\theta$ for all $R\ge R_o$ (in particular, $\mathbb P\pa{\Omega_1^c}\le e^{-c \sqrt{\ell}}$ holds by \eqref{eq:thetalocm}).

We now pick any $\omega\in \Omega_1$ and conclude from Theorem \ref{thm:Gevent} that  the fraction $1-e^{-c \sqrt{\ell}}$ of eigenstates $\psi_n$ for $H^{\T}$ with eigenvalues $E_n\in J_{loc}$ are $\pa{c/m,\theta^2}$-localized. Let $\psi$  be  such eigenfunction, with energy $E$ and a localization center at $x$. Then there exists a box $a \in  \Xi_{\caL}$ and an interval $J_i$ such that 
\[\dist\pa{{\Lambda}_{\caL}^c(a),x}\ge \caL/4,\quad \norm{\bar\chi_\Lambda \psi}\le e^{-c\caL}, \quad E\in J_i.\]
If this box happens to be a good box, then the first assertion of Theorem \ref{thm:AT-single'} holds for all $s$ by Theorem \ref{thm:hypa} while the second assertion holds for $\psi$ at  $s=0$ by Lemma \ref{outbad} below and by the assertions of Theorem \ref{thm:hypa}. It then follows from Theorem \ref{thm:AT-single} (see Remark \ref{rem:psirep} there) that the second assertion holds for all $s\in[0,1]$. Since the fraction of good boxes is $1-e^{-c \sqrt{\ell}}$, we get the result.

\end{proof}

\section{Derivation of Linear Response Theory}\label{sec:LRT}

In this section, we prove Theorem \ref{thm:QHE} assuming the setting described in Section \ref{sec:int}. The proof rests on several technical results proven at the end of the section. Since the methods used here are sufficiently standard, our arguments will be somewhat abbreviated for the most part.

\begin{proof}[Proof of Theorem~\ref{thm:QHE}] 
In the rescaled variable $s = \epsilon t$ and for the zero temperature case ($\rho=P:=P_F$, the Fermi projection at $s=-1$), \eqref{eq:measured_response}
assumes form
\[
\sigma_m =\beta^{-1} \int_0^1{
\tr\pa{\pa{P_\epsilon(s) - P} J}} ds,
\]
 see Section~\ref{sec:1.3}.

It is a standard fact in the theory of quantum Hall effect, often referred to as ``cross geometry'', that the operator $\pa{P_\epsilon(s) - P} J$ is supported (in an appropriate sense) around the origin. We make this precise in Lemma~\ref{lem:fst} and use it to show that $\pa{P_\epsilon(s) - P} J$ is trace class and that we can replace the plane by a torus of linear size $\mathcal{L}$ up to exponentially small errors. Explicitly, let $\caL=C\epsilon^{-1}$ and let $\T$ be a torus of linear size $\caL$. Then we show that
\be\label{eq:tredJ}
\mathbb E\pa{\sup_{\caB}\abs{\tr{\pa{P_\epsilon(s) - P} J}-\tr{\pa{P^{\mathbb{T}}_\epsilon(s) - P^{\mathbb{T}}} \tilde{J}}}}\le  Ce^{- c \mathcal{L}},
\ee
where $P^\mathbb{T}=P_{E_F}(H^\mathbb{T})$ is a Fermi projection on the torus,  $\tilde{J} = \chi_{\caB} J $, and  the supremum is taken over $\caB\subset \T$ satisfying $\Lambda_{\caL/4}\subset \caB \subset \Lambda_{\caL/3} $.  

In the torus geometry we can apply the local adiabatic theorem. For this we fix $\epsilon = e^{- a \sqrt{\ell}}$ and  $\ell = (\beta/a)^{-2p}$ with $2p < 1/p_1$ so that $\epsilon = e^{-\beta^{-p}}$. Then for $a$ small enough  (but $\beta$-independent) the assumptions of Theorem~\ref{thm:AT-loc} hold, i.e. there exists an 
event $\mathcal E$ for which Theorem \ref{thm:hypa} (and consequently Theorem \ref{thm:AT-loc}) is applicable, and $\mathbb P(\mathcal E)\ge 1-e^{-c\sqrt\ell}$.

We next decompose $P^\mathbb{T}$ into two components $P^\mathbb{T} = \mathcal{Q}(-1) + R$ where $\mathcal{Q}(s)$ is the smooth adiabatic projection constructed in Theorem~\ref{thm:AT-loc}  (adjusted to the interval $(-1,1)$) and $R := P^\mathbb{T} - \mathcal{Q}(-1)$. By Theorem~\ref{thm:AT-loc} we then have that for $s \geq 0$ and $N \in \mathbb{N}$,
\[
\|P^\mathbb{T}_\epsilon(s) - \mathcal{Q}(0) - R_\epsilon(s) \| \leq  C_N \epsilon^N \left( \frac{1}{\Delta^N} +  \frac{1}{\delta^{2N+1}}  \right) +   \caO(e^{-c\sqrt{\ell}}),
\]
with $R_\epsilon=U_\epsilon(s) RU^*_\epsilon(s)$, where we have used $\mathcal{Q}(s) = \mathcal{Q}(0)$ for $s\ge0$. Hence, for $a$ small,
\be\label{eq:sigmm}
\sigma_m = \frac{1}{\beta} \tr( (\mathcal{Q}(0)) - \mathcal{Q}(-1)) \tilde{J})+ \frac{1}{\beta} \int_0^1 \tr(R_\epsilon(s) - R) \tilde{J}) d s + \caO(e^{-a \sqrt{\ell}}) .
\ee
For each $\omega\in\mathcal E$, we will construct a suitable set $\caB=\caB_\omega$ that will be used in the analysis below. 
In Proposition~\ref{prop:sigma} we will establish that for such $\caB$ we have 
\be\label{eq:clsigm}
\frac{1}{\beta} \tr{ \pa{\mathcal{Q}(0) - \mathcal{Q}(-1)} \tilde{J}}  =  \sigma_H + \caO(e^{-c \sqrt{\ell}}),
\ee
where $\sigma_H$ was defined in \eqref{eqdef:sigma}.
The principle idea here is that $\mathcal{Q}$ differs from the Fermi projection by localized states that do not contribute to the Hall conductance. 

Finally, in  Proposition \ref{prop:R}
we will show that for the same $\caB$, the remainder  can be estimated as 
\be\label{eq:remter}
\abs{ \frac{1}{\beta} \int_0^1 \tr(R_\epsilon(s) - R) \tilde{J}) d s}  \le C L^2 \frac{\epsilon}{\beta} +e^{-c \sqrt{\ell}} \le  Ce^{-a \sqrt{\ell}}.
\ee

Combining the bounds \eqref{eq:tredJ}--\eqref{eq:remter}, we obtain
\[\mathbb E\abs{\sigma_m-\sigma_H} \le C e^{-a \sqrt{\ell}} +C e^{-c \sqrt{\ell}}+\mathbb P(\mathcal E^c) C\caL \leq Ce^{-a \sqrt{\ell}} ,\]
where in the last step we used the rough deterministic estimate
\be\label{eq:aprH}\abs{\tr{\pa{P^{\mathbb{T}}_\epsilon(s) - P^{\mathbb{T}}}\tilde J}}\le C\caL.\ee
This completes the proof of Theorem~\ref{thm:QHE}.

The statement of Remark  \ref{rem:Thm1}(iv) can be now verified  as follows: We first use Remark \ref{rem:analit} below to reduce the finite temperature problem to the torus, just as for the $T=0$ case. We then use the spectral theorem for self-adjoint operators to decompose 
\be
\rho_T(H)=-\int_{-\infty}^\infty P_E\rho_T'(E)dE=-\int_{J_{loc}} P_E\rho_T'(E)dE+O(e^{-d_\mu/T}),
\ee
where $J_{loc}$ is the mobility gap that contains $\mu$. Using Theorem~\ref{thm:QHE} and the fact that $\sigma_H=\sigma_H(E)$ is almost surely $\omega$-independent constant within $J_{loc}$, we deduce that 
\[\mathbb E\abs{\sigma_m+\sigma_H\int_{J_{loc}} \rho_T'(E)dE} \leq C\pa{e^{-a \sqrt{\ell}}+\epsilon^{-1}e^{- d_\mu/T}}.\]
But 
\[\int_{J_{loc}} \rho_T'(E)dE=-1+O(e^{- d_\mu/T}),\]
and the result follows.
\end{proof}

We now present the technical statements used in the proof.
\begin{lemma}\label{lem:fst}
The operator ${\pa{P_\epsilon(s) - P} J}$ is trace class almost surely, and 
 \eqref{eq:tredJ} holds.
 \end{lemma}
 \begin{rem}
 We note that $\tilde{J}$ is  supported on a strip $|x_1| \leq r$.
 \end{rem}
\begin{rem}\label{rem:analit}
 If one replaces $P$ by the Fermi-Dirac distribution $\rho_T(H)$ with $\rho_T(E)={\frac {1}{e^{(E-\mu )/T}+1}}$, where $T$ is the absolute temperature and $\mu$ is the chemical potential, then  \eqref{eq:tredJ} holds {\it deterministically} with $c=1/T$ for $\epsilon\ll T$. 
 \end{rem}
\begin{proof}
We first note that \eqref{eq:decGh'} holds with $\Theta=\Z^2$ as well (the argument is only a slight modification of the one used in the proof of  \eqref{eq:decGh'} but is also an explicit content of \cite[Theorem 13.6]{AW}). Hence we have
\be
\sum_{x,y\in\Z^2} \langle x\rangle^{-3} e^{ 4c \abs{x-y}} \mathbb E \abs{P(x,y)}  \le C
\ee
for some $c>0$. 
Let 
\be\label{eq:A(om)}
A(\omega):=\sum_{x,y\in\Z^2} \langle x\rangle^{-3} e^{ 4c \abs{x-y}} \abs{P(x,y)},
\ee
then it follows that $A(\omega)\in L_1(\mathbb P)$. We will only consider configurations $\omega$ for which $A(\omega)<\infty$ (the set of full measure in $\Omega$) from now on. 

Using the fundamental theorem of calculus, we write
\[
\begin{aligned}
P_\epsilon(s) - P&=-U_\epsilon(s)\pa{\int_{-1}^s\partial_t \pa{U^*_\epsilon(t)PU_\epsilon(t)}dt}U^*_\epsilon(s)\\&=\frac{i}{\epsilon}U_\epsilon(s)\pa{\int_{-1}^sU^*_\epsilon(t)[H(t),P]U_\epsilon(t)dt}U^*_\epsilon(s)\\&=\frac{i\beta}{\epsilon}U_\epsilon(s)\pa{\int_{-1}^sg(t)U^*_\epsilon(t)[\Lambda_2,P]U_\epsilon(t)dt}U^*_\epsilon(s).
\end{aligned}
\]
We next note  that $ \norm{\Lambda_2 e^{4cx_2}\chi_{x_2<0}}\le 1$  and $ \norm{\bar \Lambda_2 e^{4cx_2}\chi_{x_2\ge0}}\le 1$. Thus,   using \eqref{eq:A(om)}  together with $[\Lambda_2,P]=-[\bar \Lambda_2,P]$, we get
\be\label{eq:LamP}
\norm{[\Lambda_2,P]\chi_{\set{x}}}\le  2A(\omega)\, \langle x\rangle^{3}  e^{-4c\abs{x_2}}.
\ee
Combining \eqref{eq:LamP} with Proposition \ref{thm:FSP}, we deduce that 
\be\label{eq:Lamev}
\norm{[\Lambda_2,P]U_\epsilon(t)\chi_{\set{x}}}\le CA(\omega)\, \langle x\rangle^{3}  e^{-c\abs{x_2}}\mbox{ for } \abs{x_2}\ge \mathcal{L}/3.
\ee
Since  $ \norm{\chi_{\set{x}}e^{c\abs{x_1}}J}\le C$ for all $x\in \Z^2$, we arrive to the bound
\be\label{eq:trnormb}
\norm{\pa{P_\epsilon(s) - P} \chi_{\set{x}}J }\le CA(\omega)\, \langle x\rangle^{3} e^{- c\abs{x}}\le A(\omega)\, e^{- c\abs{x}} \mbox{ for } \abs{x}\ge \mathcal{L}/3.
\ee
This bound immediately implies the first assertion of the lemma. 
We also observe that by the identical argument, one can also replace $P$ and $P_\epsilon(s)$) in the equation above with $P^\T$ and $P^{\T}_\epsilon(s)$, respectively.

To get the second claim of the lemma,
we first bound 
\be\label{eq:firtsJ}
\mathbb E\pa{\sup_{\caB}\abs{\tr{\pa{P_\epsilon(s) - P} J}-\tr{\pa{P_\epsilon(s) - P} \tilde J}}}\le\mathbb E\pa{\sup_{\caB}\abs{\tr{\pa{P_\epsilon(s) - P} \bar \chi_\caB J}}}\le Ce^{- c \mathcal{L}}
\ee
using \eqref{eq:trnormb} and $A(\omega)\in L_1(\mathbb P)$. 

The comparison between the plane and torus spectral projection will be established using  the bound
\be\label{eq:toruscomp'}
\mathbb E\norm{\pa{P- P^{\mathbb T}}\chi_{\Lambda_{\caL/2}(0)}}\le e^{-c\mathcal L},
\ee
 see \cite[Lemma 4.11]{EPS}. Using it together with 
    Proposition \ref{thm:FSP} (repeatedly) in the same vein as in the proof of the first part of the assertion,   we obtain
\be\label{eq:spfin'}
\mathbb E\norm{\pa{U_\epsilon(s,0)P_EU_\epsilon(0,s)-U^{\T}_\epsilon(s,0)P_E^\T U^{\T}_\epsilon(0,s)}\chi_{\Lambda_{\caL}/3}} \le e^{-c\caL}.
\ee
It implies that
\[\mathbb E\pa{\sup_{\caB}\abs{\tr{\pa{P_\epsilon(s) - P} \tilde J}-\tr{\pa{P^{\mathbb{T}}_\epsilon(s) - P^{\mathbb{T}}}\tilde J}}}\le e^{- c \mathcal{L}}, \]
and the result follows.

The statement of Remark \ref{rem:analit} can be verified in the similar fashion, using Proposition \ref{thm:FSP} and quasi-locality of analytic functions for local Hamiltonians,
\be
\abs{\rho_T(y,x)}\le C_T e^{-|x-y|/T},
\ee
see e.g., \cite[Corollary 5.2]{Rem} for the latter property. 
\end{proof}
We construct the suitable set $\caB$ for the next two assertions, given $\omega\in\mathcal E$. Let $\caA={\cup_\gamma \caT_\gamma}$, where the union is taken over all $\gamma$ such that $\caT_\gamma\cap  \Lambda_{\mathcal{L}/4}\neq\emptyset$, and  let $\mathcal{B} = \Lambda_{\mathcal{L}/4}\cup \caA$. We note that by construction $\Lambda_{\mathcal{L}/4}\subset\mathcal{B} \subset\Lambda_{\mathcal{L}/4+L}$ and 
\be\label{eq:distB}
\min_\gamma\dist\pa{\partial\caB,\hat{\mathcal T}_\gamma}\ge \ell/4
\ee 
(see the paragraph preceding \eqref{eq:Q_gamma} for notation). These two facts will be used often in the proofs below.

We will also need a set  $\mathcal{X}$ defined by
 \be\label{eq:X}
 \mathcal{X}=\set{\hat{\mathcal T}_\gamma:\ \set{\hat{\mathcal T}_\gamma \cap \{|x_j| \leq r}\neq\emptyset,\ j=1,2}.
 \ee 
 We note that  $|\mathcal{X}| \leq CL^2$. 
\begin{prop}\label{prop:sigma}
 For any  $\omega\in\mathcal E$, the relation \eqref{eq:clsigm} holds.
\end{prop}
\begin{proof}
We note that by locality of $H$,  $\tilde{J} =i \chi_{\mathcal{B}} [H^\mathbb{T}(r), \Lambda_1]$. By the fundamental theorem of calculus,
\[
\frac{1}{\beta} \tr{ \pa{\mathcal{Q}(0) - \mathcal{Q}(-1)} \tilde{J}} = \frac{1}{\beta} \int_{-1}^0 \tr \left( \partial_r{\mathcal{Q}}(r) i \chi_{\mathcal{B}} [H^\mathbb{T}(r), \Lambda_1] \right)d r.
\]
We claim that 
\be\label{eq:stand}
\frac{1}{\beta}  \int_{-1}^0 \tr \left( \partial_r{\mathcal{Q}}(r)  \chi_{\mathcal{B}} [H^\mathbb{T}(r), \Lambda_1] \right)dr =   \int_{-1}^0 \dot{g}(r) \tr \left( \mathcal{Q}(r) [ [\mathcal{Q}(r), \Lambda_1], [\mathcal{Q}(r), \Lambda_2]] \chi_{\mathcal{B}} \right) dr+  \caO(e^{-c \sqrt{\ell}}).
\ee
Indeed, let $\hat{\Lambda}_1(r) = \mathcal{Q}(r) \Lambda_1 \bar{\mathcal{Q}}(r) + \bar{\mathcal{Q}}(r) \Lambda_1 \mathcal{Q}(r)$. We have
\begin{align*}
  \int_{-1}^0\tr \left( \partial_r{\mathcal{Q}}(r)  \chi_{\mathcal{B}} [H^\mathbb{T}(r), \Lambda_1] \right)dr &=   \int_{-1}^0\tr \left( \partial_r{\mathcal{Q}}(r) \chi_{\mathcal{B}} [H^\mathbb{T}(r), \hat{\Lambda}_1(r)] \right) +  \caO(e^{-c \sqrt{\ell}}) dr\\
																&=  \int_{-1}^0 \tr \left(-[H^\mathbb{T}, \partial_r{\mathcal{Q}}(r)]  \chi_{\mathcal{B}}  \hat{\Lambda}_1(r) \right) dr+  \caO(e^{-c\sqrt{\ell}}) \\
																&=   \int_{-1}^0\tr  \left([\dot{H}^\mathbb{T}, \mathcal{Q}(r)]  \chi_{\mathcal{B}}  \hat{\Lambda}_1(r) \right)dr +  \caO(e^{-c\sqrt{\ell}}) \\
																&=  \int_{-1}^0\tr \left([\beta \dot{g}(r) \Lambda_2, \mathcal{Q}(r)]  \chi_{\mathcal{B}}  \hat{\Lambda}_1(r) \right)dr +  \caO(e^{-c \sqrt{\ell}}),
\end{align*} 	
where in the first step we have used $\mathcal{Q}(r) \partial_r{\mathcal{Q}}(r) \mathcal{Q}(r) = \bar{\mathcal{Q}}(r) \partial_r{\mathcal{Q}}(r) \bar{\mathcal{Q}}(r)=0$ and in the third step we employed    $[H^\mathbb{T}, {\mathcal{Q}}(r)] = \caO(e^{-c \sqrt{\ell}})$  and integration by parts. We have also repeatedly used the fact that commuting $\chi_\mathcal{B}$ with other operators under the trace  contributes  $\caO(e^{-c \sqrt{\ell}})$ by virtue of \eqref{eq:distB} and the location of support of the involved operators. The relation  \eqref{eq:stand} now follows, since $\hat{\Lambda}_1 = [\mathcal{Q}(r),[\mathcal{Q}(r), \Lambda_1]]$.

The implication is that
\be
\frac{1}{\beta} \tr{ \pa{\mathcal{Q}(0) - \mathcal{Q}(-1)) \tilde{J}}} = i\, \int_{-1}^0 \dot{g}(r) \tr{\pa{ \mathcal{Q}(r) [ [\mathcal{Q}(r), \Lambda_1], [\mathcal{Q}(r), \Lambda_2]]} \chi_{\mathcal{B}}} +  \caO(e^{-c \sqrt{\ell}}).
\ee
We now define
\be\label{eq:index}
\mathrm{Ind}_{\caL}\pa{\mathcal{Q}}=\tr{\pa{ \mathcal{Q}[ [\mathcal{Q}, \Lambda_1], [\mathcal{Q}, \Lambda_2]]} \chi_{\mathcal{B}}}.
\ee
For $\Z^2$ geometry without the cutoff function $\chi_{\mathcal{B}}$, the index (when it is well-defined) is known to be integer valued and additive. I.e., for orthogonal projections $Q,R$ with a compact $R$, $\mathrm{Ind}_{\infty}(Q +R) = \mathrm{Ind}_{\infty}(Q) + \mathrm{Ind}_{\infty}(R)$, provided $Q+R$ is a projection, \cite[Proposition 2.5]{ASS}. The argument in \cite{ASS} assumes that the underlying projections are covariant and that their kernels satisfy good decay properties. The latter hold in a random setting and one can  relax the covariance requirement for such models as well, see \cite{EGS}. Moreover, $\lim_{\caL\to\infty}\mathrm{Ind}_{\caL}\pa{{P}}$ exists and we have
\be\label{eq:indexZ2}
\lim_{\caL\to\infty}\mathrm{Ind}_{\caL}\pa{{P}}=\sigma,
\ee
\cite[Section 6]{ASS}. In fact, using \eqref{eq:decGh'} one can readily show that 
\be\label{eq:ASS}
\abs{\sigma-\mathrm{Ind}_{\caL}\pa{{P}}}\le \caO(e^{-c \caL}) \mbox{ and } \abs{\mathrm{Ind}_{\caL}\pa{{P}}-\mathrm{Ind}_{\caL}\pa{{P^\T}}}\le e^{-c \caL}.
\ee
We next observe that, although $P^\T$ and $\mathcal{Q}(-1) $ do not commute, we have $\norm{[P^\T,\mathcal{Q}(-1) ]}\le  e^{-c \sqrt{\ell}}$. Hence there exists a pair of self-adjoint operators $\hat P^\T$, $\hat{\mathcal{Q}}(-1)$ such that $[\hat P^\T,\hat{\mathcal{Q}}(-1)]=0$ and $\norm{P^\T-\hat P^\T}\le  e^{-c \sqrt{\ell}}$, $\norm{\mathcal{Q}(-1) -\hat{\mathcal{Q}}(-1)}\le  e^{-c \sqrt{\ell}}$, \cite{KaS}. Moreover, applying the compression procedure used to get a projection  $Q_s$ from a near-projection ${Q}_N(s)$ in the proof of Lemma \ref{lemma:AT-loc}, without loss of generality we can assume that $\hat P^\T$, $\hat{\mathcal{Q}}(-1)$ are  projections.  Let $\check R=\hat P^\T-\hat{\mathcal{Q}}(-1)$. Since   $\norm{\mathcal{Q}(-1) R}\le  e^{-c \sqrt{\ell}}$, we conclude that ${\hat{\mathcal{Q}}(-1) \check R}=0$. In particular, the additivity of index is applicable for $\hat{\mathcal{Q}}(-1)$ and $\check R$, and yields
\be
 \abs{\mathrm{Ind}_{\caL}\pa{\hat{\mathcal{Q}}(-1)}+\mathrm{Ind}_{\caL}\pa{\check R}-\mathrm{Ind}_{\caL}\pa{{\hat P^\T}}}\le e^{-c \sqrt{\ell}}.
\ee
By construction, we deduce that 
\be
 \abs{\mathrm{Ind}_{\caL}\pa{Y_i}-\mathrm{Ind}_{\caL}\pa{Z_i}}\le e^{-c \sqrt{\ell}}, \quad i=1,2,3,
\ee
where $Y_1=\check R$, $Z_1=R$, $Y_2=\hat{\mathcal{Q}}(-1)$, $Z_2={\mathcal{Q}}(-1)$, $Y_3=\hat P^\T$ and $Z_3= P^\T$. In addition, since $\mathcal{Q}(r)$ is continuous, we conclude that 
\be\label{eq:contind}
\mathrm{Ind}_{\caL}\pa{\hat{\mathcal{Q}}(r)}=\mathrm{Ind}_{\caL}\pa{\hat{\mathcal{Q}}(-1)}+\caO(e^{-c \sqrt{\ell}}).
\ee
Putting together \eqref{eq:ASS}--\eqref{eq:contind}, we see that the statement  follows if we can show that 
\be
\mathrm{Ind}_{{\caL}}(R) = \caO(e^{-c \sqrt{\ell}}).
\ee 
To establish this bound we observe that  
\[\mathrm{Ind}_{{\caL}}(R) =\mathrm{Ind}_{{\caL}}(R^\mathcal{X}) +\caO(e^{-c \sqrt{\ell}}),\]
where $\mathcal{X}$ was defined in \eqref{eq:X}, just as in the argument used in the second step above. But 
\[\mathrm{Ind}(R^\mathcal{X})=i\tr{{R^\mathcal{X} [[R^\mathcal{X}, \Lambda_1],[R^\mathcal{X}, \Lambda_2]]}}, \]
and the right hand side is $\caO(e^{-c \sqrt{\ell}})$  using  $R^\mathcal{X}\pa{\mathds{1}-R^\mathcal{X}}=\caO(e^{-c \sqrt{\ell}})$ and the cyclicity of the trace.
\end{proof}

\begin{prop}\label{prop:R}
For any $\omega\in\mathcal E$, the relation \eqref{eq:remter} holds.
\end{prop}
\begin{proof}[Proof of Proposition \ref{prop:R}]
It will be convenient to introduce  a new scale $\tilde{\ell}$ in addition to $\ell$, defined by the modified value for $\delta$, namely  $\tilde{\delta} = 7 \delta$. We consider the operator $\tilde{Q}_s$ constructed in Lemma~\ref{lemma:AT-loc}. The important properties of $\tilde{Q}_s$ are that it covers the spectral support of $R$ and that it allows us to control the spatial support of $R$. Let $I=(E - 6 \delta, E + 6 \delta)$. Using \cref{it:2Q}, we have  (recall that $R = P^\mathbb{T} - \mathcal{Q}(-1)$)
\[
\norm{ R  - P^\mathbb{T}_{I} R P^\mathbb{T}_{I} }\leq \caO(e^{-c\sqrt{\ell}}).
\]
By the definition of ${Q}_s$ and the exponential decay of $R$, we then obtain 
\[
\norm{ R - \sum_\gamma \tilde{Q}^\gamma_{-1} R \tilde{Q}^\gamma_{-1}} \leq \caO(e^{-c \sqrt{\ell}})
\]
and, using 
 Lemma \ref{it:sppa}, we see that for $s \geq 0$,
\be\label{eq:Rloc}
\|R_\epsilon(s) - \sum_\gamma {Q}^\gamma_s R_\epsilon(s) {Q}^\gamma_s\| \leq \caO(\epsilon^\infty + e^{-c \sqrt{\ell}}).
\ee
Since ${Q}^\gamma_s$ is supported in $\hat{\mathcal T}_\gamma$ (see the paragraph preceding \eqref{eq:Q_gamma} for notation), it follows that, up to a small error, $R_\epsilon(s)$ is the sum of terms supported in the region $\hat{\mathcal T}_\gamma$. Let $\hat U_\epsilon$ denote the evolution generated by $H_{\mathcal T}(s)$, the restriction of $H^\mathbb T(s)$ to the union of all ${\mathcal{T}}_\gamma$. Then we have
\[\frac{d}{ds}\pa{\hat U^*_\epsilon(s)R_\epsilon(s)\hat U_\epsilon(s)}=\frac i\epsilon\hat U^*_\epsilon(s)[H_{\mathcal T}(s)-H(s),R_\epsilon(s)]\hat U_\epsilon(s)=\caO(\epsilon^\infty + e^{-c\sqrt{\ell}}),\]
thanks to \eqref{eq:Rloc} and Lemma \ref{lem:HQDel}. Thus we can approximate
\[\|R_\epsilon(s) - \sum_\gamma \tilde{Q}^\gamma_s \hat R_\epsilon(s) \tilde{Q}^\gamma_s\| \leq \caO(\epsilon^\infty + e^{-c\sqrt{\ell}}),\]
where $\hat R_\epsilon(s)=\hat U^*_\epsilon(s)R\hat U_\epsilon(s)$.

Considering now any $\hat{\mathcal T}_\gamma\notin \mathcal{X}$ (recall \eqref{eq:X}),  either $\dist\pa{\hat{\mathcal T}_\gamma,\set{x\in\Z^2:\ x_1=0}}\ge r$, in which case 
\[{Q}^\gamma_s \tilde{J}=0,\] or  $\dist\pa{\hat{\mathcal T}_\gamma,\set{x\in\Z^2:\ x_2=0}}\ge r$, in which case 
\[{Q}^\gamma_s \hat R_\epsilon(s) {Q}^\gamma_s={Q}^\gamma_{-1} R {Q}^\gamma_{-1}+\caO(\e^{-c \sqrt{\ell}}),\] as the perturbation is constant in that region.  Hence, using \eqref{eq:Rloc} and Lemma \ref{lem:HQDel} again (recall that $A^\Theta$ stands for the restriction of the operator $A$ to the set $\Theta$), 
\be\label{eq:Reps}
\begin{aligned}
\tr{\pa{R_\epsilon(s) - R} \tilde{J}} &= \tr{\pa{ \pa{\hat R_\epsilon(s)}^\mathcal{X} - R^\mathcal{X}} \tilde{J}} + \caO(\epsilon^\infty + e^{-c \sqrt{\ell}})\\ &= \tr{\pa{  \pa{\hat R_\epsilon(s)}^\mathcal{X} - R^\mathcal{X}} J} + \caO(\epsilon^\infty + e^{-c \sqrt{\ell}}) .
\end{aligned}
\ee
Next we observe, using the cyclicity of the trace for a trace class operator and \eqref{eq:Rloc}, Lemma \ref{it:1Q}, and Lemma \ref{lem:HQDel} one more time, that 
\[
\tr{\pa{  \pa{\hat R_\epsilon(s)}^\mathcal{X} - R^\mathcal{X}} J}=-i\tr{\pa{[H_{\mathcal T}(s),\hat R_\epsilon(s)]}^\mathcal{X}\Lambda_1}+ \caO(e^{-c \sqrt{\ell}}).
\]
However, 
\[ -i\tr{\pa{[H_{\mathcal T}(s),\hat R_\epsilon(s)]}^\mathcal{X}\Lambda_1}={\epsilon}\partial_s\tr{\pa{\hat R_\epsilon(s)}^\mathcal{X}\Lambda_1}.\]
Hence by the fundamental theorem of calculus, 
\[
\frac{1}{\beta} \int_0^1 \tr{\pa{  \pa{\hat R_\epsilon(s)}^\mathcal{X} - R^\mathcal{X}} J}d s = \frac{\epsilon}{\beta} \tr{\left(\pa{\hat R_\epsilon(1)}^\mathcal{X}  - \pa{\hat R_\epsilon(0)}^\mathcal{X} \right) \Lambda_1} + \caO(e^{-c \sqrt{\ell}}),
\]
so we finally get
\be\label{eq:hatReps}
\abs{\frac{1}{\beta} \int_0^1\tr{\pa{  \pa{\hat R_\epsilon(s)}^\mathcal{X} - R^\mathcal{X}} J}d s} \leq C L^2 \frac{\epsilon}{\beta} + \caO(e^{-c \ell}).
\ee
\end{proof}

\appendix

\section{Hybridization in $1D$}\label{sec:hyb}
In this appendix, we show eigenvector hybridization for a family of $1D$ Anderson Hamiltonians. Apart from an occasional reference to a definition or a technical lemma, this appendix is self-contained. In some places, the notation used here conflicts with the notation used in the main text. 

We consider the Hilbert space $\ell^2\pa{\Z}$ and denote its scalar product by $\langle \cdot, \cdot \rangle$. Delta functions $\{\delta_x\}_{x \in \mathbb{Z}}$, equal to $1$ at $x$ and $0$ elsewhere, form a basis for the Hilbert space. The discrete Laplacian $\Delta$ is the operator given by 
$$
\langle\delta_x, \Delta \delta_y \rangle = \begin{cases}
										-2 & x=y,\\
										1 &  x \sim y, \\	
										0 & \mbox{otherwise},
								\end{cases}
$$
where $x \sim y$ denotes that $|x - y| =1$. We recall that $\sigma(-\Delta)=[0,4]$. We will use a decomposition $\Delta = \sum_{x \sim y} \Gamma_{xy}-2$, where $\Gamma_{xy}$ is a rank one operator defined by  $\Gamma_{xy} f = f(x)\delta_y $ for $f\in\ell^2(\Z)$. For a set $Z \subset \mathbb{Z}$, we let $\chi_Z = \sum_{x \in Z} \Gamma_{xx}$ be the orthogonal projection onto $Z$.

     Our results concern the analytic family of  Hamiltonians $H(\beta)$ with $\beta \in (-1,1)$ of the form 
\be\label{eq:Hbe}
H(\beta)=-\Delta+V_\omega+\beta W
\ee 
acting on $\ell^2\pa{\Z}$. Here, $V_\omega$ is a random potential, with $V_\omega(x)=\omega_x$ the i.i.d. random coupling variables distributed according to the Borel probability measure $\mathbb P:=\otimes_{\Z} P_0$. We will assume that the single-site distribution $P_0$ is absolutely continuous with respect to Lebesgue measure on $\mathbb R$. We assume that the corresponding Lebesgue density $\mu$ is  bounded with  $\supp(\mu)\subset[0,1]$, and that the single-site probability density is bounded away from zero on its support. We denote the configuration space by $\Omega$. The perturbation $W$ is a compactly supported non-negative potential. For concreteness, we anchor $W$ at the origin by assuming that $W(0) =1$ and $\|W\| = 1$, in particular $\|H(\beta)\| \leq 6$ in our setup. We remark that $\sigma(H(0))$ is a $\mathbb P$-a.s. deterministic set (see e.g., \cite[Theorem 3.10]{AW}), which we  denote by $\Sigma$, and  that $\Sigma\supset[0,5]$.

For a region $Z \subset \Z$, we write $H^Z = \chi_Z H \chi_Z$, understood as an operator acting on $\ell^2(Z)$. We will use the  natural embedding $\ell^2(Z) \subset \ell^2\pa{\Z}$ without further comment. With some slight abuse of notation, $(a,b)$ denotes $(a,b) \cap \Z$ whenever it signifies a subset of the lattice.

We consider a length scale $\caL$, a symmetric region $\Lambda_{full}: = (-\caL, \caL)$, and an asymmetric region $\Lambda: = (-\caL, 2\sqrt\caL/\ln\caL)$ that we  divide into a right region $\Lambda_{R} = (-2\sqrt \caL/\ln\caL, 2\sqrt\caL/\ln\caL)$, and a left region $\Lambda_{L} = \Lambda \setminus \Lambda_{R}$ (the reasons for this asymmetry will be clear later on). We denote by $r$ the leftmost point of $\Lambda_R$ and by $l$ the rightmost point of $\Lambda_L$, so  by construction $l \sim r$. We consider the Hamiltonians associated with these regions, $H_{full}:=H^{\Lambda_{full}}, H: = H^\Lambda, H_{L}: = H^{\Lambda_L}$, and $H_{R}: = H^{\Lambda_R}$, as well as the decoupled Hamiltonian $H_{dec}$ obtained by erasing the coupling between the left and right regions, i.e. $H_{dec} = H_{L} + H_{R} =  H - \Gamma_{lr} - \Gamma_{rl}$. All of these Hamiltonians a priori depend on $\beta$. Here and later, we only stress the dependence on $\beta$ in some equations, and suppress the dependence in others. We will assume henceforth that $\caL$ is large enough so that $\supp(W) \subset \Lambda_{R}$. In particular, $H_{L}$ does not depend on $\beta$.

We consider an eigenvector $\varphi_L$ of $H_L$ with eigenvalue $E_L \equiv E$ and a continuous family of eigenvectors $\varphi_R(\beta)$ of $H_{R}(\beta)$ with eigenvalues $E_R(\beta)$. We will assume that these two energy levels cross, i.e. $E - E_R(\beta)$ changes sign as $\beta$ varies. In Subsection~\ref{sec:crossing}, we will show that such levels exist with large probability thanks to two-sided Wegner estimates.

For a typical realization of the disorder, the eigenvectors  $\varphi_L$, $\varphi_R:=\varphi_R(0)$ are well localized with localization centers $x_{L}$, $x_{R}$, respectively (we will make this statement quantitative later on). We  pick  the eigenvectors in such a way  that $x_{R}$ is  close to the origin and  $x_{L}$ is located at least a   distance of $\sqrt\caL$ away from  $\Lambda_R$.  Let $P_{dec}$ be the orthogonal projection onto $Span(\varphi_{L}, \varphi_{R})$. Let us consider the rank two operator $\mathbb H:=P_{dec} H P_{dec}$ acting on $Ran(P_{dec})$. We note that the matrix representation for  $\mathbb H$ with respect to the $ \set{\varphi_L, \varphi_R}$ basis is given by  a $2\times 2$ matrix
\be\label{eq:2x2}
M_\beta:= \left( \begin{array}{cc}
											E & gap \\
											gap & E_R + \beta \langle W \rangle_{\varphi_R}
										\end{array} \right)
\ee
with $gap: = \langle \varphi_L, H(0) \varphi_R \rangle = \varphi_L(l) \varphi_R(r)$,  $ \langle W \rangle_{\varphi_R}:=\langle \varphi_R, W \varphi_R \rangle$.  Moreover, $gap\neq0$ since  eigenfunctions of a Schr\"odinger operator restricted to an interval do not vanish on its boundary. We now note that for $\beta$  such that $E_{L} = E_{R} + \beta \langle W \rangle_{\varphi_{R}}$, the eigenvectors $\varphi_\pm := \varphi_{R} \pm \varphi_{L}$ of $\mathbb H $ are delocalized in a sense that these functions are not small at both of the points $x_{R}$ and $x_{L}$, which are separated by a distance comparable with the system's size. We call this phenomenon a hybridization across lengthscale $\caL$. We are going to show that such hybridization also occurs for eigenvectors of the full Hamiltonian $H_{full}(\beta)$.    

\begin{defn} \label{def:hybridization} Let $F\in(0,1/2]$ be a parameter. We say that $H_{full}(\beta)$ $F$-hybridize on a length scale $\caL$  if there exists an analytical family of eigenvectors $\varphi(\beta)$ of $H_{full}(\beta)$ for $\beta \in (-1,1)$ such that
\begin{enumerate}
\item $\|\chi_{|x| \geq \sqrt\caL/\ln\caL} \varphi(0)\| \leq e^{-c \sqrt\caL/\ln\caL}$, 
\item There exists $\beta$ such that $\|\chi_{\Lambda_{L}} \varphi(\beta) \|^2 \geq F$, and $\|\chi_{|x| <\sqrt\caL/\ln\caL} \varphi(\beta)\|^2 \geq F$.
\end{enumerate}
We call $F$ a hybridization strength and denote by $\Omega_{F, \caL} \subset \Omega$ all realizations for which  $H_{full}(\beta)$ $F$-hybridize.
\end{defn}

\begin{thm}\label{thm:A-main}
For any $F < 1/2$, $\liminf_{\caL \to \infty} \mathbb{P}(\Omega_{F, \caL}) > 0.$
\end{thm}
If we now consider an infinite volume operator $H(\beta)$ (i.e., $\Lambda_{full}=\Z$), any  $F<\frac12$, and an arbitrary sequence $\mathcal L_n\to\infty$, then by the Borel-Cantelli lemma, for almost all random configurations $\omega\in\Omega$ we can find a subsequence $\mathcal L_{n_k}\to\infty$ such that $H^{\Lambda_{\caL_{n_k}}}(\beta)$ $F$-hybridizes. 

While there could potentially be different mechanisms leading to the hybridization phenomenon, our construction below hinges on the behavior of the simple two-level system (characterized by the avoided eigenvalue crossing)  discussed above. Since the probability of multiple level crossings is much smaller than that of two-level ones, we expect that this is  the only possible mechanism of hybridization, but in this work we have not tried to formalize this statement. We chose this definition for its simplicity; our construction of the hybridization event is more detailed and exactly matches the underlying motivation.

\subsection{Perturbation of a non-avoided crossing}
We consider the eigenvalues $E_L \equiv E, E_R(\beta)$ of $H_{dec}(\beta)$ for $\beta$ in a compact interval $J$ associated with the (normalized) eigenvectors $\varphi_L, \varphi_R(\beta)$. Later, the notation $E_L$ will stand more generally for an eigenvalue of $H_L$ and $E_R$ will stand for an eigenvalue of $H_R$, but this is not important at the moment. We assume that $\varphi_R(\beta)$ is continuous, which implies that $E_R(\beta)$ is continuous. 

\begin{figure}[ht]
\begin{tikzpicture}
\draw [dashed]  (-1,-0.05) -- (4,-0.05);
\draw [dashed]  (-1,4.05) -- (4,4.05);
\draw [cyan]  (-1,2) -- (4,2);
\draw [red] plot [smooth] coordinates {(-1,0) (0,-1) (2,4) (3,2) (4,2)};

\draw [cyan, xshift=7cm]  (-1,2) -- (4,2);
\draw [red, xshift=7cm] plot [smooth] coordinates {(-1,0) (0,-1) (2,1) (3,-1) (4,1.6)};

\node [cyan] (E) at (-.8,2.3) {$E$};
\node (E+h) at (-.5,3.7) {$E+h$};
\node (E-h) at (-.5,.3) {$E-h$};
\node [red] (Er) at (1.2,3) {$E_R$};

\node [cyan] (E) at (6.2,2.3) {$E$};
\node [red] (Er) at (7.5,0) {$E_R$};
\end{tikzpicture}
\caption{The left panel shows  the crossing of $E$ and $E_{R}$ (colored in cyan and red, respectively) as $\beta$ varies in $(-1,1)$. The parameter $h>0$ captures  the crossing width. The right panel shows the avoided crossing, $h=0$. }
\end{figure}
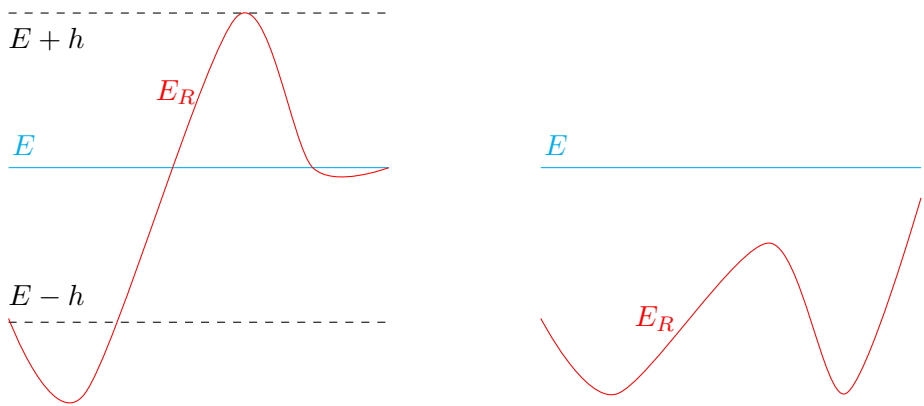
Let
\[
h := \min \{ \max_{\beta \in (-1,1)}(E - E_R(\beta))_+, \max_{\beta \in (-1,1)} (E_R(\beta)-E)_+ \},
\]
where $(x)_+$ is equal to $x$ for positive $x$ and zero otherwise. If the eigenvalues $E_R(\beta)$ do not intersect $E$, then $h=0$, otherwise $h$ is a maximal number such that both $E -h$ and $E+h$ intersect $E_R(\beta)$. 

Suppose now that the Hamiltonian $H(\beta)$ has a continuous family of spectral projections $P(\beta)$  such that (suppressing the $\beta$ dependence)
\be\label{eq:proxim}
\| P - P_{dec}\| \leq \varepsilon, \quad \| H_{dec} P_{dec} - H P\| \leq \varepsilon,
\ee
for some $\varepsilon \ll 1$. The range of $P$ is then two-dimensional and is spanned by (normalized) eigenvectors of $H$ that we denote $\varphi_{\pm}$. We denote the associated eigenvalues $E_{\pm}$. Let  $c_L^{\pm},c_R^{\pm}$ be the Fourier coefficients of $\varphi_{\pm} $ with respect to the elements $ \varphi_L, \varphi_R$ of an eigenbasis of $H_{dec}$, i.e., 
\[
\varphi_{\pm} = c_L^{\pm} \varphi_L + c_R^{\pm} \varphi_R + \varphi_{\perp}^\pm,
\]
with $\langle \varphi_{\perp}^\pm,\varphi_L\rangle=\langle \varphi_{\perp}^\pm,\varphi_R\rangle=0$, and let 
\[
F: = \max_{\beta \in (-1,1)} \min \pa{|c^+_L(\beta)|^2, |c^+_R(\beta)|^2 }.
\]
(This  value will be used for the parameter $F$ introduced in  \cref{def:hybridization}.)

Since $|c^+_L(\beta)|^2+|c^+_R(\beta)|^2\le1$, we know that $F \leq 1/2$. For $\epsilon=0$, $F$ equals zero by the continuity of $\beta$ dependence in $H$, $H_{dec}$, $P$, and $P_{dec}$, so  there is no hybridization. As can be seen from the two-level system described in \eqref{eq:2x2}, $F$ can be equal to $1/2$ for an arbitrarily small but non-zero value of $\epsilon$. Indeed, in this example $\epsilon>0$ corresponds to  $gap>0$ and  $F=1/2$ is achieved for  $\beta$ that solves $E_L= E_R + \beta\langle W \rangle_{\varphi_R}$.

Our principle indicator of hybridization will be the fact that $F$ has to be close to $1/2$ whenever the level crossing for $H_{dec}$ is avoided for the full $H$. 
\begin{lem}\label{lem:avoided_crossing} Suppose that $E_+(\beta), E_-(\beta)$ do not intersect in $J$, and $h\ge 4\varepsilon$. Then 
\[F \geq \frac{1 - \varepsilon^2}{2}.\]

\end{lem}
\begin{proof} 
By the continuity of $E_{\pm}$ and the non-crossing condition, we may assume without loss of generality that $E_+(\beta)-E_-(\beta)>0$ for $\beta\in (-1,1)$.
By the first equation in \eqref{eq:proxim}, we know that $\|\varphi_\perp^+\| \leq \varepsilon$, hence 
\be\label{eq:innpra}
|c^+_L(\beta)|^2 + |c^+_R(\beta)|^2 \geq 1- \varepsilon^2.
\ee
By the same equation, 
\be\label{eq:innpr}
\varepsilon\ge {\langle\varphi_\sharp,\varphi_\sharp-P\varphi_\sharp\rangle}=1- \pa{|c^-_\sharp|^2 + |c^+_\sharp|^2},\quad \sharp=L,R.
\ee
 On the other hand, the second equation in \eqref{eq:proxim} implies
\begin{equation}
\label{eq:A.3.1}
|c_\sharp^\pm|^2(E_\sharp - E_\pm)^2\le \varepsilon^2 ,\quad \sharp=L,R. 
\end{equation}
Using the second equation in \eqref{eq:proxim} and Weyl's theorem, \cite[Theorem 4.3.1]{HJ} we get
\be\label{eq:distE'}
\dist_H(\{0,E_L, E_R\}, \{0,E_-, E_+\})=\dist_H(\sigma(H_{dec} P_{dec}),\sigma(HP)) \leq \varepsilon,
\ee
where $\dist_H$ stands for the Hausdorff distance between a pair of sets. Hence, 
\be\label{eq:distE}
\dist_H(\{E_L, E_R\}, \{E_-, E_+\}) \leq 2\varepsilon.
\ee
The definition of $h$ implies that there exist $\beta_1, \beta_2\in (-1,1)$ such that  $E_L - E_R(\beta_1) = h$ and $E_R(\beta_2) - E_L =h$. Thus, it follows from \eqref{eq:distE} and $E_+(\beta)-E_-(\beta)>0$ for $\beta\in (-1,1)$ that 
\[\max\pa{|E_R(\beta_1) - E_-(\beta_1)|,|E_L-E_+(\beta_1) |,|E_R(\beta_2) - E_+(\beta_2)|,|E_L-E_-(\beta_2) |} \leq 2\varepsilon.\] 
 Using \eqref{eq:A.3.1} at $\beta_{1,2}$ with $\sharp=R$, we get 
$
|c_R^+(\beta_1)|^2(h - 2\varepsilon)^2 \leq \varepsilon^2
$ and $
|c_R^-(\beta_2)|^2(h - 2\varepsilon)^2 \leq \varepsilon^2$, which imply $
|c_R^+(\beta_1)|^2\le \frac1{4}$ and  $
|c_R^-(\beta_2)|^2\le \frac1{4}$. The latter relation yields $
|c_R^+(\beta_2)|^2\ge \frac3{4}-\varepsilon>\frac12$ by \eqref{eq:innpr}. It follows  from the  continuity of the coefficient  $c^+_R$  that there exists $\beta\in(\beta_1,\beta_2)$ such that 
$
|c^+_R(\beta)|^2 = \frac{1- \epsilon^2}{2}$. 
Hence, by \eqref{eq:innpra} we also have
$
|c^+_L(\beta)|^2  \geq  \frac{1- \epsilon^2}{2}$, completing the proof.
\end{proof}

\subsection{Construction of the non-avoided crossing}
\label{sec:crossing}
We first give a precise notion of eigenvector localization. 
\begin{defn}\label{def:nutheta'}
For $\omega\in\Omega$ and a pair $\pa{\nu,\theta}$ of positive parameters, we will say that $H$ is $\pa{\nu,\theta}$-localized if  all eigenvalues of $H$ are simple and for each $E\in \sigma(H)$, the corresponding eigenvector $\psi_E$ satisfies
\be\label{eq:AW}
\abs{\psi_E(y,\omega)}^2\le \frac1\theta\langle x_E(\omega)\rangle^2\e^{-\nu\abs{y-x_E(\omega)}}.
\ee
We call $x_E$ the localization center of the eigenvector $\psi_E$.
\end{defn}
One of the key results we will use in this appendix is
\begin{thm}[Eigenfunctions localization]\label{thm:AW}
There exist $C,\nu>0$ such that
\be\label{eq:thetaloc}
\mathbb P\pa{\set{\omega\in\Omega:\ \mbox{$H^\sharp(0)$ is $\pa{\nu,\theta}$-localized}}}\le 1-C\theta,\quad \sharp=\Lambda,\Lambda_L,\Lambda_R.
\ee
\end{thm}
\begin{proof}
This is a consequence  of \cite[Theorems 5.8, 7.4, and 12.11]{AW} and Markov's inequality.
\end{proof}
 We will  fix this value of $\nu$ henceforth.

\begin{defn}\label{def:crossing_con}
In this definition, we gather requirements on $\omega \in \Omega$ used in our construction. The requirements depend on a small parameter $\theta < 1$, and a large parameter $b$.

  There exists eigenvalues $E_{R}(0)$ (resp. $E_{L}$) of $H_{R}(0)$ (resp. $H_{L}$) with eigenvectors $\varphi_{L}, \varphi_{R}$ such that 
\begin{enumerate}
\item $H_L$, $H_{R}(0)$ are  $\pa{\nu,\theta}$-localizing; In particular, $\varphi_{L}, \varphi_{R}$ are localized;
\item $|E_{L} - E_{R}(0)| \leq  b \theta/\caL$; 
\item Let 
\be\label{eq:J}
J: = \{\lambda\in\mathbb R:\ \dist(\lambda,\set{E_L,E_R(0)})\le \sqrt{\theta}/\caL\}
\ee
 Then  $\sigma(H_L)\cap J=\{E_L\}$ and  $\sigma(H_{R}(0))\cap J=\{E_R(0)\}$.
\item $|\varphi_R(0)|^2 \geq - C_\nu/\ln \theta$. Here $C_\nu$ is an explicit constant given in Theorem \ref{thm:sense}.
\end{enumerate}
We will denote by $\caC$  a set of all $\omega\in \Omega$ for which (i)-(iv) hold true.
\end{defn}
For $\omega \in \caC$, let $(E_R(\beta), \varphi_R(\beta))$ be the eigenpair of $H_R(\beta)$ that depends smoothly on $\beta\in J$. 
\begin{prop}\label{prop:crossing} Suppose that $\omega \in \caC$ and that $\theta$ is small enough. Then $E_{R}(\beta), E_L$ intersect for some $\beta\in I$, where  
\be\label{eq:I}
I :=[-a,a],\quad a=4 \frac{b}{C_\nu}  \frac{\theta\ln \theta}{\caL},
\ee
 and the associated function $h$ satisfies $h \geq b \theta / \caL$. 
\end{prop}
\begin{proof}
Let $P_R(\beta)$ be the projection on $\varphi_R(\beta)$. By the Hellmann-Feynman theorem
\[
\dot{E}_R(\beta) = \tr{P_R(\beta) W};\quad \ddot{E}_R(\beta) = \tr{\dot P_R(\beta) W}.
\]
Since $\norm{H_R(\beta)-H_R(0)}\le \beta$, by  Weyl's theorem
\[\dist\pa{E_R(\beta), \sigma(H_R(\beta)) \setminus \set{E_R(\beta)}} \geq \dist\pa{E_R(0), \sigma(H_R(0)) \setminus \set{E_R(0)}}-2\beta \geq \frac{\sqrt{\theta}}{2 \caL}\]
for  $\beta\in I$ and $\theta$ sufficiently small. Hence, by standard perturbation theory,
\[
\| \dot{P}_R(\beta) \| \leq \beta /{\dist\pa{E_R(\beta), \sigma(H_R(\beta)) \setminus \set{E_R(\beta)}}} \leq 2\beta \frac{ \caL}{\sqrt{\theta}}.
\]
We now estimate
\[\dot{E}_R(\beta) =\dot{E}_R(0)+\int_0^\beta \ddot{E}_R(s) ds\ge -\frac{C_\nu}{ \ln \theta }-2\beta^2 \frac{ \caL}{\sqrt{\theta}}\ge -\frac{C_\nu}{2 \ln \theta },\quad  \beta\in I,\]
using  \cref{def:crossing_con}(iv), $Rank(P_R)=1$, and $\norm{W}\le 1$ in the second step.  Hence \[E_R\pa{a}-{E}_R(0), {E}_R(0)-{E}_R\pa{-a}\ge 2b\frac{ \theta } \caL.\] Using  \cref{def:crossing_con}(ii), we see that  
$
h \geq  b \theta / \caL
$, completing the proof.
\end{proof}

\begin{lem}
For $b$ large enough, $\mathbb{P}(\caC) \geq c b \theta$ for some constant $c$ independent of $\theta$ and $b$.
\end{lem}
\begin{proof}
Let $\caC_k$ denote the event that the property (k) with $k=i,ii,iii,iv$ in \cref{def:crossing_con} holds. 
By \eqref{eq:thetaloc},  $\mathbb P\pa{\caC_i}\ge 1 - C \theta$.

If $H_R(0)$ is $(\nu, \theta)$ localizing and \eqref{eq:comple}  is satisfied for some interval $J$ and constant $c$,  it follows from Lemma~\ref{thm:sense} that there exists an eigenvalue $E_R$ of $H_R(0)$ and the associated eigenvector $\varphi_R$ such that $|\varphi_R(0)|^2 > -C_\nu /\ln(\theta)$. As shown in Lemma \ref{lem:posit_J}, \eqref{eq:comple} indeed holds deterministically with the choice $J=[\tfrac14,\tfrac{15}4]$, $c=\tfrac1{49}$. Thus we can pick $\caC_{iv}:=\caC_i$.  

To bound $\mathbb P\pa{\caC_{ii}}$ we will invoke
\begin{thm}[Two-sided Wegner estimate]\label{thm:Weg}
Let $K\subset \Z$ be an interval. Then  for any compact subinterval $J$ of $(0,4)$ there exist $L_0>0$ and constants $C_+\ge C_->0$ such that we have
\be\label{eq:Weg}
C_-\abs{J}\abs{K}\le \mathbb E\pa{\tr \chi_J\pa{H^K}}\le C_+\abs{J}\abs{K},
\ee
provided $\abs{K}>L_0$.
\end{thm}
\begin{proof}
The upper bound is well known, see e.g.,   \cite[Corollary 4.9]{AW}. The lower bound was recently established in \cite[Theorem 1.1]{Geb} in the continuum setting, but the same proof works for the lattice systems considered here as well.
\end{proof}
We will also need the following extension of the upper Wegner bound, known as the Minami estimate: 
\begin{thm}[Minami estimate]\label{thm:Min}
Under the same assumptions as in Theorem \ref{thm:Weg}, for any $n\in\mathbb N$ we have 
\be\label{eq:Min}
 \mathbb P\pa{\tr \chi_J\pa{H^K}\ge n}\le \frac1{n!}\pa{C_+\abs{J}\abs{K}}^n.
\ee
\end{thm}
\begin{proof} In this generality, the bound goes back to \cite{CGK1}, see also  \cite[Theorem 17.11]{AW}. \end{proof}
Let $\check I:=[ E_R(0)-b \theta/\caL,E_R(0)+ b \theta/\caL]$. Combining the lower bound in \eqref{eq:Weg} with \eqref{eq:Min} and using the statistical independence of $H_L$ and $H_R(0)$, we see that 
\be
 \mathbb P\pa{\tr \chi_{\check I}\pa{H_L}\ge 1}\ge  \mathbb E\pa{\tr \chi_{\check I}\pa{H_L}}-\sum_{n=2}^\infty \pa{n-1} \mathbb P\pa{\tr \chi_{\check I}\pa{H_L}\ge n}\geq c b \theta
\ee
 for some $b$-independent constant $c>0$. This implies that $\mathbb P\pa{\caC_{ii}} \geq c b \theta$ for such $b$.

This leaves us with estimating $\mathbb P\pa{\caC_{iii}}$. Let $\hat J: = \{\lambda\in\mathbb R:\ \abs{\lambda-E_R(0)}\le 2\sqrt{\theta}/\caL\}$. Then $J\subset \hat J$  for $J$ specified in \eqref{eq:J} and, using the statistical independence of $H_L$ and $H_R(0)$, by \eqref{eq:Min} 
\be
\mathbb P\pa{\tr \chi_{J}\pa{H_L}\ge 2}\le\mathbb P\pa{\tr \chi_{\hat J}\pa{H_L}\ge 2}\le C\theta.
\ee
To complete the argument, we will use the following consequence of 
Theorem \ref{thm:Min}.  
\begin{thm}\label{thm:Lspac}
 Let $\delta>0$ and let $\mathcal E_\omega$ be an event \[\mathcal E_{\omega}:=\set{\sigma(H^{K})\mbox{ is } \delta\mbox{-level spaced on } \Lambda}.\]  
Then there exists $C>0$ such that 
\[\mathbb P\pa{\mathcal E_\omega}\ge 1-C\delta \abs{K}^2.\]
\end{thm}
\begin{proof}This statement is  essentially  \cite[Lemma 2]{KM}, in the formulation given in \cite[Lemma B.1]{EK}.\end{proof}
Applying this  with the choice $K=\Lambda_R$, we deduce that 
\be
\mathbb P\pa{\tr \chi_{J}\pa{H_R}\ge 2}\le C\sqrt{\theta}/\ln^2\caL\le {\theta}
\ee
for $\caL$ large enough. 
This yields $\mathbb P\pa{\caC_{iii}}\ge 1 - C \theta$. 

Putting our bounds on $\pa{\caC_{i}}$--$\pa{\caC_{iv}}$ together, we see that  for $b$ large enough $\mathbb{P}(\caC) \geq c b \theta$ for some constant $c>0$.

\end{proof}

\subsection{Construction of the avoided crossing}
In addition to $\omega \in \caC$, we will assume further properties of $\omega$ that will allow us to use perturbation theory to study the crossing. 
\begin{defn}\label{def:hyb_suc}
Let $\Lambda_B$ be a region of size $\caL^{1/8}$, centered at the boundary between $\Lambda_L$ and  $\Lambda_R$, i.e. (recall \eqref{eq:extbound}--\eqref{eq:ellins}) $\Lambda_B =  \pa{\partial \Lambda_R}_{\caL^{1/8}}$. We pick $b_L, b_R\in\Z$ so that  $\Lambda_B=(b_L, b_R)$. We denote by $H_B := H^{\Lambda_B}$ the Hamiltonian restricted to this region. 
We will say that $\omega \in \caA$ if $\omega \in \caC$ and the following items hold true
\begin{enumerate}
\item $H_B$ has no spectrum in the interval $\hat J:=(E_L - \theta^{-1}\caL^{-1/2}, E_L +\theta^{-1} \caL^{-1/2})$.
\item There are at most two eigenvalues of $H(0)$ in the interval $J$ defined in \eqref{eq:J}.
\item The centers of $\varphi_{L}$ and $\varphi_{R}$ are a distance of order $\sqrt\caL/\ln\caL$ away from the boundary of $\Lambda_{R}$. Specifically,
\be \label{eq:supportphi}
\norm{ \chi_{\{\abs{x} >\sqrt\caL/(4\ln\caL)\}} \varphi_{R} } + \norm{ \chi_{\{x>-3\sqrt\caL/\ln\caL\}} \varphi_{L} }  \leq e^{-c \sqrt\caL/\ln\caL}.
\ee
\item For $\lambda \in \hat J$,
\[
\norm{ \chi_{\{\abs{x-l}\ge \caL^{1/8}\}}(H_B - \lambda)^{-1} \delta_l} \leq e^{-c \caL^{1/8}}.
\] 
\item For $\sharp=L,R$, 
\[\abs{\langle \delta_{r}, \pa{ H_B - E_\sharp}^{-1}\delta_{l}\rangle-1}\ge 2\theta^{\frac{1}{4}}.\]
We note that condition  (i) above ensures that the resolvents in (iv)-(v) are well-defined.  
\end{enumerate}
\end{defn}

The dependence on the parameter $\theta$ in the above definition is chosen so that  $\mathbb P\pa{\caA}=O\pa{\theta}$. We will establish this at the end of the section.

Let $\varphi_R(\beta)$   be an eigenvector of $H_R(\beta)$, which is an analytic continuation of $\varphi_R(0)$. (Note that $H_{R}(\beta)$ is a  finite rank operator, so its eigenvectors do have  analytical continuation on the real line, c.f. \cite{Kato}). We recall that $H_L(\beta)$ is $\beta$-independent, so  $\varphi_{L}(\beta) \equiv \varphi_L$. We first show that the analogue of \eqref{eq:supportphi} holds if we replace $\varphi_R(0)$ with  $\varphi_{R}(\beta)$. For an interval $J$, we set $J_a := a J$. 

\begin{lemma}\label{loc_pert_loc}
Assume that $\omega \in \caA$. For $ \beta \in  I$ defined in \eqref{eq:I},
\be\label{eq:suppdef}
\norm{ \chi_{\{\abs{x}> \sqrt\caL/(2\ln\caL)\}} \varphi_{R}(\beta) } \leq e^{-c \sqrt\caL/\ln\caL}.
\ee
\end{lemma}
\begin{proof} Let $\hat H_R(0)=H_R(0)+(1-E_R(0))P_R(0)$, where $P_R(0)$ is an orthogonal projection onto $Span( \varphi_{R}(0) )$. We observe that by \cref{def:crossing_con}(iii) and $\abs{E_R(\beta)-E_R(\beta)}\le\beta$, 
\be\label{eq:norr}
\norm{\pa{\hat H_R(0)-E_R(\beta)}^{-1}}\le \frac{2 \caL}{\sqrt{\theta}},\quad \beta\in I. 
\ee 
We have
\[
\begin{aligned}
&\chi_{\{\abs{x}> \sqrt\caL/(2\ln\caL)\}} \varphi_{R}(\beta) =\chi_{\{\abs{x}> \sqrt\caL/(2\ln\caL)\}} \pa{\hat H_R(0)-E_R(\beta)}^{-1}\pa{\hat H_R(0)-E_R(\beta)}\varphi_{R}(\beta)\\&\hspace{1cm} =\chi_{\{\abs{x}> \sqrt\caL/(2\ln\caL)\}}  \pa{\hat H_R(0)-E_R(\beta)}^{-1}\pa{(1-E_R(0))P_R(0)+\beta W}\varphi_{R}(\beta).
\end{aligned} \]
To estimate the right hand side, we note that \[\norm{\chi_{\{\abs{x}> \sqrt\caL/(4\ln\caL)\}}\pa{P_R(0)+\beta W}}\le \e^{-c \sqrt\caL/\ln\caL}\] by  \eqref{eq:supportphi} and the compactness of $\supp(W)$. Hence \eqref{eq:suppdef} will follow once we show that
\[\norm{\chi_{\{\abs{x}> \sqrt\caL/(2\ln\caL)\}}  \pa{\hat H_R(0)-E_R(\beta)}^{-1}{\chi_{\{\abs{x}\le \sqrt\caL/(4\ln\caL)\}}}} \leq e^{-c \sqrt\caL/\ln\caL}.\]
 The latter bound is a consequence of the spectral theorem, the estimate \eqref{eq:norr}, and  the fact that $H_{R}(0)$ (and hence $\hat H_R(0)$) is  $\pa{\nu,\theta}$-localizing for  $\omega \in \caA$.

\end{proof}

We recall that $P_{dec}(\beta)$ denotes the orthogonal projection onto  $Span\pa{\varphi_{L}, \varphi_R(\beta)}$. By standard perturbation theory, $P_{dec}(\beta)$ is a spectral projection of $H_{dec}(\beta)$ for all $\beta \in  I$. We first establish that $P_{dec}(\beta)$ is close to a spectral projection of $H(\beta)$.

\begin{prop}\label{prop:full_s}
Assume that $\omega \in \caA$.  Then for $\beta\in I$ (recall \eqref{eq:J} and \eqref{eq:I}) we have 
\begin{thmlist}
\item\label{it:ta1}
 $\sigma(H(\beta))\cap J=\set{E_-(\beta),E_+(\beta)}$ where $E_\pm(\beta)$ are real analytic in $\beta$;
\item\label{it:ta2}  $\dist(\set{E_-(\beta),E_+(\beta)},\set{E_{L},E_R(\beta)}) \leq \e^{-c \sqrt\caL/\ln\caL}$; 
\item\label{it:ta3} Let $P(\beta)$ be the spectral projection on $E_{\pm}(\beta)$, then $\|{P}(\beta) -  P_{dec}(\beta)\|\le \e^{-c \sqrt\caL/\ln\caL}$;
\item\label{it:ta4} We can label $E_{\pm}(\beta)$ so that the associated eigenfunctions $\varphi_{\pm}(\beta)$ satisfy \[|\langle \varphi_{-}(0),\varphi_R(0) \rangle|^2 \leq e^{-c\sqrt\caL/\ln\caL},\quad |\langle \varphi_{+}(0),\varphi_L \rangle|^2 \leq e^{-c\sqrt\caL/\ln\caL}.\]
\end{thmlist}
\end{prop}
\begin{proof}
By Lemma~\ref{outbad}, \eqref{eq:supportphi}, and  Lemma~\ref{loc_pert_loc}  we deduce that \[\dist\pa{\sigma(H(\beta)),E_\sharp(\beta)}\le  \e^{-c \sqrt\caL/\ln\caL},\quad \sharp=L,R.\] It follows that $H(\beta)$ has at least two eigenvalues in the interval $I$. Combined with  standard perturbation theory and the fact that for $\omega\in\caA$ the operator $H(0)$ has at most two eigenvalues in  $J$, see Definition~\ref{def:hyb_suc}(ii), we see that  Proposition \ref{it:ta1}--\ref{it:ta3} holds. The last statement follows from \eqref{eq:supportphi},  Lemma~\ref{loc_pert_loc}, and Lemma~\ref{outbad}.
\end{proof}

\begin{prop}\label{prop:avoided}
Suppose that $\omega \in \caA$, then the eigenvalues $E_\pm(\beta)$ cannot intersect  each other in the interval $I$. 
\end{prop}

We start with the following preliminary observation.

\begin{lemma}\label{lem:gapco}
The operator 
$
\bar P_{dec}(\beta) \pa{H(\beta)-\lambda} \bar P_{dec}(\beta)
$
is invertible on the range of $\bar P_{dec}(\beta)$ for all $\lambda\in J$ and $\beta\in I$, and the norm of the inverse is bounded by $C \caL/\sqrt{\theta}$.
\end{lemma}
\begin{proof}
It is a standard result in perturbation theory that if $B$ is invertible and $\|B^{-1}\|  \|(A- B) \| <1$, then $A$ is invertible and 
$$
\|A^{-1}\| \leq \frac{\|B^{-1}\|}{1 - \|B^{-1}\| \|A-B\|}.
$$
To prove Lemma \ref{lem:gapco}, we combine this observation with 
$$B = \bar{P}(\beta) (H(\beta) - \lambda) \bar{P}(\beta) + P(\beta), \quad  A = \bar P_{dec}(\beta) \pa{H({\beta})-\lambda} \bar P_{dec}(\beta) + P_{dec}(\beta).$$
By Proposition~\ref{prop:full_s}, $\|A - B\| \leq e^{-c \sqrt\caL/\ln\caL}$. By $\omega \in \caA$, $B^{-1}$ is invertible with 
$$
\|B^{-1}\| \leq C \frac{\caL}{\sqrt{\theta}}.
$$
We now note that $A$ is block diagonal with respect to $P_{dec}(\beta),\bar P_{dec}(\beta)$, and that its inverse exists if and only if each associated block has an inverse.
\end{proof}

\begin{proof}[Proof of Proposition~\ref{prop:avoided}]
We will suppress the $\beta$ dependence and use the shorthand  $P$ for $P_{dec}(\beta)$ in this proof. 
Here, the idea is to use Schur complementation. Namely, given $\lambda\in J$, we consider $M=M(\beta,\lambda)$, the Schur complement of $H$ in $Ran(\bar P)$, defined as
\[M:=P\pa{H-\lambda}P-PH \bar P\pa{\bar P \pa{H-\lambda} \bar P}^{-1}\bar PH P.\]
We note that by Lemma~\ref{lem:gapco},  $M$ is well-defined for our range of $\lambda$s and $\beta$s. $M$ is a  rank-two operator whose  range is spanned by $(\varphi_R, \varphi_L)$.  Using the Guttman rank additivity formula,  \cite[14]{Z}, we see that $\tr {\chi_{\set{\lambda}}(H)}=2$ (a sufficient and necessary condition for  the intersection of two eigenvalues)  if and only if $M=0$. In particular, the non-intersection property will follow if we can show that in this range we have $M_{LR} = \langle \varphi_L, M \varphi_R \rangle \neq0$. 
We claim that 
\be
\label{eq:MLR}
M_{LR}  = \varphi_L(l) \varphi_R(r) \left(1 -  \langle \delta_{r},(H_B - E_-)^{-1}\delta_{l}\rangle + Error \right),
\ee
where $|Error| \leq \theta^2$. Since  $\omega \in \caA$, by \cref{def:hyb_suc}(v) we have 
\[\abs{\langle \delta_{r}, \pa{ H_B - E_-}^{-1}\delta_{l}\rangle-1}\ge \theta^{\frac{1}{4}}.\]  Hence, for sufficiently large $\caL$, $M_{LR}\neq0$ as  the eigenfunctions of $H_{L,R}$ cannot vanish at the respective boundary points. 

It remains to derive \eqref{eq:MLR}. We recall that  $ \Gamma := \Gamma_{lr} + \Gamma_{rl}$ is the hopping term connecting the region  $\Lambda_R$ to the region $\Lambda_L$. In particular, $\Gamma \varphi_{L} = \varphi_{L}(l) \delta_{r}$ and $\Gamma \varphi_{R} = \varphi_{R}(r) \delta_{l}$. We use these equations to evaluate the terms in  
\[M_{LR}=\langle\varphi_{L}, \pa{H-\lambda} \varphi_{R}\rangle-\langle\varphi_{L},PH \bar P\pa{ \bar{H}-\lambda}^{-1}\bar PH P\varphi_{R}\rangle,\]
where we denote  $\bar H= \bar P H \bar P$, and let $\pa{\bar H-\lambda}^{-1}$ denote the inverse of $\bar H-\lambda$ on the $Ran\pa{\bar P}$. 
The first term is equal to
\[\langle\varphi_{L},H\varphi_{R}\rangle=\langle\varphi_{L},\Gamma \varphi_{R}\rangle=\varphi_{L}(l)\varphi_{R}(r).\]
To evaluate the second term, we use the identity $\bar PH P = \bar{P} \Gamma P$ to get 
\[
\langle\varphi_{L},PH \bar P\pa{\bar H-\lambda}^{-1}\bar PHP\varphi_{R}\rangle = \varphi_{L}(l)\varphi_{R}(r) \langle \delta_{r},\pa{\bar H-\lambda}^{-1}\delta_{l}\rangle.
\]
We next use the resolvent identity
\[
\pa{\bar H-\lambda}^{-1} = \pa{H_B - \lambda}^{-1} +T , \quad T := \pa{\bar H-\lambda}^{-1} (H_B - H +\bar{P}HP)\pa{H_B - \lambda}^{-1}.
\]
We note that since  $\omega \in \caA$, by \cref{def:hyb_suc}(iii) the resolvent  $(H_B - \lambda)^{-1}$ is well-defined and its norm is bounded by $C \caL^{1/4}$. Moreover, since $(H_B - H)\chi_{\{\abs{x-l}< \caL^{1/8}\}}=0$, by  \cref{def:hyb_suc}(iv), \eqref{eq:supportphi}, and Lemma \ref{loc_pert_loc} we get
\[
\begin{aligned}
\norm{T\delta_{l}}\le 5&\norm{\pa{\bar H-\lambda}^{-1}}\norm{ \chi_{\{\abs{x-l}\ge \caL^{1/8}\}}(H_B - \lambda)^{-1} \delta_l}\\ +&\norm{\pa{\bar H-\lambda}^{-1}}\norm{P\chi_{\{\abs{x-l}< \caL^{1/8}\}}}\norm{ (H_B - \lambda)^{-1}}\le C e^{-c \caL^{1/8}},
\end{aligned}
\]
which implies  that 
\[
\abs{\langle \delta_{r},T\delta_{l}\rangle} \leq C e^{-c \caL^{1/8}}.
\]
Furthermore, by standard perturbation theory and \cref{def:hyb_suc}(iii),
\[
\norm{\pa{H_B - \lambda}^{-1} - \pa{H_B - E_-}^{-1}} \leq C |E_- - \lambda| \theta^{2}\caL.
\]
Since $E_- - \lambda$ is of order $\caL^{-1}$ for $\lambda\in J$, we get \eqref{eq:MLR}.
\end{proof}

We now show
\begin{lem} \label{lem:prob}
$\mathbb{P}(\caA) \geq c \theta$ for some constant $c$.
\end{lem}
\begin{proof} 
Let $\caA_k$ denote the event that property (k) in \cref{def:hyb_suc}  holds. 

Using the upper bound in \eqref{eq:Weg}, we get $\mathbb{P}(\caA_{i}) \geq    \mathbb{P}(\caC)- C\theta^{-1}\caL^{-1/2}\caL^{1/8}\ge  c b \theta$ for $\caL$ large enough. On the other hand, using \eqref{eq:Min}, we deduce that 
\[
\mathbb{P}(\caA_{ii}\cap \caA_{i})\ge \mathbb{P}(\caA_{i})- \mathbb P\pa{\tr \chi_J\pa{H(0)}\ge 3}\ge \mathbb{P}(\caA_{i})-C\theta^{3/2}\le  c b \theta.
\]
Let $\hat\Lambda_L=[-4\sqrt\caL/\ln\caL,l]$. Then, using the upper bound in \eqref{eq:Weg},  for $\caL$ large enough,
\[
 \mathbb P\pa{\tr \chi_{\hat J}\pa{H^{\hat\Lambda_L}}=0}\ge 1-C(\sqrt\caL  /\ln\caL)\theta^{-1}\caL^{-1/2}\ge 1-\theta^2.
\]
Let $\caE:=\caA_{ii}\cap \caD$, where $\caD$ is the event $\tr \chi_{\hat J}\pa{H^{\hat\Lambda_L}}=0$. Then we see that \[\mathbb{P}(\caE) \geq  \mathbb{P}(\caA_{ii})-\theta^2\ge  c b \theta.\]

 We claim that \eqref{eq:supportphi} holds for $\omega\in \caE$, implying that $\mathbb{P}(\caA_{iii}\cap \caA_{ii}) \geq    c b \theta$. Indeed, the bound $\norm{ \chi_{\{\abs{x} >\sqrt\caL/(4\ln\caL)\}} \varphi_{R} }  \leq e^{-c \sqrt\caL/\ln\caL}$  follows directly from \cref{def:crossing_con}, parts (i,iv) (we recall that $\caA\subset\caC$). On the other hand, if the localization center for $\varphi_L$ were located in $[-\frac72\sqrt\caL/\ln\caL,l]$, \cref{def:crossing_con}(i) would imply that $\norm{\chi_{x<-4\sqrt\caL/\ln\caL}\varphi_L}\le \e^{-c\sqrt\caL/\ln\caL}$. But then we would have $\dist\pa{E_L,\sigma(H^{\hat\Lambda_L)}}\le \e^{-c\sqrt\caL/\ln\caL}$ thanks to  Lemma~\ref{outbad}, contradicting $\tr \chi_{\hat J}\pa{H^{\hat\Lambda_L}}=0$. This implies that  the localization center for $\varphi_L$ is located in $\Lambda_L\setminus [-\frac72\sqrt\caL/\ln\caL,l]$, which in turn implies that $\norm{ \chi_{\{x>-3\sqrt\caL/\ln\caL\}} \varphi_{L} }  \leq e^{-c \sqrt\caL/\ln\caL}$ by \cref{def:crossing_con}(i).
 
 To estimate $\mathbb{P}(\caA_{iv}\cap \caA_{iii})$, we note that  our assumptions on randomness imply 
\[
\sup_{\lambda\in\mathbb R} \mathbb E\norm{ \chi_{\{\abs{x-l}\ge \caL^{1/8}\}}(H_B - \lambda-i0)^{-1} \delta_l} \le Ce^{-c \caL^{1/8}},
\]
\cite[Theorem 12.11]{AW}. Hence, denoting 
\[
\caF:=\set{\omega\in\Omega:\ \norm{ \chi_{\{\abs{x-l}\ge \caL^{1/8}\}}(H_B - \lambda)^{-1} \delta_l} \leq e^{-c \caL^{1/8}}}, 
\] 
we see that $\mathbb{P}(\caA_{iv}\cap \caA_{iii})\ge  c b \theta$ for $\caL$ large enough by Markov's inequality.

Finally, the bound $\mathbb{P}(\caA_{v}\cap \caA_{iv})\ge  c b \theta$ is a direct consequence of 
\begin{lemma}\label{lem:algset}
For a fixed  $s \in (0,1/2)$ and  $\lambda\in I$, we have 
\[\mathbb{P}\pa{\set{\omega\in\Omega: \ \abs{\langle \delta_{r}, \pa{ H_B - E}^{-1}\delta_{l}\rangle-1}\ge \theta^{\frac{1}{s}}}} \ge 1- C_s\theta.\]
\end{lemma}
\end{proof}

\begin{proof}[Proof of Lemma \ref{lem:algset}]
Let $G(x,y):=\langle \delta_x,\pa{ H_B - \lambda}^{-1}\delta_y\rangle$. 
We first observe that, thanks to the geometric resolvent identity (or directly by \cite[Eq.~ 12.7]{AW}), 
\be\label{eq:G_bdec}
G(l,r)= \hat G(l,l)G(r,r),
\ee
where $\hat G(x,y)=\langle \delta_x,\pa{\hat H_B - \lambda}^{-1}\delta_y\rangle$ and $\hat H_B$ is obtained from $ H_B$ by the removal of the $( l,r)$ bond, i.e., $\hat H_B=H_B-\Gamma_{(l,r)} - \Gamma_{(r,l)}$. 
We use the resolvent identity
$$
\tilde{G}(r,r) = G(r,r) - \tilde{G}(r,r) \hat{G}(l,l) G(r,r)
$$
to obtain
\[\frac{1}{\hat G(l,l)G(r,r)-1}=-\frac{\tilde G(r,r)}{ G(r,r)},\]
where $\tilde G(x,y):=\langle \delta_x,\pa{H_B+\hat G(l,l)\chi_{\set{r}} - \lambda}^{-1}\delta_y\rangle$. 
 An important fact to note here is that $\hat G(l,l)$ is independent of the $\omega_r$ random variable. This independence allows us to conclude that
\[\mathbb E_{\omega_1}\abs{\tilde G(r,r)}^s\le C_s,\quad s\in (0,1).\]
On the other hand, under our conditions on the probability distribution $\mu$, we also have (see \cite[Theorem 12.8]{AW}
\[\mathbb E\abs{ G(r,r)}^{-s}\le C_s,\quad s\in (0,1).\]
Combining these two bounds and using the H\"older inequality, we deduce that
\[\mathbb E\abs{\frac{1}{\hat G(l,l)G(r,r)-1}}^s\le C_s,\quad s\in (0,1/2),\]
from which the assertion follows by the Markov inequality.
\end{proof}

\subsection{Proof of Theorem~\ref{thm:A-main}} 
\begin{thm}\label{thm:full_s}
Let us denote by $\tilde \Omega_{F, \caL} \subset \Omega$ all realizations for which  $H(\beta)$ $F$-hybridize.
Let $\omega \in \caA$ and $F < 1/2$. Then for $\caL$ large enough, $\omega \in \tilde \Omega_{F, \caL}$.
\end{thm}
\begin{proof} Consider the analytical family of eigenvectors $\varphi_{R}(\beta)$,  $\varphi_L$ of $H_{dec}(\beta)$ and the analytical family $\varphi_\pm(\beta)$ of eigenvectors of $H(\beta)$. We will show that $\varphi(\beta):=\varphi_+(\beta)$ is an analytical family whose existence is required in  Definition~\ref{def:hybridization} of $\Omega_{F, \caL}$. We recall that the families are labeled in such a way that at $\beta=0$, $\varphi_+$ has exponentially small overlap with $\varphi_L$. In particular, $\varphi_+(0)$ satisfies item (i) in Definition~\ref{def:hybridization}.

  By Proposition~\ref{prop:full_s}, the families satisfy \eqref{eq:proxim} with $\varepsilon = \e^{-c \sqrt\caL/\ln\caL}$. Proposition~\ref{prop:crossing} implies that the bandwidth of the crossing satisfies $h > 4 \varepsilon$. It then follows from Lemma~\ref{lem:avoided_crossing} that there exists $\beta$ such that 
$$
\varphi_+(\beta) = c_L^+(\beta) \varphi_L + c_R^+(\beta) \varphi_R + \varphi^\perp,
$$
with 
$$
|c^+_L(\beta)|^2 = |c^+_R(\beta)|^2 \geq  \frac{1- \epsilon^2}{2}.
$$
It follows that item (ii) of Definition~\ref{def:hybridization} is satisfied for any $F<1/2$, provided $\caL$ is large enough.
\end{proof}

As a corollary of the above result and Lemma~\ref{lem:prob}, we get that for any $F < 1/2$, \[\liminf_{\caL \to \infty} \mathbb{P}(\tilde \Omega_{F, \caL}) > 0.\]
The assertion of Theorem~\ref{thm:A-main} is established completely analogously, by  splitting $\Lambda_{full}$ into $\Lambda_L$, $\Lambda_R$, and $-\Lambda_L$, and then repeating the same steps as above. The reason that we present a proof for the asymmetric region is related to the fact that, in this case, the boundary of $\Lambda_R$ consists of a single point $r$, whereas in the symmetric case it consists of two points $\pm r$, making the presentation slightly more cumbersome. \qed

\section{A  Wannier basis for  quasi-local projections}\label{sec:Wannier}
Here, we show the existence of a (generalized) Wannier basis, consisting of exponentially localized functions, for a  rank $m$ orthogonal projection $P$ on $\ell^2(\Z^d)$ that satisfies the quasi-locality property \eqref{eq:detconthetadeg} below. 
The motivation for constructing such a basis is related to the fact that it allows showing the localization property \eqref{eq:globstr} without assuming spectrum simplicity.

To illustrate the idea behind this construction, we start with the case $m=1$. 
\begin{lemma}\label{lem:cent_mass}
Suppose that  the normalized vector $\psi\in \ell^2\pa{ \mathbb{T}_L}$ satisfies 
\be\label{eq:detcontheta}
\max_{x,y\in \mathbb{T}_L}\pa{\abs{\psi(x)}\abs{\psi(y)}\e^{c\abs{x-y}}}\le \frac{1}\theta.
\ee 
Then, for any sufficiently small (but $L$-independent) $\theta$, we have 
$\|\psi\|_\infty^2 \ge  \abs{\ln \theta}^{-d-1}$, and there exists $x_o\in\T_L$ such that \[\abs{\psi(y)}\le \frac{\abs{\ln \theta}^{\frac{d+1}2}}{\theta}\e^{-c\abs{y-x_o}},\quad \forall y\in \mathbb{T}_L.
\]
\end{lemma}
\begin{proof}[Proof of Lemma \ref{lem:cent_mass}]
The second bound is an immediate consequence of the first with a (non unique, in general)  choice of $x_o$ such that $\abs{\psi(x_o)}=\|\psi\|_\infty$, so we only need to show that $\|\psi\|^2_\infty\ge  \abs{\ln \theta}^{-d-1}$. Let $r=r(c,\theta)>0$ be such that \[\sum_{{y\in \Z^d:\ \abs{y}>r}}\e^{-2c\abs{y}}\le \frac{\theta^2 \|\psi\|_\infty^2}2.\] In particular, for a fixed $c$ there exists $C$ such that we can choose $r=-C\ln\pa{ {\theta\|\psi\|^2_\infty}}$ for $\theta$ sufficiently small. 
Then by \eqref{eq:detcontheta} we can bound 
\begin{multline}\label{eq:bootst'}
1=\sum_{x\in  \mathbb{T}_L}\abs{\psi(x)}^2\le {\|\psi\|^2_\infty} \sum_{\substack{x\in  \mathbb{T}_L:\\ \abs{x-x_o}\le r}}1+ \sum_{\substack{x\in  \mathbb{T}_L:\\ \abs{x-x_o}> r}}\frac{\e^{-2 c \abs{x-x_o}}}{{\|\psi\|^2_\infty}\theta^2} \le {\|\psi\|^2_\infty} (2r+1)^d+\frac12.
\end{multline}
This implies that ${\|\psi\|^2_\infty}\ge \frac1{2(2r+1)^d}$ or, in view of the definition of $r$, ${\|\psi\|^2_\infty}\ge u$, where $u$ is a unique positive solution of 
\be\label{eq:M}
\e^{-C u^{1/d}}= \theta u^{2}.
\ee
Since $u>\abs{\ln\theta}^{-d-1}$ for $\theta$ sufficiently small, we get ${\|\psi\|^2_\infty}\ge  \abs{\ln \theta}^{-d-1}$. 
\end{proof}
While considering the rank one projection $P$ is sometimes enough for random operators (e.g., for the randomness given by the rank one single site potential as in the standard Anderson model), in general it is not known whether the spectrum of a random operator that satisfies Assumptions  \ref{assump:FRC}--\ref{assump:FMC} is a.s. \hspace{-9pt} simple  or even has finite multiplicities. For our applications, one needs to be able to decompose $P$ into a sum of rank one mutually orthogonal projections that individually exhibit exponential decay. Such a decomposition is called 
 a (generalized) Wannier basis for $P$. In general, finding a Wannier basis is a hard problem, due to a topological obstruction, see e.g., \cite{Pan}. Here, we assert its existence for a finite rank $P$ with explicit control over its rank  $m$, which is sufficient for our purposes.
\begin{thm}\label{lem:cent_mass_deg}
Let $m\in\N$, $\theta>0$ be such that $m^3\theta\ll1$. Suppose that a rank $m$ orthonormal projection $P\in \caL(\caH)$, $\caH=\ell^2\pa{ \Z^d}$ satisfies 
\be\label{eq:detconthetadeg}
\max_{x,y\in \Z^d}\pa{\abs{P(x,y)}\e^{c\abs{x-y}}}\le \theta^{-1}.
\ee 
Then we can decompose $P$ as $P=\sum_{i=1}^mP_i$, where  $P_i=\abs{\psi_i\rangle\langle\psi_i}$ are rank one mutually orthogonal projections that satisfy  
${\|\psi_i\|_\infty}\ge  \abs{\ln \theta}^{-d-1}$ and,  for some  $x_i\in\Z^d$, \[\abs{\psi_i(y)}\le {\theta}^{-2}\e^{-c\abs{y-x_i}/m},\quad y\in \Z^d.
\]
We stress that the constant $c$ here is $m$-independent.
\end{thm}
\begin{proof}
We will need some preparatory results. Using the argument  identical to the one used in 
Lemma \ref{lem:cent_mass} we obtain
\begin{lem}\label{lem:MQ}
Let $M=\max_{x\in\Z^d}P(x,x)$. Then there exists a ($\theta$-independent) $C>0$ such that $M\ge u$, where $u$ is a unique positive solution of \eqref{eq:M}. 
In particular,  for $\theta$ sufficiently small,   $M\ge  \abs{\ln \theta}^{-d-1}$. 
\end{lem}
Let $L=L(c,\theta)>0$ be such that 
\be\label{eq:conM}
\sum_{\Lambda^c_{L/4}(0)}\e^{-2c\abs{y}}\le {\theta^6 M}
\ee
 with $M$ as above. In particular,  there exists $C$ such that we can choose 
 \be\label{eq:L}
 L=-C\ln{\theta }
 \ee 
 for $\theta$ sufficiently small. 
Consider
\be\label{eq:boxesXi'L}
 \Xi_{L}:=  \pa{\tfrac 3 2L  \Z}^{d},
\ee
cf. \eqref{eq:boxesXi}, and an $L$-cover of $\Z^d$ of the form
\[\Z^d=\bigcup_{a \in  \Xi_{L}} {\Lambda}_{L}(a).\]
We note that for any $x\in\Z^d$ we can find $a \in  \Xi_{L}$ such that $\dist\pa{{\Lambda}_{L}^c(a),x}\ge L/4$.

\begin{lem}
For $L$ as above, let $T=\max_{a \in  \Xi_{L}}\tr\pa{P\chi_{\Lambda_L(a)}}$. Then $T\ge1/2$ for $\theta$ sufficiently small. 
\end{lem}
\begin{proof}
Suppose in contradiction that $\tr{P\chi_{\Lambda_L(a)}}< 1/2$ for any $a \in  \Xi_{L}$. Picking $x_o$  as in the previous lemma and letting $a\in  \Xi_{L}$ be such that $\dist\pa{{\Lambda}_{L}^c(a),x_o}\ge L/4$, we have
\[M\le P(x_o,x_o)\sum_{y\in\Lambda_{L}(a)}P(y,y)+\sum_{y\in\Lambda_{L}^c(a)}\abs{P(x_o,y)}^2\le M\sum_{y\in\Lambda_{L}(a)}P(y,y)+\theta^4 M<2M/3,\]
a contradiction.
\end{proof}
We now observe that since $\tr{P}=m$, the cardinality of a set 
\[\caS:=\set{a \in  \Xi_{L}:\ \tr{P\chi_{\Lambda_L(a)}}\ge1/2}\]
cannot exceed $2\cdot 3^d m$ as each box $\Lambda_L(a)$ can overlap with at most $3^d$ other boxes. 

Let $\caR:=\cup \Lambda_L(a)$, where the union is taken over boxes with $a\in\caS$ and boxes that overlap with them. We  note that if $y\notin\caR$, then 
\be\label{eq:unP}
P(y,y)< {2M}\theta^4
\ee for $\theta$ sufficiently small. Indeed, if $y\notin\caR$, then $\dist\pa{y,\cup_{a\in\caS} \Lambda_L(a)}\ge L/2$. In particular, 
\[P(y,y)\le P(y,y)\sum_{z\in\Lambda_{L/2}(y)}P(z,z)+\sum_{z\in\Lambda_{L/2}^c(y)}\abs{P(z,y)}^2\le \frac12P(y,y)+\theta^4 M,\]
which yields \eqref{eq:unP}. 
\begin{lemma}\label{lem:QP}
Let $Q=P\chi_{\caR} P$. Then $Q$ is close to $P$, namely $\norm{P-Q}\le \theta^3$ for $\theta$ sufficiently small. In particular, $Q$ is invertible as an operator on $Ran(P)$, with $Q\ge 1-\theta^3$.
\end{lemma}
\begin{proof}
We have
$Q^2=Q-P\chi_{\caR^c} P\chi_{\caR} P$ 
and
\[
\begin{aligned}
\norm{\chi_{\caR^c}P\chi_{\caR}}_{HS}= \sum_{y\in\caR_1^c,x\in\caR}\abs{P(x,y)}^2 &=\sum_{0<\dist\pa{y,\caR}\le L/2,x\in\caR}\abs{P(x,y)}^2\\&+\sum_{\dist\pa{y,\caR} > L/2,x\in\caR}\abs{P(x,y)}^2.
\end{aligned}
\]
The first term can be estimated  by $CmM^2\theta^4\abs{\ln\theta}^d\le \theta^3/2$ using $\abs{P(x,y)}^2\le P(x,x)P(y,y)$ and \eqref{eq:unP}. For the second sum, we use \eqref{eq:conM} to bound it by $CmM\theta^4\abs{\ln\theta}^d <\theta^3/2$. This shows that 
\be\label{eq:caR}
\norm{\chi_{\caR^c}P\chi_{\caR}}_{HS}\le \theta^3,
\ee
so $\norm{Q^2-Q}_{HS}\le \theta^3$ for $\theta$ sufficiently small.

 We next observe that, in view of \eqref{eq:detconthetadeg}, 
\be\label{eq:kerQ}
\abs{Q(x,y)}=\abs{\sum_{z\in\caR}P(x,z)P(z,y)}\le C\theta^{-2}\e^{-c\abs{x-y}}
\ee
by the properties of exponential sums. Let $\bar Q=P-Q$. Then $\bar Q$ is (a) close to be a projection on $Ran(P)$ and (b) $\abs{\bar Q(x,y)}\le C\theta^{-2}\e^{-c\abs{x-y}}$. Indeed, (a) follows from \[\tilde Q^2=P-2Q+Q^2=\tilde Q-(Q-Q^2)=\tilde Q+O(\theta^3),\]
while (b) follows directly from the decay properties of $P(x,y)$ and $Q(x,y)$. 

We next show that $\bar Q$ is close to zero, which implies the result. Indeed, suppose in contradiction that $\bar Q$ is close to a non-trivial projection, i.e., $\dist\pa{\sigma(\bar Q),1}=O(\theta^3)$.   Let $y_o\in\Z^d$ be such that $\bar M:=\max \bar Q(x,x)=\bar Q(y_o,y_o)$ for some $y_o$ which  is not necessary unique. Just as in the proof of Lemma \ref{lem:MQ}, let $\bar r=\bar r(c,\theta)>0$ be such that $\sum_{{y\in \Z^d:\ \abs{y}>\bar r}}\e^{-2c\abs{y}}\le {\theta^4 \bar M^2}$. In particular,  there exists $C$ such that we can choose $r=-C\ln\pa{\theta^2\bar M}$ for $\theta$ sufficiently small.

Essentially repeating the argument of Lemma \ref{lem:MQ}, we have
\[
\begin{aligned}
\bar M&=\bar Q(y_o,y_o)=\pa{\bar Q-\bar Q^2}(y_o,y_o)+(\bar Q^2)(y_o,y_o)\\& =O(\theta^3)+\sum_{y\in\Z^d}\abs{\bar Q(y_o,y)}^2=O(\theta^3)+\sum_{y\in\Lambda_r(y_o)}\abs{Q(y_o,y)}^2+\sum_{y\in\Lambda^c_r(x_o)}\abs{\bar Q(y_o,y)}^2\\ & \le O(\theta^3)+3^d\bar M^2 r^d.
\end{aligned}
\]
This yields $\bar M\le 3^{d+1}\bar M^2 r^d$, which in turn yields  $\bar M\ge \bar u$, where $u$ is implicitly given by the analogue of  \eqref{eq:M}. 
Since $\bar u>\abs{\ln\theta}^{-d-1}$ for $\theta$ sufficiently small, we get $\bar M\ge  \abs{\ln \theta}^{-d-1}$.  But then \eqref{eq:unP} implies 
\[\theta^4>P(y_o,y_o)=\pa{P\chi_{\caR^c}P}(y_o,y_o)+\pa{P\chi_{\caR}P}(y_o,y_o)\ge \pa{P\chi_{\caR^c}P}(y_o,y_o)=\tilde Q(y_o,y_o)>\theta,\]
a contradiction.

\end{proof}

Let $\caR=\cup_{i=j}^n \caR_i$ be a partition of $\caR$ into connected components. We note that $n\le 2m$, and that by construction,
\be\label{eq:caRd}
\dist_{i\neq j}\pa{\caR_i,\caR_j}\ge L/2
\ee
 We now introduce the operator
\be\label{eq:nop}
X=\sum_{j=1}^n j P \chi_{\caR_j} P,
\ee
which acts on $Ran(P)$. Clearly, $X$ is hermitian. 
\begin{lemma}
Let $\lambda\in\sigma(X)$. Then there exists $j\in\set{1,\ldots,n}$ such that $\abs{\lambda-j}\le \theta$ for $\theta$ sufficiently small.
\end{lemma}
\begin{proof}
For any $\lambda\in\sigma(X)$, we have 
\[\pa{X-\lambda}^2=\sum_{j=1}^n \pa{j-\lambda}^2 P \chi_{\caR_j} P+\sum_{j\neq j'}\pa{j-\lambda}\pa{j'-\lambda}P \chi_{\caR_j} P \chi_{\caR_{j'}} P.\]
The second sum can be bounded in norm by $n^2\theta^3$ using \eqref{eq:caRd} and \eqref{eq:conM}, while the first one satisfies
\[\sum_{j=1}^n \pa{j-\lambda}^2 P \chi_{\caR_j} P\ge \min_{j} \pa{j-\lambda}^2 Q\ge   \min_{j} \pa{j-\lambda}^2\pa{1-\theta^3}\]
using Lemma \ref{lem:QP}. But $0\in\sigma\pa{\pa{X-\lambda}^2}$, from which the result follows.
\end{proof}
The assertion of Theorem \ref{lem:cent_mass_deg} will follow from
\begin{lemma}
Let $(\lambda,\psi_\lambda)$ be an eigenpair for $X$ with  normalized $\psi_\lambda$. Then 
\be
\abs{\psi_\lambda(x)}\le C\theta^{-2}\e^{-c\,\dist\pa{x, \caR_{j_o}}},
\ee
where $j_o$ is chosen so that $\abs{\lambda-j_o}\le \theta$. 
\end{lemma}
\begin{proof}
Let 
\[Y_\lambda:=P\chi_{\caR_{j_o}} P+\sum_{j\neq j_o}\pa{j-\lambda} P \chi_{\caR_j} P,\ Z_\lambda:=P\chi_{\caR_{j_o}} P+\sum_{j\neq j_o}\pa{j-\lambda}^{-1} P \chi_{\caR_j} P.\] 
We have 
\[Y_\lambda\,Z_\lambda=P+\sum_{j\neq j'}f(j,j')\pa{j'-\lambda}P \chi_{\caR_j} P \chi_{\caR_{j'}} P=:P+W,\]
where $\abs{f(j,j')}\le 2n$ for all $j\neq j'$. We have $\norm{W}\le n^3\theta^3$ using \eqref{eq:caR}. Hence by standard perturbation theory, the operator $Y_\lambda$ is invertible on $Ran(P)$, with 
\be\label{eq:Y_lambda}
Y_\lambda^{-1}=Z_\lambda\pa{P+W}^{-1}=Z_\lambda\sum_{i=0}^\infty (-W)^i.
\ee
We now note that, analogously to \eqref{eq:kerQ}, 
\[\abs{Z_\lambda(x,y)}\le C\theta^{-2}\e^{-c\abs{x-y}},\]
while 
\[
\begin{aligned}
\abs{W(x,y)}&\le n^3\max_{j\neq j'}\abs{\sum_{z\in\caR_j,w\in\caR_{j'}}P(x,z)P(z,w)P(w,y)}\\ &\le Cn^3\theta^{-3}\e^{-c\abs{x-y}/2}\max_{j\neq j'}{\sum_{z\in\caR_j,w\in\caR_{j'}}\e^{-c\abs{z-w}/2}}\le \theta^2 \e^{-c\abs{x-y}/2}
\end{aligned}
\]
using \eqref{eq:caRd}, \eqref{eq:conM}, and \eqref{eq:L}. This in turn implies that 
\[ \abs{W^i(x,y)}\le \theta^i\e^{-c\abs{x-y}/2}, \quad i\in\N.\]
Using these bounds in \eqref{eq:Y_lambda}, we deduce that 
\[\abs{Y_\lambda^{-1}(x,y)}\le C\theta^{-2}\e^{-c\abs{x-y}/2}.\]
Hence we have
\[
\begin{aligned}\abs{\psi_\lambda(x)}&=\norm{\chi_{\set{x}}\psi_\lambda}=\norm{\chi_{\set{x}}Y_\lambda^{-1}Y_\lambda\psi_\lambda}\\ &=\norm{\chi_{\set{x}}Y_\lambda^{-1}\pa{Y_\lambda-X+\lambda}\psi_\lambda}\\&=\abs{1-j_o+\lambda}\norm{\chi_{\set{x}}Y_\lambda^{-1}P\chi_{\caR_{j_o}} P\psi_\lambda}\le C\theta^{-2}\e^{-c\,\dist\pa{x, \caR_{j_o}}}.
\end{aligned}\]
\end{proof}
We are now ready to complete the proof of Theorem \ref{lem:cent_mass_deg}. We pick the set $\{\psi_i\}$ to be $\set{\psi_\lambda}_{\lambda\in\sigma(X)}$, which is an orthonormal basis for $Ran(P)$ since $X$ is hermitian. Since \[\max_j\diam(\caR_j)\le 2mL=-mC\ln\theta,\]
picking some $x_{j}\in \caR_j$, we have 
\[\e^{-c\,\dist\pa{x, \caR_{j_o}}}\le  \e^{-c\pa{\abs{x-x_j}-2mL}}\le  \e^{-c\abs{x-x_j}/m} \mbox{ for } \abs{x-x_j}\ge 3mL.\]
On the other hand, since $\abs{\psi(x)}\le 1$ for all $x$, we can pick $c$ sufficiently small so that 
\[\e^{-c\abs{x-x_j}/m} \ge \theta^2  \mbox{ for } \abs{x-x_j}< 3mL,\]
and the assertion follows.
\end{proof}

\section{Auxiliary results}\label{ap:aux}
\begin{lemma}\label{lem:posit_J}
Let $H=-\Delta+V_\omega$ be the random operator on $\ell^2(\Z)$ with $V_\omega$ that satisfies assumptions introduced in \cref{sec:hyb}. Let $J=[\tfrac14,\tfrac{15}4]$ and $c=\tfrac1{49}$. Then 
\be
\sum_{E\in\sigma(H)\cap J}\abs{\psi_E(y)}^2\ge c,\quad y\in\Z,
\ee
and the same bound holds for any Dirichlet restriction $H^\Lambda$ of $H$. 
\end{lemma}
\begin{proof}
Let $P_J:=P_J(H)$. Suppose in contradiction that $\tr{\chi_{\set{y}}P_J}<c$ for some $y\in\Z$. Then we have 
\[
\tr{\chi_{\set{y}}\pa{H-2}^2}\ge \tr{\chi_{\set{y}}\pa{H-2}^2\bar P_J}\ge \tfrac{49}{16}\,\tr{\chi_{\set{y}}\bar P_J}> 3.
\]
However, the left hand side can be computed explicitly: $\tr{\chi_{\set{y}}\pa{H-2}^2}=2+V_\omega^2(y)\le3$, a contradiction. The proof for  $H^\Lambda$ is identical.
\end{proof}
\begin{thm}\label{thm:sense}
Assume that $H$ is $\pa{\nu,\theta}$-localized on $\Z$ and that there exists $c>0$ and a compact interval $J$ such that 
\be\label{eq:comple}
\sum_{E\in\sigma(H)\cap J}\abs{\psi_E(y)}^2\ge c,\quad y\in\Z.
\ee
 Then there exists $C_\nu>0$ and $E\in \sigma(H)\cap J$ such that $ \abs{\psi_E(0)}^2\ge \frac{-C_\nu}{\ln \theta}$ and $ \abs{x_E}\le \frac{-\ln \theta}{C_\nu}$. The same result holds for $H$ replaced by the finite volume Hamiltonian $H^\Lambda$, provided that $\abs{\Lambda}$ is sufficiently large, namely $\abs{\Lambda}\gg \abs{\ln \theta}$.
\end{thm}
\begin{proof}
We first observe that for any $L\in\mathbb N$ and $E\in\sigma(H)$ we have
\begin{multline}\label{eq:loceig'}
\sum_{\substack{y\in\Z:\\ \abs{y-x_E}\ge \frac12\pa{\abs{x_E}+L}}}\abs{\psi_E(y)}^2\le \frac{\langle x_E\rangle^2}\theta \sum_{\substack{y\in\Z:\\ \abs{y-x_E}\ge \frac12\pa{\abs{x_E}+L}}}\e^{-\nu\abs{y-x_E}}\\ = \frac{\langle x_E\rangle^2}\theta\sum_{\substack{u\in\Z:\\ \abs{u}\ge\frac12 \pa{\abs{x_E}+L}}}\e^{-\nu\abs{u}}=  \frac{\langle x_E\rangle^2}\theta\e^{-\frac\nu2 \pa{L+\abs{x_E}}}\,\frac2{1-\e^{-\nu}}\le \frac{C_\nu} \theta\e^{-\frac\nu2 \pa{L+\abs{x_E}}}
\end{multline}
for some $C_\nu>0$.

We next note that by the orthonormality of $\set{\psi_E}$ we have
\be\label{eq:orth}
\sum_{y\in\Z}\abs{\psi_E(y)}^2=1,\quad E\in\sigma(H).
\ee
Hence, using \eqref{eq:comple} and \eqref{eq:loceig'}, there exists $K_\nu>0$ such that
\begin{multline}\label{eq:upperb}
{4L+1}\ge\sum_{\abs{y}\le 2L}\sum_{E\in\sigma(H)\cap J}\abs{\psi_E(y)}^2\ge \sum_{\abs{y}\le 2L}\sum_{\substack{E\in\sigma(H):\\ \abs{x_E}\le L}}\abs{\psi_E(y)}^2=\sum_{\substack{E\in\sigma(H)\cap J:\\ \abs{x_E}\le  L}}\pa{1-\sum_{\abs{y}> 2L}\abs{\psi_E(y)}^2}\\ \ge \#\set{E\in\sigma(H)\cap J:\ \abs{x_E}\le L}\pa{1-\frac{C_\nu} \theta\e^{-\frac\nu2 {L}}}\ge \frac12\#\set{E\in\sigma(H)\cap J:\ \abs{x_E}\le L}
\end{multline}
for $L\ge K_\nu\abs{\ln\theta}$.

This bound together with \eqref{eq:loceig'} imply that for $L\ge K_\nu\abs{\ln\theta}$ we have
\begin{multline}\label{eq:uppb}
\sum_{\abs{y}\le L}\sum_{\substack{E\in\sigma(H)\cap J:\\ \abs{x_E}> 3L}}\abs{\psi_E(y)}^2\le \sum_{k=4}^\infty\#\set{E\in\sigma(H)\cap J:\ \abs{x_E}\le kL}\frac{C_\nu} \theta\e^{-\frac{\nu kL}2}\\ \le \frac{9C_\nu}\theta L  \sum_{k=4}^\infty k\e^{-\frac{\nu kL}2}<\frac c2
\end{multline}
for  $L\ge M_\nu\abs{\ln\theta}$ with some $M_\nu>0$.

Using this estimate, we get
\[
c\le\sum_{E\in\sigma(H)\cap J}\abs{\psi_E(0)}^2\le \sum_{\substack{E\in\sigma(H)\cap J:\\ \abs{x_E}\le 3L}}\abs{\psi_E(0)}^2+\frac c2,
\]
for  $L\ge M_\nu\abs{\ln\theta}$, so
\[
\frac c2\le \sum_{\substack{E\in\sigma(H)\cap J:\\ \abs{x_E}\le 3L}}\abs{\psi_E(0)}^2,
\]
and since $\#\set{E\in\sigma(H):\ \abs{x_E}\le 3L}\le 13L$ by \eqref{eq:upperb}, we deduce that there exists $C_\nu>0$ and $E\in\sigma(H)\cap J$ such that 
\[ \abs{\psi_E(0)}^2\ge \frac c{26L}=\frac{-C_\nu}{\ln \theta},\quad \abs{x_E}\le \frac{-\ln \theta}{C_\nu}.\]

\end{proof}

Let $H$ be  a self-adjoint operator. Here we will often use 
the integral representation
\be\label{eq:Rder}P_{[E_1,E_2]}(H)=-\frac 1{2\pi}\int_{-\infty}^\infty \sum_{j=1}^2(-1)^j\pa{H-ix-E_j}^{-1}dx,\ee
which holds provided that $E_1,E_2$ are not in the spectrum $\sigma(H)$. 
If in addition $H(s)$ is a differentiable family of operators,  the formula 
\be\label{eq:Rdera}\frac{d}{ds}\pa{H(s)-ix-E_j}^{-1}=-\pa{H(s)-ix-E_j}^{-1}\dot H(s)\pa{H(s)-ix-E_j}^{-1}\ee
holds. Furthermore, for any operator $R$, we have 
\be \label{eq: commutator derivative}
[R, \frac{1}{H-z}]= -\frac{1}{H-z}[R,H] \frac{1}{H-z}.
\ee
\begin{lemma}\label{lem:PbarP}
Let $H_1,H_2,R$ be bounded operators on  $\ell^2\pa{\Lambda}$, with $H_1,H_2$ self-adjoint.
 Let $J=[E_1, E_2]$ and denote by $J_{2\Delta}$ for $\Delta>0$, the widened interval $J+[-2\Delta,2\Delta]$.  Suppose that for some  $\epsilon_1,\epsilon_2$,
\begin{enumerate}
\item $\norm{\pa{H_1-H_2}R}=\epsilon_1$
\item $\norm{[H_2,R]P_J(H_2)}\le\epsilon_2$.
\end{enumerate}
Then
\[\norm{\bar P_{J_\Delta}(H_1)RP_J(H_2)}\le\frac{\epsilon_1+\epsilon_2}{\Delta}.\]
\end{lemma}
\begin{proof}
Let $z_1=E_1-\Delta +ix$, $z_2=E_2+\Delta +ix$ and write
\[G_{i,j}=\pa{H_i-z_j}^{-1}.\]
We first establish the identity
\begin{eqnarray*}\bar P_{J_\Delta}(H_1)RP_J(H_2)&=&\frac 1{2\pi}\sum_{j=1}^2(-1)^j \int_{-\infty}^\infty \bar P_{J_\Delta}(H_1)G_{1,j}[H_2,R]G_{2,j}P_J(H_2)dx\\&+&
\frac 1{2\pi}\sum_{j=1}^2(-1)^j\int_{-\infty}^\infty \bar P_{J_\Delta}(H_1)G_{1,j}\pa{H_2-H_1}R\,G_{2,j}P_J(H_2)dx.
\end{eqnarray*}
Indeed, we start from
$$ G_{1,j}[H_2,R]  G_{2,j} =  G_{1,j}(H_2-H_1)R  G_{2,j} +R  G_{2,j} + G_{1,j}R.
$$
Upon multiplying with $(-1)^j$, summing over $j=1,2$, integrating over $x$, and using \eqref{eq:Rder} with $[E_1,E_2]$ replaced by  $[E_1-\Delta,E_2+\Delta]$, we get the desired identity.  
We next bound
\[\max_{j=1,2}\norm{\bar P_{J_\Delta}(H_1)G_{1,j}}\le \frac1{\sqrt{x^2+\Delta^2}}, \quad \max_{j=1,2}\norm{G_{2,j}P_J(H_2)}\le \frac1{\sqrt{x^2+\Delta^2}}\]
to get
\[\norm{\bar P_{J_\Delta}(H_1)RP_J(H_2)}\le\pa{\epsilon_1+\epsilon_2}\frac 1{\pi}\int_{-\infty}^\infty \frac{dx}{{x^2+\Delta^2}}=\frac{\epsilon_1+\epsilon_2}{\Delta}.\]
\end{proof}

For the next  lemma, we will use the notation $J_a(\mu)=[\mu-a,\mu+a]$, and will let $P^\Theta_{J_a(\mu)}$ denote the spectral projection of $H_o^\Theta$ onto $J_a(\mu)$.
\begin{lemma}\label{outbad} Let  $\Phi$ and $\Theta$, with $\Phi\subset \Theta$, be finite subsets of $\Z^d $.  Let  $(\phi,\mu)$ be an eigenpair for $H_o^{\Phi}$.  Then we have
\be
\dist\pa{\mu,\sigma(H_o^\Theta)}\le C\abs{\partial_{r}\Phi}\norm{\chi_{\partial_r\Phi}\phi}_\infty,
\ee
and
\be
\dist\pa{\phi,Ran\pa{P^\Theta_{J_a(\mu)}}}\le \frac Ca\abs{\partial_r\Phi}\norm{\chi_{\partial_r\Phi}\phi}_\infty.
\ee

Conversely, if $(\psi,\lambda)$ is an eigenpair for $H^{\Theta}$, then
\be\label{eq:distmu}
\dist\pa{\lambda,\sigma(H_o^\Phi)}\le C \abs{ \Theta\setminus\Phi}
\norm{\chi_{ \Theta\setminus\Phi}\psi}_\infty
\ee
and
\be\label{eq:rangmu}
\dist\pa{\phi,Ran\pa{P^\Phi_{J_a(\lambda)}}}\le \frac Ca \abs{ \Theta\setminus\Phi}
\norm{\chi_{ \Theta\setminus\Phi}\psi}_\infty.
\ee
\end{lemma}

\begin{proof} 
We have
\be
\pa{\pa{H_o^\Theta-\mu}\phi}(y)=
\begin{cases} 
\sum_{\substack{y'\in\Phi:\\ \abs{y-y'}\le r}}H_o(y,y')\phi(y') & \mbox{if } y\in\Theta\setminus \Phi \mbox{ and } \dist\pa{y, \Phi}\le r, \\  0 & \mbox{otherwise}. 
\end{cases}
\ee
It follows that
\be
\norm{\pa{H_o^\Theta-\mu}\phi}\le C\abs{\partial_{r}\Phi}\norm{\chi_{\partial_r\Phi}\phi}_\infty.
\ee
Thus, recalling that $\phi$ is normalized,
\be
\dist\pa{\mu,\sigma(H_o^\Theta)}\le \norm{\pa{H_o^\Theta-\mu}\phi}\le C\abs{\partial_{r}\Phi}\norm{\chi_{\partial_r\Phi}\phi}_\infty.
\ee
On the other hand, we have
\be
\norm{\bar P^\Theta_{J_a(\mu)}\phi}\le \norm{{\bar P^\Theta_{J_a(\mu)}\pa{H_o^\Theta-\mu}^{-1}}}\norm{\pa{H_o^\Theta-\mu}\phi}\le \frac Ca
\norm{\chi_{ \Theta\setminus\Phi}\psi}_\infty,
\ee
from which the second assertion of the lemma follows.

Similar considerations yield
\be
\norm{\pa{H_o^\Phi-\lambda}\phi}\le C \abs{ \Theta\setminus\Phi}
\norm{\chi_{ \Theta\setminus\Phi}\phi}_\infty,
\ee
which in turn imply the bounds \eqref{eq:distmu}--\eqref{eq:rangmu}.
\end{proof}

In this paper we are interested in the evolution of the initial wave packet $\psi_o$ supported near some $x\in\Z^d$ up to the (rescaled) time $s$ of order $1$. In this context, we can always approximate the dynamics generated by  $H(s)$ with the one generated by $\hat H^{\mathbb {T}}(s)$, where $H^{\mathbb T}(s)$ is understood as an operator on $\ell^2(\Z^d)$ (extending it by zero outside of the box $\Lambda_L$), in the following sense.
\begin{prop}[The finite speed of propagation bound]\label{thm:FSP}
Let $\T$ be a torus of linear size $R$ and let  $U_\epsilon(s)$,  $ U_\epsilon^{\mathbb{T}}(s)$ be the dynamics generated by $H(s)$ and $H^{\mathbb{T}}(s)$, respectively, i.e.,
\bea
i\epsilon \partial_sU_\epsilon(s)&=&H(s)U_\epsilon(s), \quad U_\epsilon(0)=1;\\
i\epsilon \partial_s U^{\mathbb{T}}_\epsilon(s)&=&H^{\mathbb{T}}(s) U^{\mathbb{T}}_\epsilon(s), \quad U^{\mathbb{T}}_\epsilon(0)=1.
\eea
Then there exists $c>0$ such that for any $\caL$ satisfying $\mathcal{L}\ge C/\epsilon$ we have 
\be
\max_s{ \abs{{(U^\sharp_\epsilon(s))}(y,x)}\le \e^{-c\abs{x-y}}, \quad  \mbox{ for } \abs{x-y}\ge \frac \caL 4},
\ee
where $U_\epsilon^\sharp$ is either $U$ or $U^\T$.
\end{prop}
\begin{proof}
This is a standard fact for  (local) lattice Hamiltonians, see e.g.,  the proof of  \cite[Lemma 5]{EGS} for the time-independent case (which extends to the time-dependent one without effort), or, for a more general approach, \cite{LR}.
\end{proof}
\section*{Acknowledgment} We are grateful to Gian Michele Graf for helpful discussions. We would also like to thank the referees for providing insightful comments and suggestions that helped to improve this paper. 

\section*{Declarations}
\begin{itemize}
\item Competing interests: We declare no competing interests.
\end{itemize}

\printbibliography
\end{document}